\documentclass[pra,aps,reprint,superscriptaddress,preprint numbers,floatfix,nofootinbib,longbibliography]{revtex4-2}
\usepackage{wrapfig, blindtext}
\usepackage[most]{tcolorbox}
\tcbset{colback=yellow!10!white, colframe=red!50!black, 
  highlight math style= {enhanced, %<-- needed for the ?remember? options
    colframe=red,colback=red!10!white,boxsep=0pt}
}
\usepackage{amsmath}
\usepackage{tikz-cd}
\usepackage{empheq}
\usepackage{graphicx}% http://ctan.org/pkg/graphicx

\usepackage{amsthm,bbm}
\usepackage{thm-restate}
\usepackage{anyfontsize}

\usepackage[
colorlinks=true,
linkcolor=blue,
citecolor=blue,
filecolor=blue,
urlcolor=blue,
]{hyperref}
\usepackage[capitalise]{cleveref}
\usepackage{nameref}
\usepackage{orcidlink}
\usepackage{xpatch}% http://ctan.org/pkg/xpatch
\makeatletter
\xpatchcmd{\@ssect@ltx}{\@xsect}{\protected@edef\@currentlabelname{#8}\@xsect}{}{}% Patch \<section>*
\xpatchcmd{\@sect@ltx}{\@xsect}{\protected@edef\@currentlabelname{#8}\@xsect}{}{}% Patch \<section>
\makeatother
\usepackage{youngtab}

\usepackage{tikz}
\usetikzlibrary{decorations.pathreplacing}
\usepackage{pgfplots}
\pgfplotsset{compat=1.18} 
\newcommand{\U}{\operatorname{U}}
\newcommand{\SU}{\operatorname{SU}}

\newcommand\rk{\operatorname{rk}}

\newcommand{\e}{\operatorname{e}}
\renewcommand{\d}{\mathrm{d}}
\renewcommand{\i}{\mathrm{i}}

\renewcommand{\P}{\mathbf{P}}

\newcommand{\id}{\mathbb{I}}

\newcommand{\irreps}{\mathrm{Irreps}}

\newcommand{\tr}{\operatorname{Tr}}

\newcommand{\V}{\mathcal{V}}
\newcommand{\W}{\mathcal{W}}
\renewcommand{\S}{\mathbb{S}}

\newcommand{\T}{\mathrm{T}}

\makeatletter 
\renewcommand\onecolumngrid{% 
  \do@columngrid{one}{\@ne}%
  \def\set@footnotewidth{\onecolumngrid}%
  \def\footnoterule{\kern-6pt\hrule width 1.5in\kern6pt}%
}
\makeatother 
\newlength\dlf % Define a new measure, dlf

\usepackage{amsmath,amssymb}
\usepackage[shortlabels]{enumitem}
\usepackage{graphicx}
\usepackage{graphics}
\usepackage{amsmath}

\usepackage{color}
\usepackage{dsfont}
\usepackage{bbm}
\usepackage{stmaryrd}

\usepackage{lipsum}
\usepackage{lmodern}
\usepackage{tcolorbox}

\usepackage{tabularray}
\usepackage{adjustbox}
\UseTblrLibrary{amsmath}
\UseTblrLibrary{booktabs}
\UseTblrLibrary{varwidth}
\UseTblrLibrary{diagbox}

\usepackage{amsfonts}
\usepackage{graphicx,graphics,epsfig,bm,bbm,amssymb,amsmath,amsfonts,mathrsfs}
\usepackage[normalem]{ulem}
\usepackage{subfigure}
\usepackage{dsfont}
\usepackage{upgreek}
\usepackage{tikz}
\usepackage{natbib}
\usepackage{chngcntr}

\usepackage{pifont}

\theoremstyle{definition}

\newtheorem*{definition*}{Definition}

\theoremstyle{plain}
\newtheorem{theorem}{Theorem}
\newtheorem*{theorem*}{Theorem}

\newtheorem{corollary}{Corollary}
\newtheorem*{corollary*}{Corollary}

\newtheorem{lemma}{Lemma}
\newtheorem*{lemma*}{Lemma}

\newtheorem{proposition}{Proposition}
\newtheorem*{proposition*}{Proposition}

\theoremstyle{remark}

\newtheorem{remark*}{Remark}
\newtheorem{fact}{Fact}
\newtheorem*{fact*}{Fact}

\newcommand{\bes} {\begin{subequations}}
  \newcommand{\ees} {\end{subequations}}
\newcommand{\bea} {\begin{eqnarray}}
  \newcommand{\eea} {\end{eqnarray}}
\newcommand{\be} {\begin{equation}}
  \newcommand{\ee} {\end{equation}}

\def\>{\rangle}
\def\<{\langle}
\def\Tr{\operatorname{Tr}}

\newcommand{\ident}{\mathbb{I}}
\newcommand{\ignore}[1]{}

\usepackage{amsmath,amssymb}
\usepackage{qcircuit}

\newcommand{\real}{\mathbb{R}}
\newcommand{\complex}{\mathbb{C}}
\newcommand{\hilbert}[1][H]{\mathcal{#1}}

\usepackage{mathtools}
\usepackage{microtype}

\newcommand{\diff}{\mathop{}\!\mathrm{d}}

\newcommand\givensymbol[1][]{%
  \nonscript#1\vert
  \allowbreak
  \nonscript
  \mathopen{}
}
\newcommand\ngivensymbol[1][]{%
  \nonscript:
  \allowbreak
  \nonscript
  \mathopen{}
}

\DeclarePairedDelimiterX{\parens}[1]{\lparen}{\rparen}{%
 
 #1
}
\DeclarePairedDelimiterX{\bracks}[1]{\lbrack}{\rbrack}{%
 
 #1
}
\DeclarePairedDelimiterX{\braces}[1]{\lbrace}{\rbrace}{%
 
 #1
}
\DeclarePairedDelimiterX{\angles}[1]{\langle}{\rangle}{%
 
 #1
}

\DeclarePairedDelimiter{\verts}{\lvert}{\rvert}

\newcommand{\p}{\parens}

\newcommand{\set}{\braces}
\newcommand{\abs}{\verts}

\DeclarePairedDelimiter{\floor}{\lfloor}{\rfloor}

\DeclarePairedDelimiterX{\inn}[2]{\langle}{\rangle}{#1,#2}

\DeclarePairedDelimiterX{\comm}[2]{\lbrack}{\rbrack}{#1,#2}
\DeclarePairedDelimiterX{\anti}[2]{\lbrace}{\rbrace}{#1,#2}

\DeclarePairedDelimiter{\ket}{\lvert}{\rangle}
\DeclarePairedDelimiterX{\qout}[2]{\lvert}{\rvert}{%
 #1\delimsize\rangle\delimsize\langle\mathopen{}#2
}
\DeclarePairedDelimiterX{\qproj}[1]{\lvert}{\rvert}{%
 #1\delimsize\rangle\delimsize\langle\mathopen{}#1
}

\DeclarePairedDelimiterX{\qinn}[2]{\langle}{\rangle}{%
 #1\givensymbol[\delimsize]#2
}
\DeclarePairedDelimiterX{\qamp}[3]{\langle}{\rangle}{%
 #1\givensymbol[\delimsize]#2\givensymbol[\delimsize]#3
}
\DeclarePairedDelimiterX{\qavg}[2]{\langle}{\rangle}{%
 #2\givensymbol[\delimsize]#1\givensymbol[\delimsize]#2
}

\renewcommand\P[1][]{%
 \ifstrempty{#1}{%
 \mathbf{P}
 }{%
 \mathbf{p}_{#1}
 }%
 }

\usepackage{ytableau}
\ytableausetup{smalltableaux,centertableaux,nobaseline}
\usetikzlibrary{tikzmark}

\newcommand{\ydiag}[1]{\vcenter{\hbox{\tiny \ydiagram{#1}}}}

\NewDocumentCommand{\yket}{om}{%
 \IfValueTF{#1}{%
 \ket[#1]{\vcenter{\hbox{\small \begin{ytableau} #2 \end{ytableau}}}}
 }{%
 \ket*{\vcenter{\hbox{\small \begin{ytableau} #2 \end{ytableau}}}}
 }
}
\NewDocumentCommand{\tyket}{om}{%
 \IfValueTF{#1}{%
 \ket[#1]{\vcenter{\hbox{\tiny \begin{ytableau} #2 \end{ytableau}}}}
 }{%
 \ket*{\vcenter{\hbox{\tiny \begin{ytableau} #2 \end{ytableau}}}}
 }
}

\crefname{section}{Sec.}{Secs.}
\crefname{claim}{Claim}{Claims}

\begin{document}

\title{Unitary Designs from Random Symmetric Quantum Circuits}

\author{Hanqing Liu\,\orcidlink{0000-0003-3544-6048}}
\email{hanqing.liu@lanl.gov}
\affiliation{Theoretical Division, Los Alamos National Laboratory, Los Alamos, New Mexico 87545, USA}

\author{Austin Hulse}
\email{austin.hulse@duke.edu}
\affiliation{Duke Quantum Center and Department of Physics, Duke University, Durham, NC 27708, USA}
%\affiliation{Department of Physics, Duke University, Durham, NC 27708, USA}
\author{Iman Marvian}
\email{iman.marvian@duke.edu}
\affiliation{Duke Quantum Center and Department of Physics, Duke University, Durham, NC 27708, USA}\affiliation{Department of Electrical and Computer Engineering, Duke University, Durham, NC 27708, USA}

\begin{abstract}
  In this work, we study distributions of unitaries generated by random quantum circuits containing only symmetry-respecting gates. We develop a unified approach applicable to all symmetry groups and obtain an equation that determines the exact design properties of such distributions. It has been recently shown that the locality of gates imposes various constraints on realizable unitaries, which in general, significantly depend on the symmetry under consideration. These constraints typically include restrictions on the relative phases between sectors with inequivalent irreducible representations of the symmetry. We call a set of symmetric gates semi-universal if they realize all unitaries that respect the symmetry, up to such restrictions. For instance, while 2-qubit gates are semi-universal for $\mathbb{Z}_2$, U(1), and SU(2) symmetries in qubit systems, SU($d$) symmetry with $d\ge 3$ requires 3-qudit gates for semi-universality. Failure of semi-universality precludes the distribution generated by the random circuits from being even a 2-design for the Haar distribution over symmetry-respecting unitaries. On the other hand, when semi-universality holds, under mild conditions, satisfied by U(1) and SU(2) for example, the distribution becomes a $t$-design for $t$ growing polynomially with the number of qudits, where the degree is determined by the locality of gates. More generally, we present a simple linear equation that determines the maximum integer $t_{\max}$ for which the uniform distribution of unitaries generated by the circuits is a $t$-design for all $t\leq t_{\max}$. Notably, for U(1), SU(2) and cyclic groups, we determine the exact value of $t_{\max}$ as a function of the number of qubits and locality of the gates, and for $\SU(d)$, we determine the exact value of $t_{\max}$ for up to $4$-qudit gates.
\end{abstract}

\maketitle

\section{Introduction}

Random quantum circuits \cite{emerson2003pseudo, emerson2005convergence, ambainis2007quantum, dankert2009exact, harrow2009random, brown2010convergence, brandao2016efficient, brandao2016local, Hamma:2012ov, hamma2012quantum} have become a standard tool with a wide range of applications in quantum information science. The primary motivation to study random circuits was their application for generating random states and unitaries, which play a crucial role in many quantum algorithms and protocols 
\cite{divincenzo2002quantum, harrow2004superdense, hayden2004randomizing, ambainis2004small}. In recent years, random quantum circuits have also become a standard framework to model many-body systems with thermalizing, chaotic, and information-scrambling dynamics \cite{Gharibyan_2018,nahum2017quantum,nahum2018operator,von2018operator,maldacena2016bound,shenker2014black,cotler2017chaos}. They also provide insights into quantifying the complexity of unitaries \cite{roberts2017chaos,brandao2019models,Liu_2018,Liu1_2018}, with motivations including the investigation of the role of complexity in quantum gravity \cite{susskind2016computational, brown2016holographic, stanford2014complexity}.

Likewise, symmetric circuits, i.e., those that conserve charges, have prolific applications to quantum computation and control \cite{lidar1998decoherence, Bacon:2000qf, divincenzo2000universal, Zanardi:97c, kempe2002exact, brod2013computational}, quantum thermodynamics and resource theories \cite{janzing2000thermodynamic, FundLimitsNature, brandao2013resource, guryanova2016thermodynamics, lostaglio2015quantumPRX, halpern2016microcanonical, halpern2016beyond, lostaglio2017thermodynamic,gour2008resource, Marvian_thesis, marvian2013theory}, quantum refrence frames \cite{QRF_BRS_07,marvian2008building}, covariant error-correcting codes \cite{faist2020continuous, hayden2021error, woods2020continuous}, and quantum machine learning and variational eigensolvers \cite{meyer2023exploiting, nguyen2022theory, sauvage2022building, zheng2023speeding,barron2021preserving, shkolnikov2021avoiding, gard2020efficientsymmetry, streif2021quantum, wang2020x, barkoutsos2018quantum}. Random symmetric circuits are especially relevant for modeling complex systems with conserved charges 
\cite{khemani2018operator, rakovszky2018diffusive, nakata2020generic, kong2021charge,hearth2023unitary,Li:2023mek,Agrawal2022,McCulloch_2023, majidy2023noncommuting}.

In a series of recent works \cite{Marvian2022Restrict,Marvian2024Rotationally,marvian2022quditcircuit,marvian2024theoryAbelian,hulse2024framework, marvian2024synthesis}, it has been found that locality and symmetry together impose emergent constraints on the realizable unitaries, and these constraints imply that the statistical properties of random local symmetric quantum circuits cannot simulate those of all symmetric unitaries. One way of phrasing these constraints quantitatively is to ask, for which $t$ are the expectations of the moments $V^{\otimes t} \otimes V^{\ast \otimes t}$ equal, where $V$ is a random symmetric unitary on the one hand, and a unitary generated from random local symmetric circuits on the other. When these expectations are equal, we say that random local symmetric circuits are a $t$-design for all symmetric unitaries. 

Interestingly, for $d\ge 3$, it turns out that  2-qudit $\SU(d)$-invariant unitaries do not generate even a $2$-design for all $\SU(d)$-invariant unitaries  \cite{marvian2022quditcircuit}. On the other hand, recent work reveals that 3-qudit $\SU(d)$-invariant unitaries generate a $t$ design for $t$ that grows, at least, quadratically in the number of qudits $n$ \cite{hulse2024framework} (Ref. \cite{Li:2023mek} also shows quadratic scaling, albeit with 4-qudit gates.). This difference between 2-qudit and 3-qudit symmetric unitaries can be understood as a consequence of the \emph{semi-universality} \cite{marvian2024theoryAbelian,hulse2024framework} of 3-qudit symmetric gates: they generate all symmetric unitaries, up to relative phases between subspaces with different $\SU(d)$ charges (i.e., inequivalent irreducible representations of $\SU(d)$). On the other hand, 2-qudit $\SU(d)$-invariant unitaries have further constraints for $d \geq 3$, and are not semi-universal. For general symmetry, failure of semi-universality implies that the distribution generated by random symmetric circuits is not even a $2$-design for the uniform distribution over the group of all symmetric unitaries (This can be shown using the result of \cite{robert2015squares}. See also \cref{App:fail} for a slightly different proof). As another example, in the case of U(1) symmetry, using the fact 2-qubit gates are semi-universal \cite{Marvian2022Restrict}, Ref. \cite{hearth2023unitary} shows that random circuits with 2-qubit gates are $t$-design with $t$ that grows, at least, linearly with the number of qubits.

In this paper, we develop a unified approach that is applicable to any symmetry group $G$ with arbitrary finite-dimensional unitary representation, including (but not restricted to) the so-called on-site representation of symmetry. Using this approach we obtain an equation that determines the highest $t = t_{\max}$ such that $t$-design can be achieved, assuming semi-universality holds. Here, we briefly explain the result in the case of symmetric quantum circuits with the representation of compact group $G$ as $u(g)^{\otimes n} : g\in G$ on $n$ qudits (See \cref{thm:tdesign} and \cref{corf} for more general and accurate statements of results). 

Let $|\irreps_G(n)|$ denote the number of inequivalent irreducible representations of $G$ that appear in this representation. Then, we find that $k$-qudit $G$-invariant unitaries form a $t_{\max}$-design (but not a $(t_{\max} + 1)$-design) with
\begin{align}\label{Eq1}
  t_{\max} &=\frac{1}{2}\min_{{\mathbf{q}}}(\mathbf{m}\cdot {|\mathbf{q}|})-1 \ ,
\end{align}
where the minimization is over all integer vectors
with $|\irreps_G(n)|$ components, satisfying
\be\label{tyq}
{\mathbf{q}}\in \mathbb{Z}^{|\irreps_G(n)|} \text{\ \ s.t.\ \ } \mathbf{q}\neq \mathbf{0}\ \text{\ and\ } \ 
\mathbf{S} {\mathbf{q}}={\mathbf{0}}\ .
\ee 
Here, $|\mathbf{q}|$ denotes the vector obtained by taking element-wise absolute value of $\mathbf{q}$, $\mathbf{m}$ is the vector of positive integers given by the multiplicity of each irrep inside $\irreps_G(n)$, and $\mathbf{S}$ is a $|\irreps_G(k)|\times |\irreps_G(n)|$ matrix with non-negative integer elements, defined in \cref{eq:S-def}. 

In particular, when $G$ is an Abelian group, the matrix element $\mathbf{S}_{\mu,\nu}$ is the multiplicity of irrep $\mu^\ast\otimes \nu$ in the representation of group $G$ on $n-k$ qudits, i.e.,  $u(g)^{\otimes (n-k)}: g\in G$, where $\mu^\ast$ is the dual (complex conjugate) of irrep $\mu$. For example, as discussed in \cref{Sec:U(1)}, for a system of $n$ qubits with $\U(1)$ symmetry represented as $\e^{\i\theta Z}: \theta\in[0,2\pi)$ on each qubit, the irreps can be labeled by the eigenvalues of $\sum_{j=1}^n Z_j$, which are in the form $n-2w$ for integers $w=0,\cdots, n$, or equivalently, by the corresponding Hamming weight $w$. In this case, the multiplicity of the irrep is the binomial coefficient $m_w=\binom{n}{w}$, so the vector $\mathbf{m} = (m_0, m_1, \dots, m_n)^\T$.  Then, as shown in \cref{eq:S-U1}, $\mathbf{S}$ is a $(k+1)\times (n+1)$ matrix with matrix elements 
\begin{align}\label{S:U(1)}
    \mathbf{S}_{l,w}= \binom{n - k}{w - l}\ \ \ : l=0,\cdots, k ; \, w=0, \cdots, n\ ,
\end{align}
which is the dimension of the sector with Hamming weight $w-l$ in a system with $n-k$ qubits. 
 Then, finding $t_{\max}$ amounts to find non-zero vectors $\mathbf{q}$ in $\mathbb{Z}^{n+1}$, which are in the null space of $\mathbf{S}$ and minimize $\mathbf{m}\cdot |\mathbf{q}|$. (See \cref{Sec:U(1)} for further details).

Recall that semi-universality means that the only constraints on the realizable unitaries are on the relative phases between the sectors with inequivalent irreps of the symmetry.  Matrix $\mathbf{S}$ determines exactly these constraints. In particular, each $\mathbf{q}$ satisfying \cref{tyq} determines a symmetry-respecting Hamiltonian $Q$ that cannot be realized with $k$-local $G$-invariant gates (See \cref{thm:tdesign} and discussion below it). At the same time, each $\mathbf{q}$ also determines a set of operators that commute with $V^{\otimes t}$ for any $V$ that can be realized with $k$-local symmetry-respecting gates, but not for general $G$-invariant unitaries $V$ (See \cref{Sec:comm}).

%In this work, we prove the formula for $t_{\max}$ in  \cref{Eq1}  and show that, while finding the minimal $\mathbf{q}$ may be intractable in general, still the above formulation lends itself to determining $t_{\max}$ for numerous examples. 
%Indeed, in none of the examples discussed in this paper, we need to directly calculate the matrix $\mathbf{S}$. Rather, we use equivalent linear conditions that are easier to work with (mathematically corresponding to matrix multiplications of $\mathbf{S}$).

%It is worth noting that the results in \cref{tab:tmax} do not

%require explicitly calculating the matrix $\mathbf{S}$ from above; in practice, there are equivalent linear conditions that can be used (mathematically corresponding to matrix multiplications of $\mathbf{S}$).

\begin{table}
%\resizebox{\columnwidth}{!}{
  \begin{tblr}{colspec={c|c|c|c}}
    \toprule
    $G$ & $k$ & $t_{\max} + 1$ & $n$ \\ 
    \midrule
    $\mathbb{Z}_p$ & $\geq p$ & $\displaystyle \begin{cases} 2^{n - 1} & p \text{ even} \\ \infty & p \text{ odd} \end{cases}$ & $> k$ \\
    \midrule
    $\U(1)$ & $\geq 2$ & $\displaystyle \begin{cases} 2^k \binom{(n-1)/2}{k/2} & k \text{ even} \\ 2^k \binom{n/2}{(k+1)/2} & k \text{ odd} \end{cases}$ & $\geq 2^{\floor{\frac{k}{2}}} \floor{\frac{k + 3}{2}}$ \\
    \midrule
    $\SU(2)$ & $\geq 2$ & $2^{2\floor{\frac{k}{2}}+1} \binom{(n-1)/2}{\floor{k/2}+1}$ & $\geq 2^{\floor{\frac{k}{2}}}(\floor{\frac{k}{2}}+2)$ \\
    \bottomrule
  \end{tblr}%}
  \caption{\textbf{Design properties of circuits with $\mathbb{Z}_p$, $\U(1)$ and $\SU(2)$ symmetries.} For a system with $n$ qubits, we consider the uniform distribution over the group generated by $k$-qubit gates respecting $G = \mathbb{Z}_p$, $\U(1)$ and $\SU(2)$ symmetries, denoted as $\mathcal{V}_{n, k}^G$. The second column determines the minimum $k$ for which semi-universality is achieved (this is proven in \cite{marvian2024theoryAbelian} for $\mathbb{Z}_p$, in \cite{Marvian2022Restrict} for $\U(1)$, and in \cite{Marvian2024Rotationally} for $\SU(2)$). $t_{\max}+1$ is the minimum $t$ such that the uniform distribution over the group $\mathcal{V}_{n, k}^G$ is not a $t$-design for the uniform distribution over the group $\mathcal{V}_{n, n}^G$ of all symmetric unitaries. With the exception of $k=2, 3$ in the case of $\SU(2)$ symmetry, the third column is the exact value of $t_{\max}+1$ provided that the number of qubits $n$ satisfies the lower bound given in the fourth column (for smaller $n$, the third column is an upper bound but not necessarily the exact value of $t_{\max}+1$). For the cases of $k=2, 3$ with $\SU(2)$ symmetry, the third column is still the exact value of $t_{\max}+1$, provided that $n\ge 13$. Note that while in some cases the upper index of the binomial coefficient becomes half-integer,  $t_{\max}+1$ is always an integer (See \cref{Sec:U(1)} for further discussion and definition). }
   \label{tab:tmax}
\end{table}

\begin{table}
%\resizebox{\columnwidth}{!}{
  \begin{tblr}{colspec={c|c|c}}
    \toprule
    $G$ & $k$ & $t_{\max} + 1$ \\ \midrule
    \SetCell[r=5]{c} $\U(1)$ & $2$ & $2(n - 1)$ \\
    & $3$ & $n(n-2)$ \\
    & $4$ & $2(n-1)(n-3)$ \\
    & $5$ & $\frac{2}{3}n(n-2)(n-4)$ \\
    & $6$ & $\frac{4}{3} (n - 1)(n-3)(n-5)$ \\ \midrule
    \SetCell[r=3]{c} $\SU(2)$ & $2, 3$ & $(n - 1) (n - 3)$ \\
    & $4, 5$ & $\frac{2}{3} (n-1)(n-3)(n-5)$ \\
    & $6, 7$ & $\frac{1}{3}(n-1)(n-3)(n-5)(n-7)$ \\ \midrule
    \SetCell[r=2]{c} $\SU(d)$ & $3$ & $(n - 1) (n - 3)$ \\
    & $4$ & $\frac{2}{3} (n - 1) (n - 3)(n - 5)$ \\
    \bottomrule
  \end{tblr}%}
  \caption{\textbf{Explicit expressions of $t_{\max} + 1$ for small $k$ and sufficiently large $n$.} See the caption of \cref{tab:tmax} for further details. Lower bounds on $n$ for which $t_{\max}$ is exact for $\SU(d)$ can be found in \cref{tab:SUd}.}
   \label{tab:tmax-example}
\end{table}

It is worth emphasizing that from the point of view of \cref{Eq1,tyq}, the exact form of matrix $\mathbf{S}$ is not important. Rather, the only relevant property of this matrix is its kernel (null space), which determines the symmetric Hamiltonians that are not realizable. In practice, we often use this freedom and never need to directly use the matrix $\mathbf{S}$, as defined in \cref{eq:S-def}.

Using this result, we completely determine the exact value of $t_{\max}$ in the case of $\U(1)$, $\SU(2)$, and the cyclic group $\mathbb{Z}_p: p\ge 2$ for arbitrary locality $k$ and a sufficiently large number of qubits $n$, as summarized in \cref{tab:tmax}. These results imply that when the number of qubits $n\gg 1$ and $2\le k \le \log n $, for both $\U(1)$ and $\SU(2)$ symmetry, $t_{\max}$ grows polynomially with the number qubits $n$, i.e.,
\begin{align}
    t_{\max} \sim 
    \begin{cases}
        a_k\times  n^{\lfloor\frac{k+1}{2}\rfloor} & \text{\ \ for } \U(1), \\
        b_k\times  n^{\lfloor\frac{k}{2}\rfloor+1} & \text{\ \ for } \SU(2)\ , \\
    \end{cases}
\end{align}
where
\begin{align}
    a_k = \frac{2^{\lfloor\frac{k}{2}\rfloor}}{\lfloor \frac{k+1}{2}\rfloor!}\quad \text{and} \quad b_k = \frac{2^{\lfloor\frac{k}{2}\rfloor}}{(\lfloor\frac{k}{2}\rfloor+1)!}\ ,
\end{align}
and symbol $\sim$ here means the ratio of the two sides goes to one in the limit $n\rightarrow\infty$. Remarkably, we find that for the cases of $\U(1)$ and $\SU(2)$ the dominant behaviors of $t_{\max}$ are very similar, and in fact, identical when $k$ is odd.  We also consider $\SU(d)$ symmetry for arbitrary $d\ge 3$ and determine $t_{\max}$ for $k=3$ and $4$ (recall that semi-universality is achieved with $k$-local $\SU(d)$-invariant gates only  when $k \geq 3$ \cite{hulse2024framework}). \cref{tab:tmax-example} shows this result as well as several examples for small $k$ in the case of $\U(1)$ and $\SU(2)$ symmetries.
 
Here, we highlight one of the new tools developed in this paper with applications extending beyond $t$-designs to areas such as Hamiltonian learning \cite{zhukas2024observation}. Specifically, the optimal integers $\mathbf{q}$ that determine $t_\text{max}$ via \cref{Eq1} have an interesting interpretation, which makes them relevant for these applications, such as detecting $k$-body interactions. 

As an example, we focus on $n$-qubit systems with $\U(1)$ symmetry (See \cref{Sec:U(1)} for further details). Then, the optimal integers $\mathbf{q}$ are the eigenvalues of a set of operators $\{F_k: k=0,\cdots, n\}$ defined in \cref{eq:Fk}, where each operator $F_k$ is the unique operator that satisfies the following properties:
\begin{enumerate}[(i)]
    \item  It is a U(1)-invariant Hermitian operator commuting with all U(1)-invariant Hamiltonians (In other words, it is in the center of the Lie algebra of U(1)-invariant Hamiltonians). 
    \item It is orthogonal to all $(k-1)$-local operators $O$, i.e., $\Tr(F_k O)=0$. This, in particular, implies that $F_k$ is orthogonal to any U(1)-invariant Hamiltonian $H$ that is realizable with $(k-1)$-local $\U(1)$-invariant gates.
    \item Its support is restricted to $k+1$ irreps with the smallest multiplicities (namely, Hamming weights $w$ with either $w \leq \floor{\frac{k}{2}}$ or $w \geq n - \floor{\frac{k-1}{2}}$).
    \item Its eigenvalues are integers with the greatest common divisor being 1. Namely, they are the components of vector $\mathbf{q}$ achieving the minimum in \cref{Eq1}. More precisely, for the Haar measure over the group generated by $k$-qubit gates, $t_\text{max}=\frac{1}{2}\|F_{k+1}\|_1-1$, provided that $n$ is sufficiently large (e.g., $\log n>k$).
\end{enumerate}

\noindent{\textbf{Application to Hamiltonian learning:}}  The above properties, in particular, imply that while operator $F_k$ is orthogonal to $(k-1)$-local operators, it is not orthogonal to $k$-local operator $Z^{\otimes k} \otimes \mathbb{I}^{\otimes (n-k)}$. This, together with property (iii) listed above, makes operators $\{F_k\}$ useful for applications in the context of Hamiltonian learning \cite{degen2017quantum, huang2023learning}. Specifically, it enables the direct detection of $k$-body interactions, such as $Z^{\otimes k}\otimes \mathbb{I}^{\otimes (n-k)}$, under $\U(1)$ symmetry, without requiring prior knowledge of any existing $l$-body interactions with $l < k$, and without the need for full process tomography.  Taking advantage of these properties, a recent experiment has demonstrated the direct detection of 3-body interactions $Z^{\otimes 3}$ \cite{zhukas2024observation}. In particular, this work introduces a method for measuring quantities such as
\begin{align}\label{Eq:delt}
  \Delta_k := -\int_0^T \d t~  \Tr[H(t) F_k]  \pmod {2\pi}\ ,
\end{align}
for $k\ge 1$ for time evolution under U(1)-invariant Hamiltonian $H(t): 0\le t\le T$.\footnote{We note that \cite{zhukas2024observation}  introduces $F_3$ and shows the above properties for this operator. In this work, we extend this to 
the full basis $\{F_k\}$ and prove its aforementioned properties. }.
Thanks to property (iii) listed above, measuring  $\Delta_k$ only requires probing the system in the  $k+1$ sectors with lowest multiplicities, whose dimension grows polynomially with $n$ as $n^{\floor{\frac{k}{2}}}$, rendering this measurement exponentially more efficient than full process tomography. In particular,  \cite{zhukas2024observation} reports the measurement of $\Delta_3$ on 4 qubit systems.\footnote{In particular, for $H(t)=Z^{\otimes l}\otimes \mathbb{I}^{\otimes (n-l)}$, we find $\Delta_k= T\times 2^{k} \delta_{k,l} \pmod {2\pi}$.} \\

\noindent\textbf{Outline:} The rest of this work is organized as follows. In \cref{sec:prelim}, we review some basic definitions of symmetric unitaries, $t$-designs, semi-universality, as well as some previous results that will be useful in this work. In \cref{sec:t-design}, we present our general result on characterizing the design properties of symmetric unitaries, where we provide both an easy-to-determine lower bound on $t_{\max}$, and a set of equations to determine the exact value of $t_{\max}$. Then, in \cref{Sec:Examples}, we use these results to fully determine the exact value of $t_{\max}$ for arbitrary locality $k$ and a sufficiently large number of qubits $n$ in the case of $\U(1)$, $\SU(2)$ and cyclic groups $\mathbb{Z}_p$. We also determine the exact value of $t_{\max}$ for qudit systems with $\SU(d)$ symmetry, for up to $k=4$-local gates. Finally, in \cref{sec:proofs}, we present the proofs of  \cref{prop1,thm:tdesign}, which are stated in \cref{sec:t-design}. Further details on the studied examples are presented in the appendices.

\section{Preliminaries}\label{sec:prelim}

In this section, we review some basic definitions including symmetric unitaries, symmetric quantum circuits, $t$-designs, and semi-universality. Assuming semi-universality holds, the only constraints on the realizable unitaries are the relative phases. We review previous results on characterizing the constraints on these relative phases in the context of symmetric unitaries and symmetric quantum circuits, as well as a no-go theorem for unitary designs. 

\subsection{Basic definitions: symmetric unitaries and \texorpdfstring{$t$}{t}-designs}

Consider a quantum system with a finite-dimensional Hilbert space $\mathcal{H}$. Suppose $U(g): g\in G$ 
is a given unitary representation of a compact group $G$ on $\mathcal{H}$. We are interested in unitary transformations on $\mathcal{H}$ that respect this symmetry, i.e.,
\be
\mathcal{V}^G:= \{V: [V,U(g)]=0\ , \forall g\in G\}
\ .
\ee
We refer to such unitaries as $G$-invariant or symmetric unitaries. $G$-invariant Hamiltonians can be defined similarly. Under the action of this representation the Hilbert space $\mathcal{H}$ decomposes into irreducible representations (irreps) of group $G$, as 
\be
\mathcal{H}\cong\bigoplus_{\lambda\in\irreps_G} (\mathcal{Q}_\lambda\otimes \mathcal{M}_\lambda)\ , 
\ee
where $\irreps_G$ is the set of inequivalent irreducible representations of symmetry $G$ that appear in this representation, 
$\mathcal{Q}_\lambda$ corresponds to irrep $\lambda\in \irreps_G$, and $\mathcal{M}_\lambda$ corresponds to the space of multiplicity of irrep $\lambda$ in this representation (This is known as \emph{isotypic} decomposition). Then with respect to this decomposition, any $G$-invariant unitary $V$ decomposes as
\be\label{sym-unit}
V=\bigoplus_{\lambda\in\irreps_G} (\mathbb{I}_\lambda \otimes v_\lambda)\ ,
\ee
where $v_\lambda\in \U(\mathcal{M}_\lambda)$ for all $\lambda\in\irreps_G$ (we use $\U(\mathcal{M}_\lambda)$ and $\SU(\mathcal{M}_\lambda)$ to denote the unitary group and the special uniatry group on $\mathcal{M}_\lambda$).

Clearly, $\mathcal{V}^G$ is a compact Lie group, and therefore has a unique normalized left and right-invariant Haar (uniform) measure, denoted by $\mu_\mathrm{Haar}$. We say a distribution $\nu$ is a $t$-design for $\mu_\mathrm{Haar}$, if 
its first $t$ moments are identical with the first $t$ moments of $\mu_\mathrm{Haar}$, such that
\begin{align}\label{design1}
  \mathbb{E}_{V\sim\nu }[V^{\otimes t}\otimes {V^\ast}^{\otimes t}]= \mathbb{E}_{V\sim\mu_\mathrm{Haar}}[V^{\otimes t}\otimes {V^\ast}^{\otimes t}]\ ,
\end{align}
where ${V^\ast}$ is the complex conjugate of $V$ with respect to an orthonormal basis. Equivalently, $\nu$ is a $t$-design if
\begin{align}\label{design2}
  \mathbb{E}_{V\sim\nu }[V^{\otimes t} A {V^\dag}^{\otimes t}]=\mathbb{E}_{V\sim\mu_\mathrm{Haar} }[V^{\otimes t} A {V^\dag}^{\otimes t}]\ ,
\end{align}
for any operator $A$ on $\mathcal{H}^{\otimes t}$. 
It is straightforward to show that if this equation holds for $t$, then it also holds for all $t'\le t$. 
When this equation does not hold for all integer $t$, we define $t_{\max}$ to be the maximum value of $t$ for which the equation holds. Otherwise, we define $t_{\max}=\infty$.

Suppose $\nu$ is the uniform distribution over a compact subgroup of $\mathcal{V}^G$, denoted by $\mathcal{W}^G$. Then, the super-operators defined on the left-hand and the right-hand sides of \cref{design2} are the projections to the commutants of $V^{\otimes t}: V\in \mathcal{W}^G$ and $V^{\otimes t}: V\in \mathcal{V}^G$, respectively. Therefore, an equivalent way to phrase the condition in \cref{design2} is to say the two commutants are identical, i.e.,
\be
\mathrm{Comm}\{V^{\otimes t}: V\in \mathcal{W}^G\}=\mathrm{Comm}\{V^{\otimes t}: V\in \mathcal{V}^G\}\ .
\ee

In the context of random quantum circuits, the Haar measure arises as the distribution of unitaries realized by circuits with random gates, in the limit of infinite depth. Formally, this is a consequence of the following general theorem which applies even when the gates are chosen from a symmetry-respecting (possibly discrete)  set.\\

%is that, provided that the gate set satisfies mild conditions, the infinite-depth limit of the distribution of unitaries obtained by sequentially choosing gates at random converges to the uniform Haar distribution. \color{red}In particular,\\ 

\begin{fact}[Theorem 8.2 of \cite{kawada1940probability}]\label{fact:Haar}
    Let $p$ be a probability distribution over the compact Lie group $\mathcal{W}$. Suppose unitaries $W_1, \cdots, W_T \in \mathcal{W}$ are sampled independently according to $p$. Then, in the limit $T \to \infty$ the distribution of their product $W_T \cdots W_1$ converges weakly\footnote{Recall that a sequence of distributions $\nu_n$ is said to converge weakly to $\nu$ if $\lim_{n \to \infty} \mathbb{E}_{V\sim \nu_n}[f(V)] = \mathbb{E}_{V \sim \nu}[f(V)]$ for all bounded continuous functions $f$ \cite{durrett2019probability,billingsley1968convergence}.} to the uniform Haar distribution over group $\mathcal{W}$ provided that (1) the support of $p$ generates a dense subgroup of $\mathcal{W}$, and (2) the support of $p$ is not contained in any coset of any closed  {proper} subgroup of $\mathcal{W}$.\footnote{Ref. \cite{kawada1940probability} proves this more generally for any compact group which is also topologically first-countable. See also \cite{kloss1958}.}
\end{fact}

%(note that conditions (1) and (2) apply equally to discrete and continuous gatesets)

\subsection{Semi-universality}
The group of symmetric unitaries has two important normal subgroups, that will play central roles in the following discussions: First, the 
center of the group $\mathcal{V}^G$, which is the subgroup of relative phases
\be\label{rel}
\mathcal{P}=\big\{\sum_{\lambda\in\irreps_G} \e^{\i\theta_\lambda} \Pi_\lambda : \theta_\lambda\in[0,2\pi)\ \big\} ,
\ee
where $\Pi_\lambda$ is the projector to the subspace $\mathcal{Q}_\lambda\otimes \mathcal{M}_\lambda$ of $\mathcal{H}$ corresponding to irrep $\lambda$. 

The second normal subgroup that will play an important role in the following is the commutator subgroup of $\mathcal{V}^G$, i.e., the subgroup of unitaries generated by $V^\dag_2 V_1^\dag V_2V_1 $ for all $V_1,V_2\in \mathcal{V}^G$, denoted by 
\be
\mathcal{SV}^G := [\mathcal{V}^G, \mathcal{V}^G]=\langle V^\dag_2 V_1^\dag V_2V_1 : V_1,V_2\in\mathcal{V}^G\rangle \ .
\ee
In terms of the decomposition of symmetric unitaries in \cref{sym-unit}, unitaries in $V\in \mathcal{SV}^G$ are those satisfying $v_\lambda\in \SU(\mathcal{M}_\lambda)$. Conversely, for any set of $v_\lambda\in \SU(\mathcal{M}_\lambda)$, there exists $V\in \mathcal{SV}^G$ with decomposition in the form of \cref{sym-unit}.

Following Ref. \cite{marvian2024theoryAbelian}, any subgroup $\mathcal{W}^G\subseteq \mathcal{V}^G$ that contain $\mathcal{SV}^G$, such that
\be
\mathcal{SV}^G \subseteq \mathcal{W}^G \subseteq \mathcal{V}^G 
\ee
will be called semi-universal\footnote{This name refers to universality on the semi-simple part of Lie algebra.}. Note that any symmetric unitary $V\in\mathcal{V}^G$ decomposes as $V=\widetilde{V} P=P\widetilde{V}$, where $P\in\mathcal{P}$ and $\widetilde{V}\in \mathcal{SV}^G$.  This decomposition is unique, up to a unitary in $\mathcal{P}\cap \mathcal{SV}^G$, which is a finite group. Therefore, if the set of realizable unitaries is semi-universal, then all symmetric unitaries are realizable, up to constraints on the relative phases between inequivalent irreps.

An equivalent definition of semi-universality that will be useful later is the following: Subgroup $\mathcal{W}^G\subseteq \mathcal{V}^G$ is called semi-universal, if for any $G$-invariant Hamiltonian $H$ satisfying the constraint 
\be
\forall \lambda\in\irreps_G:\ \ \ \Tr(H\Pi_\lambda)=0\ ,
\ee
we have $\exp(\i H s)\in\mathcal{W}^G$, for all $s\in\mathbb{R}$.  In \cref{App:fail} we present a useful lemma from \cite{hulse2024framework} that characterizes the failure of semi-universality. See also \cite{hulse2024framework} for more discussion and a powerful framework for characterizing semi-universality.

\subsection{Characterizing the constraints on the relative phases}
By definition, if the set of realizable unitaries $\mathcal{W}^G$ is semi-universal, then the only constraints on the realizable unitaries are on the relative phases between sectors with inequivalent irreps. To characterize such constraints, it is useful to consider the projection of the realizable Hamiltonians to the 
linear space spanned by projectors $\{\Pi_\lambda\}$ (From a Lie-algebraic perspective, the linear space spanned by $\{\i \Pi_\lambda\}$ is the center of the Lie algebra associated to the Lie group of symmetric unitaries $\mathcal{V}^G$ \cite{Marvian2024Rotationally, zimboras2015symmetry}).
This projection can be described by the map 
$$O\ \mapsto \ |\chi_O\rangle=\sum_{\lambda\in \irreps_G} \Tr(\Pi_\lambda O) |\lambda\rangle \ ,$$
which determines the projection of $O$ to $\mathrm{span}\{\Pi_\lambda\}_\lambda$ , where $\{|\lambda\rangle: \lambda\in \irreps_G\}$ denotes an orthonormal basis for a space with dimension $|\irreps_G|$. Equivalently, this projection can be obtained from the characteristic function $\chi_O: G\rightarrow \mathbb{C}$, defined by 
\be
\chi_O(g)=\Tr(U(g)O) \ \ \ 
: \forall g\in G\ .
\ee

Then, Ref. \cite{Marvian2022Restrict} shows that
\begin{proposition}[\cite{Marvian2022Restrict}]
  Let $\{H_l\}_l$ be a set of G-invariant Hermitian operators and $\mathcal{W}^G$ be the Lie group generated by one-parameter groups $\{\exp(\i s H_l): s\in\mathbb{R}\}_l$. Then, 
  \be\label{dkd}
  \dim(\mathcal{V}^G)-\dim(\mathcal{W}^G)\ge |\irreps_G|-\dim(\mathcal{S})\ ,
  \ee
  where 
  \be
  \mathcal{S}=\mathrm{span}_\mathbb{R}\{ \sum_{\lambda\in \irreps_G} \Tr(H_l\Pi_\lambda) |\lambda\rangle \}\ .
  \ee
  For any $G$-invariant Hermitian operator $A$, $\mathcal{W}^G$ contains 
  $\exp(\i s A): s\in\mathbb{R}$ only if
  $|\chi_A\rangle \in \mathcal{S}$, or equivalently, only if $\chi_A\in \mathrm{span}_{\mathbb{R}}\{\chi_{H_l}\}$, where $\chi_{H_l}(g)=\Tr(U(g) H_l)$, and $\chi_{A}(g)=\Tr(U(g) A)$, for all $g\in G$. When $\mathcal{W}^G$ is semi-universal, i.e., contains $\mathcal{SV}^G$, then this necessary condition is also sufficient, that is, $\mathcal{W}^G$ also contains $\exp(\i s A): s\in\mathbb{R}$, if $G$-invariant Hermitian operator $A$ satisfies the above condition.
\end{proposition}

\subsection{Symmetric quantum circuits}
An important special case is that of composite systems, such as $n$ qudits with the total Hilbert space $\mathcal{H}\cong (\mathbb{C}^d)^{\otimes n}$. In this context, a particular type of representation of symmetry often arises in physical systems, usually referred to as on-site representations. More precisely, in this case the representation of group $G$ on the joint system 
is 
$$U(g)=u(g)^{\otimes n}\ \ \ : g\in G\ ,$$ 
where $u(g) : g\in G$ is the unitary representation of the symmetry on the individual qubit. Then, we denote the irreps of the group that appear in $U(g)=u(g)^{\otimes n}: g\in G$, as $\irreps_G(n)$.

Recall that a unitary transformation is called $k$-local if it acts non-trivially on, at most, $k$ qudits (subsystems) in the system. Following the terminology of quantum circuits, we sometimes refer to such unitaries as $k$-qudit gates. Then, an important subgroup of $\mathcal{V}^G$, denoted as $\mathcal{V}^G_{n,k}$ is the subgroup that can be decomposed as a sequence of $k$-local $G$-invariant unitaries, as $V=V_T \cdots V_1$, where each $V_j\in\mathcal{V}^G$ is $k$-local, i.e., acts non-trivially on at most $k$ qudits, and $T$ is an arbitrary integer.  With this definition, the group of all symmetric unitaries is $\mathcal{V}^G=\mathcal{V}^G_{n,n}$. 

\subsection*{A no-go theorem for universality and designs}

Is it possible to achieve $\mathcal{V}^G_{n,n}=\mathcal{V}^G_{n,k}$ with $k<n$? 
According to a well-known result in quantum computing, when the representation of symmetry is trivial, this is achieved with $k=2$ \cite{DiVincenzo:95, lloyd1995almost}. However, in the presence of symmetries, the following result of \cite{Marvian2022Restrict} imposes strong constraints on the minimum $k$ that is needed to achieve universality. 

\begin{proposition}[\cite{Marvian2022Restrict}]
  \label{Thm-1} 
  For any integer $k\le n$, the group generated by $k$-local $G$-invariant unitaries on $n$ qudits, denoted by $\mathcal{V}_{n,k}^G$, is a compact connected Lie group. The difference between the dimensions of this Lie group and the group of all $G$-invariant unitaries is lower bounded by
  \be\label{bound1}
  \dim(\mathcal{V}_{n,n}^G)-\dim(\mathcal{V}_{n,k}^G)\ge |\irreps_G(n)|-|\irreps_G(k)|\ ,
  \ee
  where $|\irreps_G(k)|$ is the number of inequivalent irreps of group $G$ appearing in the representation $u(g)^{\otimes k}: g\in G$. If $\mathcal{V}_{n,k}^G$ is semi-universal, i.e., contains $\mathcal{SV}_{n,n}^G=[\mathcal{V}_{n,n}^G, \mathcal{V}_{n,n}^G]$ and $G$ is a connected group then the above bound holds as equality. 
\end{proposition} 

In this paper, we study the implications of this no-go theorem on the distribution of unitaries generated by random $k$-local $G$-invariant gates. More specifically, we are interested in the properties of the Haar distribution over the group generated by such gates, namely 
$\mathcal{V}_{n,k}^G$.    According to \cref{fact:Haar}, in quantum circuits with random $k$-local $G$-invariant gates,  as the depth of the circuit  grows, under mild conditions on the distribution of gates,    the resulting distribution of $n$-qudit unitaries converges to  the Haar distribution over 
$\mathcal{V}_{n,k}^G$. For instance, if qudits are nodes of a connected graph (e.g., a chain) and at each time step we randomly pick $k$ nearest-neighbor qudits and 
apply a random $k$-local $G$-invariant gate chosen from a distribution with full support on all such gates, then the resulting distribution converges to the Haar distribution over $\mathcal{V}_{n,k}^G$.

The compactness of the Lie groups $\mathcal{V}_{n,k}^G$ and $\mathcal{V}_{n,n}^G$ implies that unless they have equal dimensions, the Haar measure over $\mathcal{V}_{n,k}^G$ cannot be a $t$-design for $\mathcal{V}_{n,n}^G$ for arbitrary large $t$.\footnote{\label{fn:phase}It is worth noting that both groups $\mathcal{V}_{n,k}^G$ and $\mathcal{V}_{n,n}^G$ contain the U(1) subgroup corresponding to the global phases $e^{\i\theta}\mathbb{I}:\theta\in[0,2\pi)$. The presence of these global phases does not affect the design properties of $\mathcal{V}^G_{n,k}$. More precisely, they disappear when one considers the group of unitaries $V\otimes V^\ast: V\in \mathcal{V}_{n,k}^G$, which can be denoted as $\mathcal{V}_{n,k}^G/\U(1)$, and is itself a connected compact Lie group. Clearly, removing this U(1) subgroup from both $\mathcal{V}_{n,n}^G$ and $\mathcal{V}_{n,k}^G$ does not change the difference between their dimensions. That is, $\dim(\mathcal{V}_{n,n}^G/\U(1))-\dim(\mathcal{V}_{n,k}^G/\U(1))\ge |\irreps_G(n)|-|\irreps_G(k)|$.} Obviously, the uniform distribution over compact submanifolds with lower dimensions cannot fully mimic the uniform distribution over a manifold with a larger dimension (See \cref{lem:design} and \cref{app:design} for a more general formulation of this fact, using Urysohn's lemma and Stone-Weierstrass theorem). Therefore, an immediate corollary of \cref{Thm-1} 
 is
\begin{corollary}[corollary of \cite{Marvian2022Restrict}]
  For a continuous symmetry group $G$ that is non-trivially represented on $n$ qudits, the number of inequivalent irreducible representations $|\irreps_G(n)|$
  grows unboundedly with $n$. Therefore, there is no fixed $k$ such that $k$-local $G$-invariant unitaries become universal such that 
  $\mathcal{V}_{n,k}^G=\mathcal{V}_{n,n}^G$, for all $n\ge k$. Furthermore, since $\mathcal{V}_{n,k}^G$ is a compact subgroup of $\mathcal{V}_{n,n}^G$ with lower dimension, 
  the Haar measure over $\mathcal{V}_{n,k}^G$ cannot be a $t$-design for 
  $\mathcal{V}_{n,n}^G$ for arbitrary large $t$, i.e., $t_{\max} < \infty$. 
\end{corollary}

%, or at least, over a subset of them that generates a dense subgroup of  $\mathcal{V}_{n,k}^G$.

In the following sections, we will show that despite this no-go theorem, if $k$-local $G$-invariant unitaries are semi-universal, i.e., if $\mathcal{SV}^G_{n,n}\subseteq \mathcal{V}^G_{n,k}$, then the Haar distribution over $\mathcal{V}^G_{n,k}$ can be a $t$-design for 
$\mathcal{V}^G_{n,n}$ with large $t$ that grows with $n$ and $k$. A useful tool for establishing this result, which is also used to prove \cref{Thm-1} is the following lemma from \cite{Marvian2022Restrict}.

\begin{lemma}[\cite{Marvian2022Restrict}]\label{lemma-1}
  For any $G$-invariant Hermitian operator $H$, $\forall s\in\mathbb{R}: \exp(\i H s)\in\mathcal{V}^G_{n,k}$ only if function $\chi_H$, defined by $\chi_H(g)=\Tr(H u(g)^{\otimes n})$, satisfies 
  \be\label{span}
  \chi_H \in \mathrm{span}_\mathbb{R}\{r^{n-k} f_\nu: \nu\in \irreps_G(k)\}\ ,
  \ee
  where $r(g)=\Tr(u(g))$, and $f_\nu$ is the character of irrep $\nu$ of group $G$. Furthermore, if $k$-local $G$-invariant gates are semi-universal, i.e., if $\mathcal{SV}^G_{n,n}\subseteq \mathcal{V}^G_{n,k}$ then the converse also holds. That is, for any $G$-invariant Hermitian operator $H$ satisfying the above constraint, $\forall s\in\mathbb{R}: \exp(i H s)\in\mathcal{V}^G_{n,k}$.
\end{lemma}
An immediate corollary of the above lemma is
\begin{lemma}[\cite{Marvian2022Restrict}]\label{lem1}
  Let $\mathcal{S}_k$ be the space spanned by the projection of $k$-local Hermitian operators to $\{\Pi_\lambda :\lambda\in \irreps_G(n)\}$, i.e., 
  \be
  \mathcal{S}_k=\mathrm{span}_\mathbb{R}\{|\chi_A\rangle: A=A^\dag, \text{$A$ is $k$-local}\}\ .
  \ee
  Then, $\dim(\mathcal{S}_k)\le |\irreps_G(k)|$, and the equality holds if $G$ is a connected group, or if $\Tr(u(g))\neq 0$ for all $g\in G$. 
\end{lemma}

\section{From semi-universality to \texorpdfstring{$t$}{t}-designs}\label{sec:t-design}

In this section, we present the main results of this paper, namely we characterize the design properties of a group of $G$-invariant unitaries, assuming semi-universality holds. In particular, we study \texorpdfstring{$t_{\max}$}{tmax}, the maximum value of $t$ for which the condition in \cref{design1} holds. First, we establish a lower bound on \texorpdfstring{$t_{\max}$}{tmax}, and then, in \cref{thm:tdesign}, we provide an equation that fully determines \texorpdfstring{$t_{\max}$}{tmax}. While the main focus of this paper is on symmetric quantum circuits on qudits with $k$-qudit gates, we present the statement of these results more generally in terms of an arbitrary set of $G$-invariant Hamiltonians $\{H_l\}$ (more precisely, both results are phrased in terms of the projection of these Hamiltonians to the subspace $\{\Pi_\mu\}$). In \cref{Sec:kqudit}, we focus on the case of symmetric quantum circuits. We note that some of the ideas needed to establish our results have been previously applied to find lower bounds on $t_{\max}$ for symmetry groups such as $\U(1)$ and $\SU(d)$ for small $k$ (in particular, see \cite{hearth2023unitary,Li:2023mek}).

\subsection{Failure of semi-universality implies lack of 2 design}

Before discussing the consequences of semi-universality, first, we explain what happens when semi-universality does not hold.  
Suppose a compact group $\mathcal{W}^G\subset \mathcal{V}^G$ is not semi-universal, i.e., does not contain the commutator subgroup $\mathcal{SV}^G=[\mathcal{V}^G,\mathcal{V}^G]$. In this situation what can we say about the statistical properties of the uniform distribution over $\mathcal{W}^G$?

Assume $\mathcal{W}^G$ is a connected Lie group, which is true in the case of $\mathcal{W}^G =\mathcal{V}^G_{n,k}$ generated by $k$-local $G$-invariant gates. Then, as is previously noted in \cite{marvian2022quditcircuit}, failure of semi-universality implies that the uniform distribution over $\mathcal{W}^G$ is not even a 2-design for the uniform distribution over $\mathcal{V}^G$. In particular, 
\begin{proposition}[corollary of \cite{robert2015squares}]\label{prop:con}
  Let $\mathcal{W}^G$ be a compact connected 
  subgroup of $G$-invariant unitaries $\mathcal{V}^G$. If $\mathrm{Comm}\{V\otimes V: V\in \mathcal{V}^G\}=\mathrm{Comm}\{V\otimes V: V\in \mathcal{W}^G\}$, then $\mathcal{W}^G$ is semi-universal, i.e., it contains the $\mathcal{SV}^G=[\mathcal{V}^G,\mathcal{V}^G]$. Therefore, if $\mathcal{W}^G$ is not semi-universal, then the uniform distribution over $\mathcal{W}^G$ is not a 2-design for the uniform distribution over $\mathcal{V}^G$, i.e., $t_{\max} \leq 1$. 
\end{proposition}
As noted in \cite{marvian2022quditcircuit}, this proposition can be shown using 
a remarkable result of \cite{robert2015squares, zimboras2015symmetry}, which builds on the seminal work of Dynkin \cite{dynkin1957maximal}. While such arguments apply more broadly to the case of all compact Lie groups, in \cref{App:fail} we present a simpler proof of \cref{prop:con}, which uses a characterization of semi-universality in \cite{hulse2024framework}.\footnote{It is worth noting that the above proposition does not hold if $\mathcal{W}^G$ is not connected. For instance, Clifford unitaries, which are a discrete subgroup of the unitary group, generate a 3-design. }
As we see in an example in \cref{Sec:ex:last}, the converse of this proposition is not generally valid. That is, while semi-universality of $\mathcal{W}^G$ is a necessary condition,  
on its own is not sufficient to guarantee that the uniform distribution over $\mathcal{W}^G$ is a 2-design for $\mathcal{V}^G$. However, as it follows from \cref{cor6},  the converse holds under some further assumptions about the multiplicities.

It is also worth mentioning that in the case of quantum circuits with Abelian symmetries, 
failure of semi-universality implies a stronger constraint. In particular, \cite{marvian2024theoryAbelian} shows that if $k$-local gates are not semi-universal, then they can not generate a 1-design. 
\begin{proposition}[\cite{marvian2024theoryAbelian}] %\label{prop:Abelian}
  For an Abelian group $G$, $\mathrm{Comm}(\mathcal{V}^G_{n,k})=\mathrm{Comm}(\mathcal{V}^G_{n,n})$ if, and only if, $k$-local $G$-invariant gates are semi-universal. Therefore, if semi-universality does not hold, then the uniform distribution over $\mathcal{V}^G_{n,k}$ is not a 1-design for the uniform distribution over the group of all $G$-invariant unitaries 
  $\mathcal{V}^G_{n,n}$.
\end{proposition}

To understand the failure of semi-universality and the above proposition better, it is useful to recall the classification of restrictions on universality and semi-universality in \cite{marvian2024theoryAbelian} (See \cref{lem:formal} in \cref{App:fail}, and \cite{hulse2024framework} for more discussion on the failure of semi-universality). In this reference, type $\mathbf{I}$ constraints refer to restrictions on the relative phases, i.e., 
when $\mathcal{W}^G$ does not contain the subgroup $\mathcal{P}$ defined in \cref{rel}. On the other hand, type $\mathbf{II}-\mathbf{IV}$ constraints are restrictions on the commutator subgroup of $\mathcal{W}^G$. Equivalently, these constraints describe the situation in which $\mathcal{W}^G$ does not contain $\mathcal{SV}^G$. 

In particular, type $\mathbf{II}$ constraints refer to the case where there exists an irrep $\lambda_\ast\in\irreps_G$ such that the action of $\mathcal{W}^G$ in $\mathcal{M}_{\lambda_\ast}$ is reducible. In other words, type $\mathbf{II}$ constraints exist if, and only if 
\be
\mathrm{Comm}\{V: V\in \mathcal{W}^G\}\neq \mathrm{Comm}\{V: V\in \mathcal{V}^G\}\ .
\ee
Therefore, in this case, the uniform distribution over $\mathcal{W}^G$ is not a 1-design for the uniform distribution over $\mathcal{V}^G$.

Type $\mathbf{III}$ constraints, on the other hand, refer to the case where there exists $\lambda_\ast\in\irreps_G$ such that the action of $\mathcal{W}^G$ on $\mathcal{M}_{\lambda_\ast}$ is irreducible, but it does not contain the full $\SU(\mathcal{M}_{\lambda_\ast})$. 
If $\mathcal{W}^G$ is connected, which is the case of $\mathcal{W}^G=\mathcal{V}^G_{n,k}$, then 
\be\label{double-comm}
\mathrm{Comm}\{V\otimes V: V\in \mathcal{W}^G\}\neq \mathrm{Comm}\{V\otimes V: V\in \mathcal{V}^G\}\ .
\ee
This follows from the fact that a compact connected group of unitaries $\mathcal{T}\subseteq \U(d) $ contains $\SU(d)$ if, and only if, $\mathrm{Comm}\{U\otimes U: U\in \mathcal{T}\}= \mathrm{Comm}\{U\otimes U: U\in \SU(d)\}$ \cite{dynkin1957maximal, zimboras2015symmetry, robert2015squares}
\footnote{This is sometimes called lack of quadratic symmetries \cite{zimboras2015symmetry, robert2015squares}}. We conclude that in the presence of type $\mathbf{III}$ constraints, assuming $\mathcal{W}^G$ is connected, then the uniform distribution over $\mathcal{W}^G$ is not a 2-design for the uniform distribution over $\mathcal{V}^G$.

Finally, according to characterization in \cite{marvian2024theoryAbelian} type $\mathbf{IV}$ constraints refer to the case in which the realized unitaries in a sector $\mathcal{M}_\lambda$ dictate the unitaries in one or multiple other sectors. In other words, in general, the time evolution of different sectors is not independent of each other. As we further explain in \cref{App:fail}, in the absence of types $\mathbf{II}$ and $\mathbf{III}$ constraints, type $\mathbf{IV}$ constraints happen only if there exist distinct irreps $\lambda_1,\lambda_2\in\irreps_G$ for which $\dim(\mathcal{M}_{\lambda_1})=\dim(\mathcal{M}_{\lambda_2})$, and the unitary realized in $\mathcal{M}_{\lambda_2}$ (up to a phase and change of basis) is equal to the unitary realized in $\mathcal{M}_{\lambda_1}$, or its complex conjugate. In both cases, \cref{double-comm} holds as an inequality.

It is worth noting that both types $\mathbf{III}$ and $\mathbf{IV}$ constraints exist in the case of circuits with 2-qudit $\SU(d)$-invariant gates for $d\ge 3$ \cite{marvian2022quditcircuit}. On the other hand, Ref. \cite{marvian2024theoryAbelian} shows that while types $\mathbf{II}$ can appear in the case of Abelian symmetries, types $\mathbf{III}$ and $\mathbf{IV}$ do not exist in quantum circuits with such symmetries.

\subsection{Lower bound on \texorpdfstring{$t_{\max}$}{tmax}}

In the rest of this paper, we assume semi-universality holds. When the set of realizable unitaries $\mathcal{W}^G$ is semi-universal, i.e., contains $\mathcal{SV}^G$, then the statistical properties of 
the Haar distribution over $\mathcal{W}^G$ is determined by the constraints on the relative phases of the realizable unitaries, which is characterized by the subgroup $\mathcal{W}^G \cap\mathcal{P}$.

To understand how the constraints on the relative phases affect the statistical properties of random $G$-invariant unitaries, it is useful to consider a subset of irreps $\Delta \subseteq \irreps_G$ for which there are no constraints on the relative phases, such that the projections of the groups $\mathcal{W}^G$ and $\mathcal{V}^G$ to the subspace $\bigoplus_{\lambda\in\Delta} (\mathcal{Q}_\lambda\otimes \mathcal{M}_\lambda)$ are equal. In other words, for any $V \in \mathcal{V}^G$, there exists $W \in \mathcal{W}^G$ such that $V\sum_{\lambda\in\Delta} \Pi_\lambda=W\sum_{\lambda\in\Delta} \Pi_\lambda$. 

Assuming $\mathcal{W}^G$ is generated by 1-parameter groups 
$\{\exp(i s H_l): s\in\mathbb{R}\}$ for $l=1,\cdots, N$, and it contains $\mathcal{SV}^G$, then this condition is satisfied if, and only if, 
\be\label{cond}
\dim(\mathrm{span}_\mathbb{R}\{\sum_{\mu\in\Delta} \Tr(H_l\Pi_\mu) |\mu\rangle \}_l)=|\Delta|\ ,
\ee
where $|\Delta|$ is the number of inequivalent irreps in $\Delta$. Equivalently, 
this condition can be phrased in terms of the matrix 
\be\label{matrix}
\mathbf{M}_{l,\mu}=m_\mu \frac{\Tr(\Pi_\mu H_l)}{\Tr(\Pi_\mu)}\ : l=1,\cdots, N ; \mu\in \irreps_G\ ,
\ee
where 
\be
m_\mu=\dim(\mathcal{M}_\mu)
\ee
is the multiplicity of irrep $\mu\in\irreps_G$ in the representation $U(g): g\in G$. Let $\mathbf{M}^{\Delta}$ be $N\times |\Delta|$ submatrix obtained from this matrix by picking columns that correspond to irreps $\mu\in\Delta$. Then, the condition in \cref{cond} is equivalent to the condition that 
\be\label{cond9}
\rk\big(\mathbf{M}^{\Delta})=|\Delta|\ ,
\ee
where $\rk(\cdot)$ denotes the rank of the matrix.

Then, this condition together with the semi-universality assumption, implies the following lower bound on $t_{\max}$ (See the definition below \cref{design2}).

\begin{proposition}\label{prop1}
  Suppose $\mathcal{W}^G$ is a compact Lie group generated by 1-parameter groups 
  $\{\exp(\i s H_l): s\in\mathbb{R}\}_l$, where $\{H_l\}_l$ are $G$-invariant Hermitian operators. If $\mathcal{W}^G$ is semi-universal, i.e., $\mathcal{W}^G \supseteq \mathcal{SV}^G$, then $t_{\max}$ for the Haar measure over $\mathcal{W}^G$ satisfies
  \be
  t_{\max}  \ge \min\big\{ \dim(\mathcal{M}_\lambda): \lambda\in\irreps_G- \Delta \big\}-1\ ,
  \ee
  where $\Delta \subseteq \irreps_G$ is any subset of irreps 
  satisfying \cref{cond}, or equivalently, \cref{cond9}.
\end{proposition}
We prove this proposition in \cref{Sec:proof1} and present various examples of its applications in \cref{Sec:Examples}. Note that $\Delta$ can always be chosen to contain, at least, one irrep, i.e., $|\Delta|\geq 1$. This is because, without loss of generality, we can assume that the global phases $\e^{\i \theta} \ident$ are contained in $\mathcal{W}^G$, since they do not affect the $t$-design properties of the group (see also the discussion in \cref{fn:phase}).  Then, the above proposition implies that

\begin{corollary}\label{cor6}
Assume semi-universality holds and 
let $m_0\leq m_1$ be the two minimum multiplicities. Then,  $t_{\max}\ge m_1-1$. 
\end{corollary}
For example, as discussed in \cref{Sec:SU(2)}, in the case of SU(2) symmetry, this result together with the semi-universality of 2-qubit  gates established in \cite{Marvian2024Rotationally}, implies that for the distribution of unitaries generated by random 2-qubit gates,  $t_{\max}\ge n-2$, i.e., grows at least linearly with the system size (The actual value of $t_{\max}$ in this case is $(n-1)(n-3)-1$).

It is also worth noting that for the group $\mathcal{V}^G_{n,k}$ in the context of symmetric quantum circuits, as discussed in \cref{lem1}, the result of \cite{Marvian2022Restrict} implies that 
\be
\dim(\mathrm{span}_\mathbb{R}\{\sum_{\mu\in\Delta} \Tr(H_l\Pi_\mu) |\mu\rangle \}_l)\le |\irreps_G(k)|\ . 
\ee
Therefore, to satisfy the condition in \cref{cond}, the number of irreps inside $\Delta$ can be, at most, $|\irreps_G(k)|$.

\subsection{Full characterization}

While \cref{prop1} establishes useful lower bounds on $t_{\max}$, it cannot determine the actual value of this quantity. Next, we present equations that achieve this (See \cref{Sec:comm} and \cref{Sec:design} for the proof). 
\begin{theorem}
  \label{thm:tdesign}
  Let $\mathcal{W}^G$ be the Lie group generated by the 1-parameter families $\{\exp(-\i H_l t): t\in\mathbb{R}\}_l$, where $H_l$ are $G$-invariant Hermitian operators. Assume $\mathcal{W}^G$ is compact and semi-universal, i.e., $\mathcal{W}^G \supseteq \mathcal{SV}^G$. Then, $t_{\max}$ for the uniform distribution over $\mathcal{W}^G$ is 
  \begin{align}\label{eq:tmax}
  t_{\max} & = \frac{1}{2} \min_{Q} \|Q\|_1-1 = \frac{1}{2} \min_{q} \hspace{-3mm} \sum_{\lambda\in \irreps_G} \hspace{-3mm} |q(\lambda)| m_\lambda-1 ,
  \end{align}
  where
  \begin{align}\label{QQ}
    Q=\sum_{\lambda \in \irreps_G} q(\lambda) m_\lambda \frac{\Pi_\lambda}{\tr \Pi_\lambda}\ ,
  \end{align} 
  $q(\lambda): \lambda\in \irreps_G$ is a set of integers, and the minimization is over all operators $Q$, or equivalently integers $q(\lambda)$, satisfying the following conditions :
  \begin{subequations}\label{eq:cond}
  \begin{align}
      \|Q\|_1 &= \sum_{\lambda\in \irreps_G} |q(\lambda)| m_\lambda \neq 0, \label{cond1} \\
      \Tr(Q) &= \sum_{\lambda\in \irreps_G} q(\lambda) m_\lambda = 0\ , \label{cond0} \\
      \Tr(H_l Q) &= \sum_{\lambda \in \irreps_G} \hspace{-3mm} q(\lambda) m_\lambda \frac{\Tr(H_l \Pi_\lambda)}{\Tr(\Pi_\lambda)} = 0, \label{cond2}
  \end{align}
  \end{subequations}
  for $l=1,\cdots, N$. On the other hand, if there does not exist any non-zero operator $Q$ satisfying \cref{eq:cond}, i.e., $\Tr(Q)=\Tr(H_l Q)=0$ for $l=1,\cdots, N$, then $\mathcal{W}^G$ contains all $G$-invariant unitaries, up to a possible global phase, which means $t_{\max} = \infty$.    
\end{theorem}
Note that \cref{cond0} implies that $\|Q\|_1=\sum_{\lambda\in \irreps_G} |q(\lambda)| m_\lambda$ is an even integer \footnote{Recall that for any operator $A$, $\|A\|_1=\Tr(\sqrt{A A^\dag})$, or, equivalently, the sum of the singular values of the operator.}, which means the right-hand side of \cref{eq:tmax} is an integer, consistent with the fact that $t_{\max}$ is an integer. \Cref{cond2} means that operator $Q$ is a $G$-invariant Hamiltonian that is orthogonal to all the Hamiltonians $\{H_l\}$. This implies that the family of $G$-invariant unitaries $\exp(\i s Q): s\in\mathbb{R}$, does not belong to $\mathcal{W}^G$ (except for some special points such as $s=0$). Equivalently, $Q$ does not belong to the Lie algebra generated by $\{H_l\}$.\footnote{This follows from the fact that $Q$ is orthogonal to all the Hamiltonians $\{H_l\}$ together with the fact that it also commutes with them. The latter implies that 
for any pair of $G$-invariant operators $B_1, B_2$, it holds that  $\Tr([B_1,B_2]Q)=\Tr(B_1[B_2,Q])=0$. It follows that $Q$ is orthogonal to all elements of the Lie algebra \cite{zimboras2015symmetry, Marvian2022Restrict}.}   In \cref{Sec:comm} we discuss another interpretation of integers $q(\lambda)$ and operator $Q$.   We also note that \cref{cond0} follows from \cref{cond2} if the identity operator is in the span of the $H_l$.

Equivalently, we can phrase conditions in \cref{thm:tdesign},  in terms of matrix $\mathbf{M}$ defined in \cref{matrix}, a vector of integers 
\be
\mathbf{q}=(q(\lambda_1),\cdots, )^\T\in \mathbb{Z}^{|\irreps_G|}\ , 
\ee
and the vector of (non-negative integer) multiplicities 
\be
\mathbf{m}=(m_{\lambda_1}, \cdots, )^\T\in \mathbb{N}_0^{|\irreps_G|} \ . 
\ee
Then, the conditions in \cref{eq:cond} can be rewritten as
\bes\label{ue}
\begin{align}
  0<\mathbf{m} \cdot \mathbf{|q|} &\le 2t\\ 
  \mathbf{m} \cdot \mathbf{q} &={0}\\ 
  \mathbf{M} \mathbf{q}&=\mathbf{0}\ , \label{ds2}
\end{align}
\ees
where $|\mathbf{q}|=(|q(\lambda_1)|,\cdots, )^\T$ is the vector with element-wise absolute value of $\mathbf{q}$. Using this notation, we have
\begin{align}
  t_{\max} &= \frac{1}{2}\min_{\mathbf{q}} |\mathbf{q}|\cdot \mathbf{m}-1\ ,
\end{align}
where the minimization is over all integers $q(\lambda) : \lambda\in \irreps_G$ satisfying conditions in 
\cref{ue}.  This, in particular, implies the bound
\begin{equation}\label{eq:tdims}
  t_{\max} \geq \frac{1}{2} \min_{\mathbf{q}}
  \sum_{\lambda \in \operatorname{supp}(\mathbf{q})} m_\lambda-1\ ,
\end{equation}
where the summation is over the support of function $\mathbf{q}$, or equivalently, over all irreps $\lambda$ for which $Q \Pi_\lambda$ is non-zero. On the other hand, choosing any arbitrary integers $\mathbf{q}$, provided that the corresponding operator $Q$ satisfies $\Tr(H_l Q)=0: l=1,\cdots, N$ and $\Tr(Q)=0$, then we obtain an upper bound on $t_{\max}$, as
\be\label{upper}
t_{\max} \le \frac{1}{2}\|Q\|_1-1\ .
\ee

It is useful to note that for any choice of integers $\textbf{q}$ 
satisfying the condition $\sum_{\lambda\in \irreps_G} q(\lambda) m_\lambda = 0$ in \cref{cond0}, the norm of the corresponding operator $Q=\sum_{\lambda \in \irreps_G} q(\lambda) m_\lambda \frac{\Pi_\lambda}{\tr \Pi_\lambda}$ is lower bounded by
\begin{align}\label{Eq:change}
  \|Q\|_1= \sum_{\lambda} |q(\lambda)|\times m_\lambda\ge 2 \max_{\lambda\in \text{supp}(\textbf{q})} m_\lambda\ .
\end{align}
In \cref{Sec:Examples}, we often use this bound to show that the support of the optimal operator $Q$ should be restricted to sectors with the lowest multiplicities.

\subsection*{Bounding \texorpdfstring{$t_{\max}$}{tmax} using matrix \texorpdfstring{$\mathbf{M}$}{M}}\label{sec:upper}

To apply \cref{prop1} and \cref{thm:tdesign}, it is useful to consider an ordering of the irreps
$\lambda_1, \lambda_2, \cdots$ with weakly increasing multiplicity space dimension, such that
\be\label{oredering}
i < j\ \ \Longrightarrow \ \ m_{\lambda_i}\leq m_{\lambda_j}\ .
\ee
Considering this ordering we can express the optimal lower bound on $t_{\max}$ obtained from \cref{prop1} in the following way. If it exists, let $\ell$ be the smallest integer such that condition in \cref{cond}, or equivalently \cref{cond9}, does not hold, i.e.,
\begin{align}\label{eq:min_ell}
  \ell & = \min \big\{ a: \rk [\mathbf{M}_{l,\mu_i}] = a - 1, \notag \\ & \phantom{= \min \big\{ a :} {} i=1,\cdots, a ; \ \  l=1,\cdots, N\big\} .
\end{align}
Then $t_{\max} \geq m_{\lambda_\ell}-1$. Furthermore, choosing $\Delta = \{\lambda_1, \lambda_2, \cdots, \lambda_\ell\}$, it follows that the kernel of $\mathbf{M}^\Delta$ is one-dimensional, so if there is a null vector $\mathbf{q}$ with integer coefficients (e.g., if $\mathbf{M}^\Delta$ has rational entries), then \cref{thm:tdesign} implies that $t_{\max} + 1 \leq \frac{1}{2} \mathbf{m} \cdot |\mathbf{q}|$.

\subsection{Symmetric quantum circuits}\label{Sec:kqudit}

Next, we focus on $\mathcal{V}^G_{n,k}$, the group of $G$-invariant unitaries that can be realized with $k$-qudit $G$-invariant gates, on a system with the total Hilbert space $(\mathbb{C}^d)^{\otimes n}$ and with on-site representation $U(g)=u(g)^{\otimes n}: g\in G$.  In this case, we find a general and canonical way to write the matrix $\mathbf{M}$, which is determined by the parameters $n$, $k$, and the representation $u$ of the group $G$.  This matrix is denoted by $\mathbf{S}$, as defined below in \cref{eq:S-def}.

First, note that the projectors $\Pi_\lambda$ are permutationally-invariant operators. Therefore, if $H_1$ and $H_2$ are related to each other by a permutation $\P$, as $H_2=\P H_1 \P^\dag$, then $\Tr(H_1\Pi_\lambda)=\Tr(H_2\Pi_\lambda)$. It follows that rather than all $k$-local $G$-invariant Hamiltonians, it suffices to consider only a subset that acts on a fixed set of $k$ qudits.

In fact, results of Ref. \cite{Marvian2022Restrict}, in particular, \cref{lemma-1} provides a more systematic approach for characterizing the projection of these Hamiltonians to $\mathrm{span}\{\Pi_\lambda:\lambda\in \irreps_G(n)\}$. We can rephrase the condition in \cref{span} of this lemma, in terms of the inner products with a set of operators 
\be\label{def:S}
S_\nu=\sum_{\lambda\in \irreps_G(n)} s_\nu (\lambda) \frac{\Pi_\lambda}{{m_\lambda}}\ \ : \nu\in \irreps_G(k)\ , 
\ee
where 
\begin{equation}\label{int}
  s_\nu(\lambda) =\int \mathrm{d} g \, [\Tr(u(g))]^{n - k} f_\lambda(g)^\ast f_\nu(g) \ ,
\end{equation}
and the integral is with respect to the Haar measure over group $G$, and $f_\lambda$ is the character of irrep $\lambda$ of this group. As we explain at the end of this subsection, $s_\nu(\lambda)$ are non-negative integers.

Using $\Tr(\Pi_\lambda u(g)^{\otimes n})= m_\lambda f_\lambda(g)$, one can see that these operators satisfy the property
\begin{align}
    \Tr(S_\nu u(g)^{\otimes n}) = [\tr u(g)]^{n - k} \times f_\nu(g)\ .
\end{align}
Then, \cref{span} can be restated as 
\be
\sum_{\lambda\in \irreps_G(n)} \frac{\Tr(H\Pi_\lambda)}{\Tr(\Pi_\lambda)}\Pi_\lambda \in \mathrm{span}_\mathbb{R}\{S_\nu: \nu\in \irreps_G(k)\}\ .
\ee
In other words, the subspace spanned by operators $S_\nu: \nu \in \irreps_G(k)$ is equal to the projection of the Lie algebra associated with $\mathcal{V}_{n,k}^G$ to 
$\mathrm{span}_\mathbb{R}\{\i \Pi_\lambda:\irreps_G(n)\}$, which is the center of the Lie algebra of all $G$-invariant operators (See \cite{Marvian2022Restrict} for further discussion). Using this characterization, condition in \cref{cond2} can be rewritten as
\be\label{jdsa}
\Tr(Q S_\nu)=0 \ \ \ \ \ \ : \nu\in \irreps_G(k)\ .
\ee
Furthermore, $\e^{\i \theta \ident} \in \mathcal{V}_{n, k}^G$ for all $\theta \in [0, 2 \pi)$, which means that \cref{cond2} implies \cref{cond0}. Therefore, \cref{thm:tdesign} implies that 
\begin{corollary}\label{corf}
  Assuming $k$-local $G$-invariant gates are semi-universal, i.e., 
  $\mathcal{SV}^G_{n,n}\subset\mathcal{V}^G_{n,k}$, then the uniform distribution over the group $\mathcal{V}^G_{n,k}$ generated by $k$-local $G$-invariant gates is an exact $t$-design for the uniform distribution over the group of all $G$-invariant unitaries $\mathcal{V}^G_{n,n}$ with the maximum $t$ equal to 
  \begin{align}\label{sol}
    t_{\max} &=\min_{\mathbf{q}} \big\{\frac{\mathbf{m} \cdot {\mathbf{|q|}}}{2}-1 : {\mathbf{q}}\in \mathbb{Z}^{|\irreps_G(n)|}, {\mathbf{q}}\neq {\mathbf{0}}, {\mathbf{S} \mathbf{q}}={\mathbf{0}} \big\} ,
  \end{align}
  where the matrix $\mathbf{S}$ is defined as
  \begin{align}\label{eq:S-def}
    {\mathbf{S}_{\nu, \lambda }} := s_{\nu}(\lambda)\ : \nu\in \irreps_G(k) , \lambda\in \irreps_G(n) ,
  \end{align}
  and $s_{\nu}(\lambda)$ are non-negative integers defined in \cref{int}.
\end{corollary}

Note that except matrix $\mathbf{S}$, the rest of this characterization is independent of $k$, the locality of interactions. This matrix, whose role is analogous to that of matrix $\mathbf{M}$ in \cref{ue}, determines 
the projection of the Lie algebra associated to $\mathcal{V}^G_{n,k}$ 
to the linear space of operators $\mathrm{span}_\mathbb{R}\{\i\Pi_\lambda: \lambda\in \irreps_G(n)\}$, which is the center of the Lie algebra of $G$-invariant operators.   In particular, according to \cref{lemma-1,lem1},
\begin{align}
\rk(\textbf{S}) &=\dim(\mathrm{span}_\real\{S_\nu: \nu\in\irreps_{G}(k)\})\nonumber\\ &=
\dim(\mathrm{span}_\real\{r^{n-k} f_\nu: \nu\in\irreps_{G}(k)\})\nonumber\\ &=\dim(\mathrm{span}_\mathbb{R}\{|\chi_A\rangle: A=A^\dag, \text{$A$ is $k$-local}\})
\ ,
\end{align}
where $|\chi_A\rangle=\sum_{\lambda\in \irreps_G(n)} \Tr(\Pi_\lambda A) |\lambda\rangle$. In other words, the rank of matrix $\textbf{S}$ is equal to the dimension of the subspace spanned by the projection of $k$-local operators to the  
center of the Lie algebra of all $G$-invariant operators. Then, \cref{dkd} can be rewritten as
\be
  \dim(\mathcal{V}_{n,n}^G)-\dim(\mathcal{V}_{n,k}^G)\ge |\irreps_G(n)|-\rk(\textbf{S}) \ .
\ee
Recall that according to \cref{lem1}, $\rk(\textbf{S}) \le |\irreps_G(k)|$, and the equality holds when $G$ is a connected group, or if $\Tr(u(g))\neq 0$ for all $g\in G$.

Crucially, to determine $t_{\max}$, the only relevant property of $\mathbf{S}$ is its null space. In particular, for any matrix $\mathbf{T}$ if $\rk(\mathbf{T} \mathbf{S})=\rk(\mathbf{S})$ then the condition $\mathbf{T} \mathbf{S} \mathbf{q} = \mathbf{0}$ is equivalent to $\mathbf{S} \mathbf{q} = \mathbf{0}$, and therefore it imposes exactly the same constraint. For instance, suppose rather than matrix $\mathbf{S}$ we consider the matrix $\mathbf{M}$ in \cref{matrix} for a set of Hamiltonians $\{H_l\}$ that forms a basis for the Lie algebra associated to $\mathcal{V}^G_{n,k}$. Then, the corresponding matrix can be written as $\mathbf{M}=\mathbf{T} \mathbf{S}$ for some matrix $\mathbf{T}$, and they satisfy $\rk(\mathbf{M})= \rk(\mathbf{T} \mathbf{S})=\rk(\mathbf{S})$, which means the constraint $\mathbf{M} \mathbf{q}=0$ is equivalent to $\mathbf{S} \mathbf{q}=0$. 

In the following, we see other examples of this freedom in \cref{Sec:U(1)}, and \cref{Sec:SU(2)}, where we study U(1) and SU(2) symmetry. In particular, we see that rather than $S_\nu: \nu\in \irreps_G(k)$ it is convenient to use other bases for the space $\mathrm{span}_\mathbb{R}\{S_\nu: \nu\in \irreps_G(k)\}$. Then, instead of matrix $\mathbf{S}$ we obtain other matrices whose null spaces are identical with that of $\mathbf{S}$.

\subsection*{Integers \texorpdfstring{$s_{\nu}(\lambda)$}{snu(lambda)} and their interpretation}

%\red{should we call $m_{\mu,\nu}^\lambda$ "Kronecher coefficients", or this name is reserved only for $\mathbb{S}_n$? Maybe we should emphasize more about these coefficients, and mention them in the intro.  At the end matrix $S$ is obtained from contracting this tensor with the multiplicity on $n-k$ qudits. }

According to \cref{corf}, $t_\text{max}$ is determined by the kernel of matrix  ${\mathbf{S}_{\nu, \lambda }} := s_{\nu}(\lambda)$,  where  $s_{\nu}(\lambda)$ is defined in \cref{int}. Here, we discuss an interesting interpretation of coefficients  $s_{\nu}(\lambda)$, which immediately implies that they are integers and provides a simple formula for calculating them. 

For irreducible representations $\mu$ and $\nu$, the product of characters $f_{\mu} f_{\nu}$ can be decomposed as a linear combination of characters,
\begin{equation}
    f_{\mu} f_{\nu}=\sum_{\lambda \in \text{Irreps}_G} m_{\mu, \nu}^{\lambda} f_\lambda,
\end{equation}
where $m_{\mu, \nu}^{\lambda}$ is the multiplicity of irrep $\lambda$ in the representation $\mu\otimes \nu$. Using the isotypic decomposition of $u(g)^{\otimes (n - k)}$, with $m_\mu(n-k)$ the multiplicity of irrep $\mu$, it follows that
\begin{equation}
    \tr (u(g))^{n - k} = \sum_{\mu \in \text{Irreps}_G} m_\mu(n-k) f_\mu(g).
\end{equation}
Combining these and using Schur orthogonality of irreducible characters, we find
\begin{equation}\label{eq:s-multiplicity}
    {\mathbf{S}_{\nu, \lambda}} = s_\nu(\lambda) = \sum_{\mu \in \text{Irreps}_G} m_{\mu}(n-k)\times  m_{\nu, \lambda^\ast}^{\mu^\ast},
\end{equation}
where $\lambda^\ast$ is the dual (complex conjugate) representation to $\lambda$. In particular, when $G$ is Abelian, $\nu\otimes \lambda^\ast$ is a single irrep, and therefore 
\be\label{multi}
{\mathbf{S}_{\nu, \lambda }}=s_\nu(\lambda)=m_{\nu^\ast\otimes \lambda}(n-k)\ ,
\ee
which is the multiplicity of irrep $\nu^\ast \otimes \lambda \cong (\nu \otimes \lambda^\ast)^\ast$ in representation $u(g)^{\otimes (n-k)}: g\in G$. For instance, for U(1) symmetry on $n$ qubits, this immediately 
implies \cref{S:U(1)}.

\subsection{The extra commutants (Proof of \texorpdfstring{$t_{\max} \leq \frac{1}{2}\|Q\|_1-1$}{tmax<=1/2Q-1})}\label{Sec:comm}

Below \cref{thm:tdesign}, we provided one interpretation of integers $\mathbf{q}$ and the corresponding operator $Q$, as a $G$-invariant Hamiltonian that is not contained in the Lie algebra associated to $\mathcal{W}^G$. Next, we provide another interpretation of integers $\mathbf{q}$ and the corresponding operator $Q$
and prove that $t_{\max} \leq \frac{1}{2}\|Q\|_1-1$.

Recall that if $\mathrm{Comm}\{V^{\otimes t}: V\in \mathcal{W}^G\}\neq \mathrm{Comm}\{V^{\otimes t}: V\in \mathcal{V}^G\}$, then the uniform distribution over $\mathcal{W}^G$ is not a $t$-design for the uniform distribution over $\mathcal{V}^G$. Here, we show how any solution to \cref{eq:cond}, or equivalently \cref{sol}, gives operators that commute with all $V^{\otimes t}: V\in \mathcal{W}^G$  but not with all $V^{\otimes t}: V\in \mathcal{V}^G$. This also proves $t_{\max} \leq \frac{1}{2}\|Q\|_1-1$ in \cref{thm:tdesign} and \cref{corf}.

In particular, consider group $\mathcal{W}^G$ generated by the one-parameter families $\{\exp(-\i H_l t): t\in\mathbb{R}\}_l$, and let $\mathbf{M}$ be the matrix defined in \cref{matrix}.
Recall the decomposition $\mathcal{H}\cong\bigoplus_{\lambda\in\irreps_G} (\mathcal{Q}_\lambda\otimes \mathcal{M}_\lambda)$.
Assume there exists $\mathbf{q}\in \mathbb{Z}^{|\irreps_G|}$ satisfying $\mathbf{q}\neq \mathbf{0}$, ${\mathbf{m}} \cdot {\mathbf{q}} ={0}$, and $\mathbf{M} \mathbf{q}=\mathbf{0}$. Define $t= \frac{1}{2} |\mathbf{q}|\cdot \mathbf{m}$, which is an integer because $\mathbf{m} \cdot {\mathbf{q}} ={0}$. Consider the vector $|\Psi\rangle=|\Psi_+\rangle\otimes |\Psi_-\rangle $, where $|\Psi_\pm\rangle\in \mathcal{H}^{\otimes t}$ are defined by 
\begin{align}
  |\Psi_\pm\rangle&=\bigotimes_{\lambda: \pm 
 q(\lambda)>0} 
         \bigotimes_{j=1}^{\pm q(\lambda)}|\Psi^j_\lambda\rangle \ ,
\end{align}
and
\be
|\Psi^j_\lambda\rangle=|\psi^j_\lambda\rangle_{\mathcal{Q}_\lambda} \otimes |\mathrm{singlet}\rangle_{\mathcal{M}_\lambda}\in \mathcal{H}^{\otimes m_\lambda}\ .
\ee
Here, $|\psi^j_\lambda\rangle_{\mathcal{Q}_\lambda}$ is an arbitrary vector in $ \mathcal{Q}_\lambda^{\otimes m_\lambda}$ and $|\mathrm{singlet}\rangle_{\mathcal{M}_\lambda}$ is the totally anti-symmetries state in $\mathcal{M}_\lambda^{\otimes m_\lambda}$, which is unique up to a normalization (recall that $m_\lambda=\dim(\mathcal{M}_\lambda)$). We show that for any such vector 
\be\label{ds}
V\in \mathcal{W}^G: \ \ \ \ [V^{\otimes t}\otimes {V^\ast}^{\otimes t}]|\Psi\rangle=|\Psi\rangle\ ,
\ee
which, in turn, implies $
\mathbb{E}_{V\sim\nu }[V^{\otimes t}\otimes {V^\ast}^{\otimes t}]|\Psi\rangle =|\Psi\rangle $, 
where $\nu$ is the uniform distribution over $\mathcal{W}^G$. On the other hand, for the uniform distribution over the group of all $G$-invariant unitaries $\mathbb{E}_{V\sim\mu_\mathrm{Haar} }[V^{\otimes t}\otimes {V^\ast}^{\otimes t}]|\Psi\rangle=0$ as explained below. Equivalently, this means that 
operator $|\Psi_+\rangle\langle \Psi_-|$ 
is in the commutant of $V^{\otimes t}: V\in \mathcal{W}^G$, whereas it is not in the commutant of $V^{\otimes t}: V\in \mathcal{V}^G$. Note that in the case of symmetric quantum circuits with $k$-qudit $G$-invariant gates, i.e., $\mathcal{W}^G=\mathcal{V}^G_{n,k}$, the same construction works for any solution $\mathbf{q}\in \mathbb{Z}^{|\irreps_G(n)|}$ satisfying $\mathbf{q}\neq \mathbf{0}$ and $\mathbf{S} \mathbf{q}=\mathbf{0}$.

To show \cref{ds}, first note that the totally anti-symmetric state $|\mathrm{singlet}\rangle_{\mathcal{M}_\lambda}$ satisfies
\be
U^{\otimes m_\lambda} |\mathrm{singlet}\rangle_{\mathcal{M}_\lambda}=\det(U)|\mathrm{singlet}\rangle_{\mathcal{M}_\lambda}\ ,
\ee
for any operator $U$ on $\hilbert[M]_\lambda$. This implies that for all $V\in \mathcal{V}^G$
\begin{align}
  [V^{\otimes t}\otimes {V^\ast}^{\otimes t}]|\Psi\rangle &=\big[\prod_{\lambda} \det(v_{\lambda})^{q(\lambda)}\big]|\Psi\rangle \ ,
\end{align}
where we have used the decomposition $V=\bigoplus_\lambda \mathbb{I}_{\mathcal{Q}_\lambda}\otimes v_\lambda$. Obviously for general $V\in \mathcal{V}^G$, $\prod_{\lambda} \det(v_{\lambda})^{q(\lambda)}$ can be any arbitrary phase, which in turn implies 
$\mathbb{E}_{V\sim\mu_\mathrm{Haar} }[V^{\otimes t}\otimes {V^\ast}^{\otimes t}]|\Psi\rangle=0$.

Recall that $\mathcal{W}^G$ is the group generated by $\{\exp(\i r_l H_l): r_l\in\mathbb{R}\}_l$. Using the fact that for any operator $A$, $\det(\e^{A})=\e^{\Tr(A)}$, we find that for any $G$-invariant Hamiltonian $H$,  and the corresponding unitary $V=\exp(\i H)$, 
\be
\det(v_{\lambda})=\exp\Big(\i m_\lambda \frac{\Tr(H \Pi_\lambda) }{\Tr(\Pi_\lambda)}\Big)\ .
\ee
Therefore
\begin{align}
  \prod_{\lambda} \det(v_{\lambda})^{q(\lambda)}&=\exp\Big(\i \sum_{\lambda} q(\lambda) m_\lambda \frac{\Tr(H \Pi_\lambda)}{\Tr(\Pi_\lambda)}\Big)\ ,
\end{align}
which is equal to $\exp(\i \Tr(H Q))$.  We conclude that for $V=\exp(\i H)$, 
\begin{align}
  [V^{\otimes t}\otimes {V^\ast}^{\otimes t}]|\Psi\rangle &=\exp(\i \Tr(H Q) ) |\Psi\rangle \ ,
\end{align}
or equivalently,
\begin{align}\label{qssaq}
&\exp(\i H)^{\otimes t}|\Psi_+\rangle\langle \Psi_-| \exp(-\i H)^{\otimes t} \\ &\ \ \ \ \ \ \ \ \ \ \ \ \ \ \ \ \ \ \ =\exp(\i   \Tr(H Q))\ |\Psi_+\rangle\langle \Psi_-|\ .\nonumber
\end{align}
In other words, for all $G$-invariant  Hamiltonian $H$, if $H$ is orthogonal to $Q$, then $\exp(\i H)^{\otimes t}$ commutes with $|\Psi_+\rangle\langle \Psi_-|$.  
 Therefore, since a general element $V\in \mathcal{W}^G$ is obtained by composing one-parameter families 
$\{\exp(\i H_l r_l):r_l\in\mathbb{R} \}_l$, if $\Tr(H_l Q)=0$ for all $H_l$, then $|\Psi_+\rangle\langle \Psi_-|$ commutes with $V^{\otimes }: V\in \mathcal{W}^G$ (or, equivalently, \cref{ds} holds for arbitrary $V\in \mathcal{W}^G$).

In summary, integers $q(\lambda)$, or equivalently, operator $Q$ in \cref{thm:tdesign} determine operators that commute with $V^{\otimes t}: V\in\mathcal{W}^G$, but not with $V^{\otimes t}$ for some elements of $V\in\mathcal{V}^G$. In particular, \cref{qssaq} implies that while operator $|\Psi_+\rangle\langle \Psi_-|$ commutes with $V^{\otimes t}$ for all  $V\in \mathcal{W}^G$,  it does not commute with $\exp(\i s Q)^{\otimes t}$ for generic values of $s\in\mathbb{R}$ and $Q$ defined in \cref{QQ},  because in this case  the right-hand side of \cref{qssaq} gets an additional phase $\exp(\i s  \Tr(Q^2))$, which is not 1 for generic $s\in\mathbb{R}$. (Note that permuted versions of $|\Psi_+\rangle\langle \Psi_-|$ and their linear combinations also satisfy this relation.) This implies that for $t= \frac{1}{2} |\mathbf{q}|\cdot \mathbf{m}$, the uniform distribution over $\mathcal{W}^G$ is not a $t$-design for the uniform distribution over $\mathcal{V}^G$. Therefore, $t_{\max} \leq \frac{1}{2} |\mathbf{q}|\cdot \mathbf{m} -1 = \frac{1}{2}\|Q\|_1-1$.

\subsection*{Example of the extra commutants: \texorpdfstring{\\}{} 2 qubits with \texorpdfstring{$\mathbb{Z}_2$}{Z2} symmetry}\label{Sec:ex:last}

Here, we consider a simple example of $n=2$ qubits with on-site $\mathbb{Z}_2$ symmetry, represented by $\mathbb{I}\otimes \mathbb{I} , Z\otimes Z$. Then, the total Hilbert space decomposes to two 2D sectors $\hilbert_0$ and $\hilbert_1$, corresponding to the trivial irrep of $\mathbb{Z}_2$, spanned by $|00\rangle$ and $|11\rangle$, and the nontrivial irrep, spanned by $|01\rangle$ and $|10\rangle$, respectively.

In this case,
\be
\mathcal{V}^G= \{(v_0\oplus v_1)\ : v_0, v_1\in \U(2)\}\ ,
\ee
where $v_i$ acts on $\hilbert_i$, while
\be
\mathcal{SV}^G= \{(v_0\oplus v_1)\ : v_0, v_1\in \SU(2)\}\ ,
\ee
which means $\mathcal{V}^G \cong \U(2)\times \U(2)$ and $\mathcal{SV}^G \cong \SU(2)\times \SU(2)$.

In this example, $\mathcal{W}^G$ being semi-universal does not guarantee that the uniform distribution over it is a 2-design. Suppose that $\mathcal{W}^G = \mathcal{SV}^G$ and consider the states of two copies of the total Hilbert space,
\begin{align}
|\mathrm{singlet}_0\rangle &=\frac{|00\rangle \otimes |11\rangle-|11\rangle \otimes |00\rangle}{\sqrt{2}}\\
|\mathrm{singlet}_1\rangle &=\frac{|01\rangle \otimes |10\rangle-|10\rangle \otimes |01\rangle}{\sqrt{2}}\ .
\end{align}
Then, the operator $\qout{\mathrm{singlet}_0}{\mathrm{singlet}_1}$ commutes with $W \otimes W$ for all $W \in \mathcal{W}^G = \mathcal{SV}^G$ but not $V\otimes V$ for all   $V \in \mathcal{V}^G$. In particular, let $Q = \Pi_0 - \Pi_1 =Z\otimes Z$ be the difference between the Hermitian projectors to the charge sectors $\hilbert_0$ and $\hilbert_1$ (see also \cref{sec:Zp}, where a similar operator is defined for all on-site $\mathbb{Z}_p$ symmetry for even $p$). Then $V = \exp(\i \theta Q) \in \mathcal{V}^G$ is not an element of $\mathcal{W}^G$ for general $\theta$, and
\be
\begin{split}
(V\otimes V) & |\mathrm{singlet}_0\rangle\langle \mathrm{singlet}_1| (V^\dag\otimes V^\dag) \\ & = \exp(\i4\theta)|\mathrm{singlet}_0\rangle\langle \mathrm{singlet}_1|\ .
\end{split}
\ee

\section{Examples: \texorpdfstring{$\U(1)$}{U(1)}, \texorpdfstring{$\SU(2)$}{SU(2)}, \texorpdfstring{$\mathbb{Z}_p$}{Zp} and \texorpdfstring{$\SU(d)$}{SU(d)}}\label{Sec:Examples}

Equipped with the tools developed in the previous section, now we are ready to calculate $t_{\max}$ for systems of particular interest, namely qubit systems with $\U(1)$, $\SU(2)$ and $\mathbb{Z}_p$ symmetries. For these systems, we are able to determine the exact value of $t_{\max}$ for arbitrary locality $k$ and a sufficiently large number of qubits $n$. In addition, we also consider qudit systems with $\SU(d)$ symmetry for up to 4-local gates.

\subsection{\texorpdfstring{$\U(1)$}{U(1)} symmetry}\label{Sec:U(1)}

Consider a qubit system with $\U(1)$ symmetry represented as $u(\theta) = \e^{\i \theta Z}$, where $\theta \in [0, 2 \pi)$ and $Z$ is the Pauli $Z$ operator. The on-site representation $u(\theta)^{\otimes n}$ decomposes into invariant subspaces labeled by the eigenvalues of $\sum_{a = 1}^n Z_a$, which are 
$2n- w: w = 0, 1, \cdots, n $, where $w$ is called the Hamming weight associated to the irrep. 
Denote the Hermitian projector to the eigenspace with Hamming weight $w$ by $\Pi_w$, and note that $\tr \Pi_w = \binom{n}{w}$ since we may understand the multiplicity of Hamming weight $w$ combinatorially as the number of ways to have $w$ ``excitations'' of the $n$ qubits.

In the following, we consider the groups $G_{XX + YY, Z}$ and $\mathcal{V}_{n, k}^{\U(1)}$, where
\begin{equation}
  \begin{split}
    G_{XX + YY, Z} & = \langle \e^{\i \theta_0} \ident, \e^{\i \theta_1 Z_a}, \e^{\i \theta_2 R_{ab}} : \\
                   & \phantom{{} = {} \langle} \ \ \ \theta_l \in [0, 2 \pi), 1 \leq a < b \leq n \rangle
  \end{split}
\end{equation}
is the group generated by global phases, single-qubit unitaries $\e^{\i \theta Z}$, and interaction $R = \frac{1}{2} (X \otimes X + Y \otimes Y)$. We can also write
\begin{equation}
  \mathcal{V}_2^{\U(1)} = \langle G_{XX + YY, Z}, \e^{\i \theta Z_a Z_b} : \theta \in [0, 2 \pi), 1 \leq a < b \leq n \rangle.
\end{equation}
We compare the uniform distribution over groups $G_{X X + Y Y, Z}$, $\mathcal{V}_{n, 2}^{\U(1)}$, and $\mathcal{V}_{n, k}^{\U(1)}$ with the uniform distribution over the group of all $\U(1)$-invariant unitaries $\mathcal{V}_{n, n}^{\U(1)}$, and show that
\begin{itemize}
\item $G_{X X + Y Y, Z}$ is an $(n - 1)$-design, but not an $n$-design.
\item $\mathcal{V}_{n, 2}^{\U(1)}$ is a $(2 n - 3)$-design, but not a $(2 n - 2)$-design.
\item $\mathcal{V}_{n, k}^{\U(1)}$ is a $t$-design with $t_{\max}$ given in \cref{eq:tmax-U1}.
\end{itemize}

To establish these results, we use the fact that $G_{XX + YY, Z}$ is semi-universal \cite{Marvian2022Restrict, bai2024synthesis}, and therefore the groups generated by $k$-local $\U(1)$-invariant gates, $\mathcal{V}_{k, n}^{\U(1)} \supseteq G_{XX + YY, Z}$, are also semi-universal for $k \geq 2$. Furthermore, since $\U(1)$ is Abelian, each of its irreps are 1D, so the dimensions $\tr \Pi_w = m_w = \binom{n}{w}$ are just the multiplicities. 
Therefore, in this case operator $Q$ in \cref{thm:tdesign} simplifies to 
\be
Q=\sum_{w=0}^n q(w) \Pi_w\ ,
\ee
where $q(w)$ are integers. Similarly,
\be\label{eq:MU1}
\mathbf{M}_{l,w}= \Tr(\Pi_w H_l)\ : l=1,\cdots, N ; \mu\in \irreps_G\ ,
\ee

Furthermore, taking into account that $m_w = m_{n - w}$, there is an ordering on these dimensions which is weakly increasing. We denote this order $i = 0, 1, \cdots, n + 1$ where the Hamming weight $w_i = \floor{i / 2}$ when $i$ is even and $w_i = n - \floor{i / 2}$ when $i$ is odd. 

\subsubsection{Example:  \texorpdfstring{$k=2$}{k=2}}

We first note that $\tr (X_a X_b + Y_a Y_b) \Pi_w = 0$ for all Hamming weights $w = 0, 1, \cdots, n$. Furthermore, since each $\Pi_w$ is permutationally-invariant, it follows that $\tr Z_a \Pi_w = \tr Z_b \Pi_w$ for all qubits $a$ and $b$. Thus, in the matrix $\mathbf{M}_{i; l}$ defined in \cref{matrix}, for $G_{XX + YY, Z}$ we need only consider the Hamiltonians $H_0 = \ident$ and $H_1 = Z_1$ and the sectors $\Delta = \{w_0 = 0, w_1 = n\}$,
\begin{equation}
  \mathbf{M}^\Delta = \begin{pmatrix}
    1 & 1 \\ 1 & -1
  \end{pmatrix}\ ,
\end{equation}
which is clearly full-rank, i.e., satisfies $\rk(\mathbf{M}^\Delta)=|\Delta|=2$. 
The minimum multiplicity of charge sectors not included in $\Delta$ is $n = \tr \Pi_1 = \tr \Pi_{n - 1}$. Thus, according to \cref{prop1}, the uniform distributions over $G_{XX + YY, Z}$ and $\mathcal{V}_{n, k}^{\U(1)}$ are (at least) $(n - 1)$-designs for the uniform distribution over $\mathcal{V}^{\U(1)}$.

To show that this bound is tight for $G_{XX + YY, Z}$, we use \cref{thm:tdesign}. In particular, we consider 
a special case of a family of operators $\{F_k\}$ defined below in \cref{eq:Fk}. Namely, we choose $Q = F_2$, where
\be
F_2 = (n - 1) \Pi_0 - \Pi_1 + \Pi_n\ .
\ee
Note that $F_2$ has integer eigenvalues, is traceless, and satisfies $\Tr(F_2 Z_j)=0$ for all $j=1,\cdots, n$ (See \cref{app:U1}). Therefore, it satisfies the assumptions of \cref{thm:tdesign} with $t = \frac{1}{2} \|F_2\|_1 = n$, establishing that $G_{XX + YY, Z}$ is not an $n$-design, i.e., $t_{\max} = n - 1$. (Note that operator $Q=F_2$ is uniquely determined up to an integer multiple by the condition that $\mathbf{q}$ has support on only these three irreps, $w_0, w_1, w_2$, and $\mathbf{M} \mathbf{q} = 0$.)

On the other hand, for $\mathcal{V}_2^{\U(1)}$ we have in addition the $ZZ$ interaction. We may extend the $\mathbf{M}$ matrix to include $\mathbf{M}_{2, i} = \tr \Pi_{w_i} Z_a Z_b$, and the result is full-rank,
\begin{equation}
  \mathbf{M}^{\Delta'} = \begin{pmatrix}
    1 & 1 & n \\
    1 & -1 & n - 1 \\
    1 & 1 & n - 2
  \end{pmatrix},
\end{equation}
(now with the three lowest in the order, $\Delta' = \{w_0 = 0, w_1 = n, w_2 = 1\}$). However, the lower bound of \cref{prop1} does not change since $m_{w_2} = m_{w_3} = n$.

Furthermore, as we will see, this lower bound is not tight. Indeed, using \cref{thm:tdesign} we show that for the uniform distribution over $\mathcal{V}_{n, 2}^{\U(1)}$, for all $n\ge 4$, it holds that $t_{\max} = 2 n - 3$. In particular, choosing $Q=F_3$, where 
\be
F_3 = (n - 2) (\Pi_0 - \Pi_n) - (\Pi_1 - \Pi_{n - 1})\ ,
\ee
is defined in \cite{zhukas2024observation}, which is also a special case of the family of operators $\{F_k\}$ defined in \cref{eq:Fk}.  This operator  is orthogonal to 2-local operator $H$, i.e., $\Tr(F_3 H)=0$ \cite{zhukas2024observation}  (See \cref{lem:Fk} below for a generalization of this result). Therefore, \cref{thm:tdesign} immediately implies that 
\be
t_{\max} \leq \frac{1}{2} \|F_3\|_1-1=2n-3 \ .
\ee
It can be easily shown that this is indeed the optimal operator $Q$ with required properties, and therefore, this bound holds as an equality. To see this, we can use the fact that, up to normalization, $F_3$ is the only operator whose support is restricted to sectors with $w=0,1,n-1,n$ and is orthogonal to all 2-local operators (see \cref{cor:Fk-unique}). The requirement that the eigenvalues of $Q$ should be integer means that $Q=\pm F_3$, which has an eigenvalue $\pm 1$, is optimal among such operators. On the other hand, if another operator $Q'$ satisfying \cref{eq:cond} has support in any sector other than $w'\neq 0,1,n-1,n$, then for $n\ge 4$, 
\be\label{jdj}
\frac{\|Q'\|_1}{2} \ge m_{w'}= \binom{n}{w'} \ge \binom{n}{2} \ge 2(n-1)= \frac{\|F_3\|_1}{2}\ , 
\ee
where the first inequality follows from 
\cref{Eq:change}. We conclude that for $n\ge 4 $, $t_{\max}=2n-3$. 

\subsubsection{General \texorpdfstring{$k$}{k}}

It is helpful to recall the definition and properties of operators $C_l$  from \cite{Marvian2022Restrict}, 
\begin{equation}\label{eq:u1cls}
  C_l = \sum_{\mathbf{b} : \ell(\mathbf{b}) = l} \mathbf{Z}^{\mathbf{b}} = \sum_{w = 0}^n c_l(w) \Pi_w,
\end{equation}
where the sum is over all bitstrings $\mathbf{b} = b_1 b_2 \cdots b_n$ with $l$ 1s, i.e., $\ell(\mathbf{b}) := \sum_a b_a = l$,
\begin{equation}
  \mathbf{Z}^{\mathbf{b}} = Z^{b_1} \otimes Z^{b_2} \otimes \cdots \otimes Z^{b_n},
\end{equation}
and 
\be\label{Eq:clw}
c_l(w) = \sum_{r = 0}^w (-1)^r \binom{n-w}{l-r} \binom{w}{r}\ .
\ee
For example, $C_0=\mathbb{I}^{\otimes n}$, and
\begin{align}\label{eq:C1}
  C_1 = \sum_{a = 1}^n Z_a = \sum_{w=0}^n (n-2w) \Pi_w\ .
\end{align}
The $C_l$ operators are Hermitian, have integer eigenvalues, and satisfy the orthogonality relation 
\be\label{eq:Cl-orthogonal-U1}
\Tr(C_l C_{l'})=\delta_{l,l'} 2^n \binom{n}{l}\ ,
\ee
and form a complete basis for the space $\mathrm{span}\{\Pi_w\}=\mathrm{span}\{C_l\}$ (This is indeed the center of the $\U(1)$-invariant Hamiltonians). These properties
imply that for any arbitrary operator $B$ in this space, we have the completeness relation 
\begin{align}\label{eq:Cl-complete-U1}
    \Tr(B O)=\sum_{l=0}^n \frac{\Tr(B C_l)\times  \Tr(C_l O)}{\Tr(C_l^2)} ,
\end{align}
where $O$ is an arbitrary operator on $n$ qubits.  Furthermore, the definition of $C_l$ operators in \cref{eq:u1cls}, immediately implies that
for any $k$-local operator $O$,
\be\label{orth}
\tr (O C_l) = 0\ \ : l>k\ . 
\ee
Furthermore, $\i C_l : l = 0, \cdots, k$ span the center of the Lie algebra of the group $\mathcal{V}_{n, k}^{\U(1)}$ \cite{Marvian2022Restrict}. \\

\noindent\textbf{Lower bound on $t_{\max}$}: We start by establishing a lower bound on $t_{\max}$ by applying \cref{prop1}. In particular, we show that for the group $\V_{n,k}^{\U(1)}$, the condition $\rk \mathbf{M}^\Delta = |\Delta|$ holds for any subset of irreps $\Delta$ with $|\Delta| = k + 1$.

For any subset of $k+1$ Hamming weights  $\Delta\subseteq \{0, \cdots, n\}$, consider the matrix 
\begin{align}\label{eq:M-Delta-U1}
    \mathbf{M}_{l, w}^{\Delta} & = \bigl[\tr \Pi_{w} C_l \bigr] = c_l(w) \binom{n}{w} \ , 
\end{align}
where $0\le l\le k$, $w\in\Delta$, and $c_l(w)$ is the eigenvalue of $C_l$ on the irrep $w$. In order to see $\rk \mathbf{M}^\Delta = |\Delta|$,  it is useful to introduce another (non-orthogonal) basis for the space $\operatorname{span}_\real \set{C_l : l = 0, \cdots, k}$, obtained from the powers of $C_1$, i.e., for all integer $k\le n$
\begin{align}\label{eq:spanC1l=Cl}
  \operatorname{span}_\real \set{C_1^l : l = 0, \cdots, k} = \operatorname{span}_\real \set{C_l : l = 0, \cdots, k} .
\end{align}
This implies that there exists a full-rank $(k+1) \times (k+1)$ matrix $\mathbf{T}$, such that for all $l\le k$, 
\begin{align}
    C_1^l = \sum_{l'=0}^k \mathbf{T}_{ll'} C_{l'}\ .
\end{align}
As it can be seen from \cref{eq:C1}, operator $C_1$ has distinct eigenvalues on sectors with different Hamming weights, which implies  
\begin{align}\label{eq:TM-Delta-U1}
    (\mathbf{TM}^\Delta)_{l, w} & = \bigl[\tr \Pi_{w} C_1^l \bigr] = \big(n-2 w\big)^l \binom{n}{w}
\end{align}
has full-rank.\footnote{One can understand $\mathbf{TM}^\Delta$ as the Vandermonde matrix associated with powers of the eigenvalues of $C_1$, times a diagonal matrix whose elements are the multiplicities $m_w$.} Finally, since $\mathbf{T}$ and $\mathbf{TM}^\Delta$ both have full rank, we know $\mathbf{M}^\Delta$ must also have full rank.

So far $\Delta$ can be an arbitrary subset of $k+1$ irreps. To obtain the strongest lower bound from proposition \cref{prop1}, we use the weakly increasing order of the dimension of multiplicity discussed after \cref{eq:MU1}. Thus, the $(k + 2)$th smallest multiplicity $m_{w_{k + 1}}$ gives the lower bound 
\begin{equation}
 t_{\max}+1 \geq  m_{w_{k + 1}} = \binom{n}{\floor{\frac{k + 1}{2}}}\ .
\end{equation}

\noindent\textbf{Upper bound on $t_{\max}$}: Next, we establish a general upper bound on $t_{\max}$. Recall that, according to \cref{upper} any operator $Q=\sum_{w} q(w)\Pi_w$ with integer eigenvalues $q(w)$ that satisfy $\Tr(Q O)=0$ for all $k$-local operator $O$, gives an upper bound $t_{\max}\le \frac{1}{2}\|Q\|_1-1$.

For this, we can use another basis for the space $\mathrm{span}_{\mathbb{R}}\{\Pi_w\}$. Namely, the basis recently defined in \cite{zhukas2024observation} with operators
\begin{equation}\label{Ak}
  A_k = \sum_{w = 0}^k (-1)^w \binom{n - w}{k - w} \Pi_w\ \ \ : k=0,\cdots, n\ ,
\end{equation}
which satisfy
\begin{align}\label{Ak0}
  \Tr(A_k C_l)&=0\ \ \ \ \ \ : l<k\ .
\end{align}
More generally, as we show in \cref{app:U1}, 
\begin{align}
  A_k =2^{k-n}   \sum_{l} \binom{l}{k} C_l \ ,
\end{align}
which implies \cref{Ak0} using the orthogonality of the $C_l$ operators in \cref{eq:Cl-orthogonal-U1} and that the binomial coefficient $\binom{l}{k}=0$ when $k> l \ge 0$. 

Then, $Q=A_{k + 1}$ satisfies all conditions of \cref{thm:tdesign} for $\mathcal{V}_{n, k}^{\U(1)}$. In particular, for any $k$-local operator $O$, whether it is U(1)-invariant or not, \cref{eq:Cl-complete-U1} implies $\Tr(A_{k+1} O)=0$.  Then, combining \cref{orth} and \cref{Ak0} we find that if $O$ is $k$-local then $\Tr(O A_k)=0$. This implies
\begin{equation}
  t_{\max} \leq \frac{1}{2} \|A_{k + 1}\|_1 - 1 = 2^{k} \binom{n}{k + 1} - 1 \ ,
\end{equation}
where $\|A_{k + 1}\|_1$ is calculated in \cref{app:U1}. \\

\noindent{\textbf{Exact value of  $t_{\max}$}:} Finally, to determine the actual value of $t_{\max}$, we introduce yet another basis for $\mathrm{span}\{\Pi_w\}$, namely 
\begin{equation}\label{eq:Fk}
  F_k = \sum_{w = 0}^n f_k(w) \Pi_w = \sum_{w = 0}^n (-1)^{w} \binom{n - \floor{\frac{k + 1}{2}} - w}{n-k} \Pi_w\ ,
\end{equation}
for $k = 0, \dots, n$.\footnote{In \cref{app:U1} we also consider a  slightly different basis, namely  $\widetilde{F}_k=(-1)^k X^{\otimes n}F_k X^{\otimes n}$, which has similar properties as $F_k$ and gives the same result on $t_{\max}$.} 
In particular, it can be easily seen that $$F_0=A_0=\Pi_0=|0\rangle\langle 0|^{\otimes n}\ ,$$
and 
$$F_n = A_n = C_n = Z^{\otimes n}\ .$$ 
More generally,  
\be \label{eq:Fk-kernel}
f_k(w) =0: \ \ \floor{\frac{k}{2}} + 1 \leq w \leq n - \floor{\frac{k+1}{2}}\ ,
\ee 
whereas 
\be
f_k(w) =(-1)^{w}  \binom{n - \floor{\frac{k + 1}{2}} - w}{n-k} \ 
\ee
is non-zero outside this interval. Note that for $w> n - \floor{\frac{k + 1}{2}}$ the binomial coefficient is non-zero because its upper index is negative.
Recall that for arbitrary complex number $\alpha$ and non-negative integers $k$, the binomial coefficient  $\binom{\alpha}{k}$ is defined by 
\begin{align}\label{eq:binom-def}
    \sum_{k=0}^\infty \binom{\alpha}{k} x^k = (1+x)^\alpha,
\end{align}
or equivalently, applying the Taylor expansion to the right-hand side at $x=0$, 
\begin{align}\label{eq:binom}
    \binom{\alpha}{k} := \frac{1}{k!}\Big(\frac{\d}{\d x} \Big)^k (1+x)^\alpha \Big|_{x=0} = \frac{(\alpha)_k}{k!}\ ,
\end{align}
where $(\alpha)_k = \alpha(\alpha-1) \cdots (\alpha-k+1)$ is the falling factorial. Therefore, $\binom{\alpha}{k}$ is zero if and only if $0$ appears in the falling factorial, which is the case when $\alpha$ is a non-negative integer and $k>\alpha$. 

We saw that $\{F_k\}$ operators have a restricted support relative to $\{\Pi_w\}$ basis. Similar to $\{A_k\}$ operators, a remarkable property of these operators is that they also have restricted support in the $\{C_l\}$ basis. (This makes them useful for applications in the context of $t$-design considered in this paper, as well as characterization of $k$-body interactions considered in \cite{zhukas2024observation}.) In particular, in \cref{app:U1} we prove that   
\begin{restatable}{lemma}{lemFk}\label{lem:Fk}
On a system with $n$ qubits, the operators $F_k$ defined in \cref{eq:Fk} has the following expansion in the $C_l$ basis: 
\begin{align}
  F_k &=
  \begin{cases}\label{Eq:ClFk}
      2^{k-n} \sum\limits_l \binom{\lfloor l/2\rfloor}{\lfloor k/2\rfloor} C_l & \text{$k:$  even} \\
      2^{k-n} \sum\limits_{l:\mathrm{odd}} \binom{\lfloor l/2\rfloor}{\lfloor k/2\rfloor} C_l & \text{$k:$  odd }.
  \end{cases}
\end{align}
Furthermore,
\begin{align}\label{eq:Fk-norm}
  \| F_k \|_1 &= 
  \begin{cases}
      2^k \binom{n/2}{k/2} & \text{$k:$ even} \\
      2^k \binom{(n-1)/2}{(k-1)/2} & \text{$k:$  odd}\ .
  \end{cases}
\end{align}
\end{restatable}
Note that while for odd $n-k$  the upper indices of the binomial coefficients in \cref{eq:Fk-norm} become half-integer,  and therefore the binomial coefficients are not integer, the norm $\| F_{k} \|_1 $ is always integer (See \cref{eq:binom} for the general definition of binomial coefficient).

An immediate corollary of \cref{Eq:ClFk} is that 
\begin{corollary}\label{cor:Fk-unique}
   The operator $F_k$ is orthogonal to all $(k-1)$-local operators. Furthermore, together with \cref{eq:Fk-kernel}, this condition uniquely determines the operator $F_k$, up to a normalization.
\end{corollary}
\begin{proof}
The first part of this corollary follows from \cref{lem:Fk}, which implies $F_k$ is orthogonal to all $C_l$ for $l<k$, i.e.,
\begin{align}\label{eq:TrFkCl-l<k}
    \Tr(F_k C_l)=0\ \ : l<k\ .
\end{align} 
To see this note that $\binom{\lfloor l/2\rfloor}{\lfloor k/2\rfloor} = 0$ when $\lfloor l/2\rfloor < \lfloor k/2\rfloor$. Combining this with the orthogonality of the $C_l$ operators in \cref{eq:Cl-orthogonal-U1}, we immediately obtain
\cref{eq:TrFkCl-l<k} for $l\le k-2$. To see why this equation also holds for $l=k-1$, we consider even $k$ and odd $k$ separately. In the case of even $k$, the inequality $\lfloor l/2\rfloor < \lfloor k/2\rfloor$ is equivalent to $l<k$, which again combined with  \cref{eq:Cl-orthogonal-U1}  implies \cref{eq:TrFkCl-l<k}. On the other hand, when $k$ is odd,  $l=k-1$ is even and, therefore, does not show up in the expansion in \cref{Eq:ClFk}. Then, from the completeness relation in \cref{eq:Cl-complete-U1}, we know that this equation implies  $\Tr(F_k O) = 0$ for all $k-1$-local operators $O$, as claimed in the above corollary.

Next, we show the second statement in this corollary. Let $\Pi^\Delta=\sum_{w\in\Delta} \Pi_w$ be the projector to a subset of arbitrary $k+1$ Hamming weights $\Delta\subseteq \{0,\cdots, n\}$. Then, the discussion around \cref{eq:TM-Delta-U1} implies that for any $\Delta$, the operators $C_l \Pi^\Delta : l=0,\cdots, k-1$ span a $k$-dimensional subspace of the $k+1$-dimensional space $\mathrm{span}\{\Pi_w: w\in \Delta\}$. It follows that, up to a normalization, there exists a unique non-zero operator $B\in \mathrm{span}\{\Pi_w: w\in \Delta\}$  satisfying $\Tr(C_l B)=0$ for all $l<k$. Now suppose we choose $\Delta$ to be the set of Hamming weights that are not constrained by \cref{eq:Fk-kernel}. That is, the set of Hamming weights $w\le \floor{\frac{k}{2}}$ and  $w\ge n - \floor{\frac{k-1}{2}}$, which correspond to exactly $k+1$ irreps.  Then, the above argument implies that, up to a normalization, there is a unique operator satisfying these constraints, namely $F_k$. 
\end{proof}

%To see the second statement in this corollary, we note that the discussion around \cref{TT} implies that $C_l: l=0, \cdots, k-1$ are linearly independent on any $k$ irreps. 

%Therefore, any operator that is orthogonal to $C_l : l=0, \cdots, k-1$ and has support on certain $k+1$ irreps is unique up to normalization}

\cref{cor:Fk-unique} implies that $F_{k+1}$ is orthogonal to all $k$-local operators, and, in particular, is traceless. Furthermore,  its eigenvalues are integers. Therefore, we can apply \cref{thm:tdesign} for  $Q=F_{k+1}$, which implies that $t_{\max}$ for the uniform distribution over $\mathcal{V}_{n, k}^{\U(1)}$ satisfies
\begin{align}\label{eq:tmax-upper-U1}
  t_{\max} + 1 \le \frac{1}{2} \|F_{k + 1}\|_1 \ .
\end{align}
In \cref{app:U1}, we prove that this bound is tight, i.e., 
\begin{align}\label{eq:tmax-U1}
  t_{\max} + 1 = \frac{1}{2} \|F_{k + 1}\|_1 
  =
  \begin{cases}
    2^k \binom{(n-1)/2}{k/2} \text{ $k: $ even} \\
    2^k \binom{n/2}{(k+1)/2} \text{ $k: $  odd} 
  \end{cases},
\end{align}
as long as
\begin{equation}\label{eq:nbound-U1}
  n \geq 2^{\floor{\frac{k}{2}}} \floor{\frac{k + 3}{2}}\ .
\end{equation}
The explicit expressions of $t_{\max}$ for some small $k$ can be found in \cref{tab:tmax-example}. Recall that for $n\gg k$, $\binom{n}{k} \sim \frac{n^k}{k!}$.\footnote{More precisely, here $\sim$ notation means $\lim\limits_{n/k \rightarrow \infty} \binom{n}{k}\big/\frac{n^k}{k!} = 1$.} Therefore, assuming $k$ is bounded by the upper bound given in \cref{eq:nbound-U1}, we find that in the  limit $n \rightarrow \infty$, the asymptotic behavior of $t_{\max}$ is given by 
\begin{align}
  t_{\max} \sim \frac{2^{\lfloor\frac{k}{2}\rfloor}}{\lfloor \frac{k+1}{2}\rfloor!} n^{\lfloor\frac{k+1}{2}\rfloor}\ .
\end{align}
We note that Ref. \cite{hearth2023unitary} has previously shown that for $k\ge 2$, $t_{\max}$ grows, at least, linearly with $n$. Our result reveals that the actual value of $t_{\max}$ grows as $n^{\lfloor\frac{k}{2}+1\rfloor}$. \\

\noindent{\textbf{Matrix $\mathbf{S}$}:} Finally, we present an example calculation of the matrix $\mathbf{S}$ defined in \cref{eq:S-def}. For Hamming weights $v = 0, \cdots, k$ and $w = 0, \cdots, n$, we have
\begin{align}\label{eq:S-U1}
    s_v(w) & = \frac{1}{2 \pi} \int_0^{2 \pi} \diff \theta \, (\tr \e^{\i \theta Z})^{n - k} \e^{\i \theta (k - 2 v)} \e^{-\i \theta (n - 2 w)} \nonumber\\
          & = \frac{1}{2 \pi} \int_0^{2 \pi} \diff \theta \, (\e^{\i \theta} + \e^{-\i \theta})^{n - k} \e^{-\i \theta(n - k - 2(w - v))} \nonumber\\
          & = \binom{n - k}{w - v},
\end{align}
which we remark is zero unless $0 \leq w - v \leq n - k$.  While it is not immediately obvious from these matrix elements, \cref{lem1} implies that this matrix has rank $k+1$. Note that as expected from \cref{multi}, $s_v(w)$  is the multiplicity of the sector with Hamming weight $w-v$ in a system with $n-k$ qubits.

\begin{figure}[htb]
  \centering
  \includegraphics[width=0.4\textwidth]{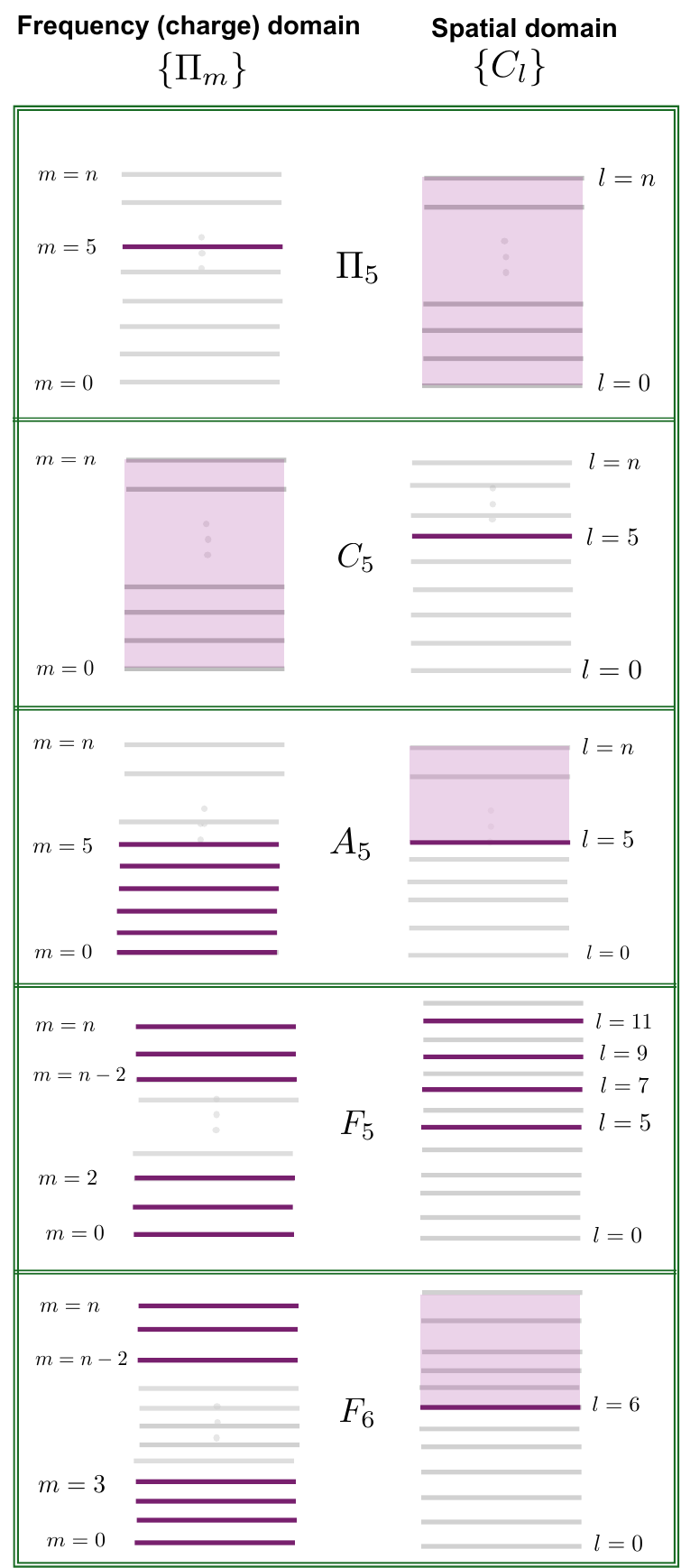}
  \caption{\textbf{Frequency versus spatial domains} The representation of operators $\Pi_5$, i.e., the projector to the sector with Hamming weight $5$, $C_5$ defined in \cref{eq:u1cls}, $F_5$ and $F_6$ defined in \cref{eq:Fk}, and $\{A_5\}$ defined in \cref{Ak}, in the $\{\Pi_m\}$ basis (left) and $\{C_l\}$ basis (right). The purple lines indicate the support of the operator. }
  \label{domains}
\end{figure}

\begin{figure}[htb]
  \centering
  \includegraphics[width=0.93\columnwidth]{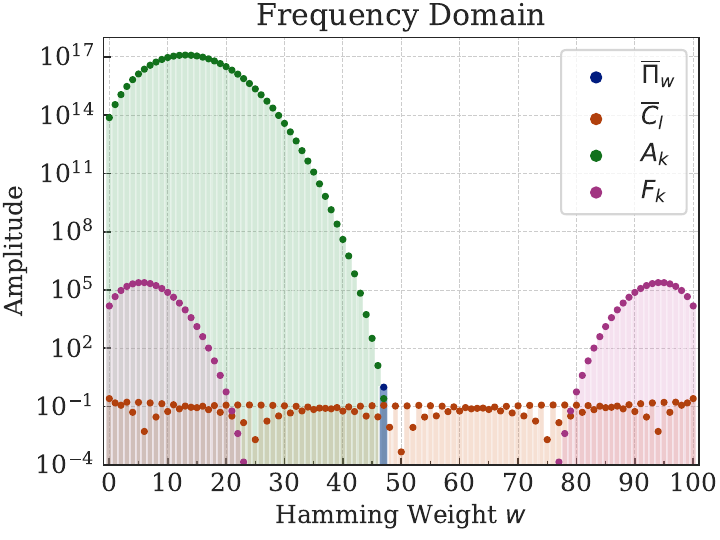}\vspace{5mm}
  \includegraphics[width=0.93\columnwidth]{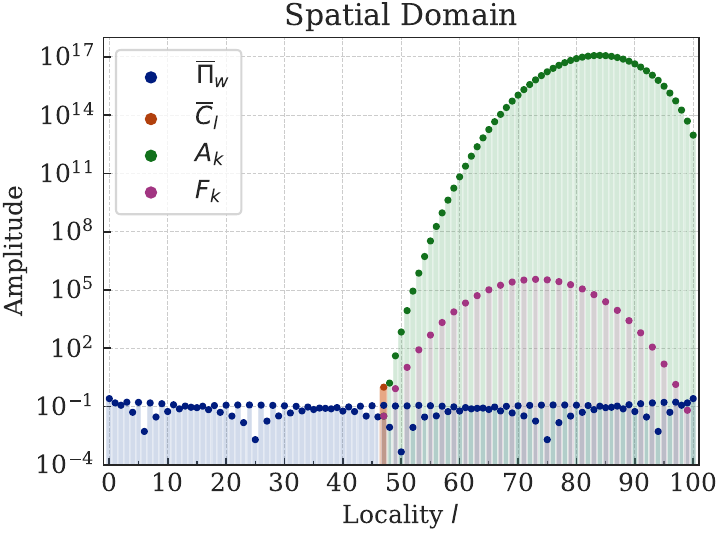}
  \caption{\textbf{Amplitudes of operators $\overline\Pi_w$, $\overline C_l$, $A_k$ and $F_k$ in the frequency and spatial domains for $w=l=k=47$ and $n=100$.} Recall that $\{\overline C_l\}$ and $\{\overline\Pi_w\}$ are both orthogonal bases for the center of the space of $\U(1)$-invariant Hermitian operators. $\{A_k\}$ and $\{F_k\}$ are also two (non-orthogonal) bases for this space. The top and bottom plots represent the absolute values of coefficients of the expansion of these operators in $\{\overline\Pi_w\}$ and $\{\overline C_l\}$ bases, respectively.   More precisely, the top plot represents $|o_{w}| = 2^{-n} \abs{\Tr(\overline\Pi_w O)}$ for the expansion of any operator $O=\sum_w o_w \overline\Pi_w$, while the bottom plot represents $|\tilde{o}_{l}| = 2^{-n} \abs{\Tr(\overline C_l O)}$ for the expansion $O=\sum_l \tilde{o}_l \overline C_l$.   Note that while $\overline\Pi_w$ is by definition localized in the frequency domain, it is spread out in the spatial domain. Similarly,  $\overline C_l$ is localized in the spatial domain but is spread out in the frequency domain.  $A_k$ operator has support only on low Hamming weights $w\leq 47$, $F_k$ operator has support only on low-multiplicity Hamming weights $w\leq 20$ and $w\geq 81$, and both $A_k$ and $F_k$ have support only on high locality $l\geq 47$. }
  \label{fig:dual-domains-U1}
\end{figure}

\subsection*{Spatial and Frequency domains:\\ Relation between 4 bases}

We have introduced 4 different bases for the same space, namely,
\be
\mathrm{span}\{\Pi_m\}=\mathrm{span}\{C_l\}=\mathrm{span}\{A_k\}=\mathrm{span}\{F_k\}\ .
\ee
Both $\{\Pi_m\}$ and $\{C_l\}$ form orthogonal bases with respect to the Hilbert-Schmidt inner product. The basis $\{C_l\}$ defines a sharp notion of locality in the following sense: $\Tr(C_l A)$ is zero for any $k$-local operator with $k<l$. Furthermore, $\Tr(\mathbf{Z}^{\mathbf{b}} C_l)=2^n \delta_{l,\ell(\mathbf{b}) }$, where 
$\ell(\mathbf{b}) := \sum_a b_a$ is the Hamming weight of $\mathbf{b}$. In other words, $\Tr(\mathbf{Z}^{\mathbf{b}} C_l)$ is non-zero only if $l$ is equal to the number of qubits that $Z^{\mathbf{b}}$ act non-trivially on them.

On the other hand, elements of the basis $\{\Pi_m\}$ define a sharp notion of charge, i.e., the irreps of symmetry, which in this case corresponds to the Hamming weight. We note that the relationship between them is similar to that between dual Fourier domains.

Operators $\{A_k\}$ and $\{F_k\}$, on the other hand, are both non-orthogonal bases for this space. 
However, they have another desirable property that is relevant for applications in the context of $t$-designs, as discussed in this paper,  as well as for experimental characterization of $N$-body interactions, discussed in \cite{zhukas2024observation}: namely, they have restricted support relative to both $\{C_l\}$ basis, which can be interpreted as the spatial domain, as well as restricted support relative to  $\{\Pi_m\}$ basis, which corresponds to the charge (frequency) domain. In particular, for $l<k$, we have $\Tr(C_l A_k)=\Tr(C_l F_k)=0$, which via \cref{eq:Cl-complete-U1} implies for any $(k-1)$-local operator $O$, we have
\be
\Tr(O A_k)=\Tr(O F_k)=0\ .
\ee
Furthermore, they have limited supports in the charge space, namely $\Tr(\Pi_w A_k)=0$ for $w>k$, and $\Tr(\Pi_w F_k)=0$ for $\floor{\frac{k}{2}} + 1 \leq w \leq n - \floor{\frac{k+1}{2}}$. Roughly speaking, this means that while the support of operators $A_k$ is restricted to sectors with Hamming weight $\le k$, operators $F_k$ have support in both low and high Hamming weight sectors. 
As an example, \cref{domains} schematically presents operators $\Pi_5, C_5, A_5, F_5$, and $F_6$ in these two dual bases. 

%It is natural to rescale the two orthogonal bases $\{\Pi_w\}$ and $\{C_l\}$ to $\{\overline\Pi_w\}$ and $\{\overline C_l\}$ such that

To obtain orthonormal bases from $\{\Pi_w\}$ and $\{C_l\}$, it is useful to define their rescaled versions 
$\overline\Pi_w:=\sqrt{\frac{2^n}{\Tr \Pi_w}} \Pi_w $ and $\overline C_l:=
\sqrt{\frac{2^n}{\Tr C^2_l}} C_l$ that satisfy
\begin{align}
  \Tr\overline\Pi_w^2 = \Tr\overline C_l^2 = \Tr\id = 2^n,
\end{align}
for all $w, l = 0, \cdots, n$, which implies $\{2^{-\frac{n}{2}}\overline\Pi_w\}$ and $\{2^{-\frac{n}{2}} \overline C_l\}$ are orthonormal. Furthermore, since all elements of these bases are Hermitian operators, the elements of the two bases are related via an orthogonal transformation, namely for $w,l=0,\cdots, n$, 
\begin{subequations}\label{eq:reciprocity-U1}
  \begin{align}
   \overline C_l &= \sum_{w=0}^n \bar{c}_{l,w} \overline\Pi_w\ , \label{eq:reciprocity-U1a}\\
   \overline\Pi_w &= \sum_{l=0}^n \bar{c}_{l,w} \overline C_l\ , \label{eq:reciprocity-U1b}
  \end{align}
\end{subequations}
where $\bar{c}$ is an $(n+1)\times(n+1)$ matrix with elements
\begin{align}\label{eq:c-lw}
  \bar{c}_{l,w} := 
  \sqrt{\frac{\Tr(\Pi_w)}{\Tr(C_l^2)}} \times c_l(w) ,
\end{align}
and $c_l(w)$ is defined in \cref{Eq:clw}. \cref{eq:reciprocity-U1a} follows immediately from \cref{eq:u1cls}, i.e.,  the fact that $c_l(w)$ is
the eigenvalues of $C_l$ in the subspace with Hamming weight $w$, and the second equation follows from the fact that matrix $\bar{c}$ preserves the Hilbert-Schmidt inner product and therefore, is orthogonal.

Remarkably, as we show in 
\cref{app:U1}, it turns out that in addition to being an orthogonal matrix, $\bar{c}$ is also symmetric and therefore involutory, i.e.,
\begin{align}\label{eq:cbar-symmetric-orthogonal}
  \bar{c} = \bar{c}^\T = \bar{c}^{-1}\ .
\end{align}

In \cref{fig:dual-domains-U1}, we plot the components of $\overline\Pi_{47}$, $\overline C_{47}$, $A_{47}$, and  $F_{47}$ in the  bases $\{\overline\Pi_w\}$ and $\{\overline C_l\}$ respectively for $n=100$ qubits. While $\overline\Pi_w$ is obviously localized in the frequency domain, namely the $\{\overline\Pi_w\}$ basis, it is delocalized in the spatial domain, i.e., $\Tr(\overline\Pi_w \overline C_l) =2^n \bar{c}_{l,w}$ is non-zero 
for generic values of $w$ and $l$. 
Similarly, $\overline C_l$ is localized in the spatial domain but not in the frequency domain. 

There are two additional features to notice in \cref{fig:dual-domains-U1}: (i) the orange dots in the first plot and the blue dots in the second are identical: this is a consequence of $\bar{c}$ being a symmetric matrix. (ii) the orange dots and the blue dots themselves are symmetric with respect to the mirror reflection about the vertical lines defined by $w=50$ and $l=50$, respectively. This is a consequence of  
\begin{align}\label{eq:mirror-symmetry}
\Tr(\Pi_w C_l)=(-1)^l \Tr(\Pi_{n-w} C_l)=(-1)^w \Tr(\Pi_{w} C_{n-l}) \ ,
\end{align}
that are shown in \cref{app:U1}.% using relations \cref{eq:Xn,eq:Zn}.

%Finally, in \cref{app:U1} we also consider a  slightly different basis, namely  $\widetilde{F}_k=(-1)^k X^{\otimes n}F_k X^{\otimes n}$, which has similar properties as $F_k$ and gives the same result on $t_{\max}$. 

\subsection{\texorpdfstring{$\SU(2)$}{SU(2)} symmetry}\label{Sec:SU(2)}

Next, we consider a system with $n$ qubits with on-site $\SU(2)$ symmetry. That is, we consider $n$-qubit unitaries commuting with $u^{\otimes n}: u\in\SU(2)$.  The  irreps of $\SU(2)$ can be labeled by the eigenvalues $j(j+1)$ of the total angular momentum operator $J^2 = J_x^2 + J_y^2 + J_z^2$, where  $J_v = \frac{1}{2} \sum_{a = 1}^n \sigma_a^v$ is the angular momentum operator in direction $v = x, y, z$, $j= j_{\min}, j_{\min}+1,\cdots, \frac{n}{2}$, where $j_{\min}=\frac{n}{2} - \lfloor\frac{n}{2}\rfloor$. Recall that $\SU(2)$ irrep with angular momentum $j$ has dimension $2j+1$. Furthermore,  
the multiplicity of this irrep is 
\begin{align}
  m(n,j) = \binom{n}{\frac{n}{2}-j} \frac{2j+1}{\frac{n}{2}+j+1}\ ,
\end{align}
which as we show in \cref{app:SU2-dim} is strictly decreasing with $j$ for $j > \sqrt{n\ln n}$.

Our previous work in Ref. \cite{Marvian2024Rotationally} proves that semi-universality is achieved with 2-qubit $\SU(2)$-invariant gates, i.e., $\mathcal{SV}_{n,n}^{\SU(2)}\subset \mathcal{V}_{n,2}^{\SU(2)}$. In this section, we use \cref{prop1,thm:tdesign} to first establish a lower bound, and then find the exact value of $t_{\max}$ for the Haar measure over  $\mathcal{V}_{n,k}^{\SU(2)}$ for general $k\ge 2$. We start with the example of $k=2$.

\subsubsection{Example: \texorpdfstring{$k=2$}{k=2}}

Before studying the group $\mathcal{V}_{n,2}^{\SU(2)}$ generated by 2-qubit gates, it is instructive to consider the design properties of the group $\mathcal{SV}_{n,n}^{\SU(2)}$, or equivalently, the group $\mathcal{SV}_{n,n}^{\SU(2)}$ together with global phases $\e^{\i\theta} \mathbb{I}: \theta\in[0,2\pi)$ (recall that including global phases does not change the design properties of the group).  

In this case, applying \cref{prop1}, we can immediately obtain a simple lower bound on $t_{\max}$: recall that the irrep with the lowest multiplicity is the highest angular momentum irrep with $j_{\max}=\frac{n}{2}$, which has multiplicity 1. Furthermore, as $j$ decreases from its maximum value $j_{\max}=\frac{n}{2}$ down to $j \approx \sqrt{n\ln n}$, the multiplicity $m(n,j)$ is monotonically increasing. In particular, when $n \geq 5$, the irrep with the second lowest multiplicity corresponds to $j=\frac{n}{2}-1$ whose dimension is $n-1$. Therefore, from \cref{prop1}, we know that $t_{\max} \geq n-2$. 

It turns out that this bound holds as equality, i.e., $t_{\max}=n-2$, which can be shown using  \cref{thm:tdesign} to find an upper bound on $t_{\max}$.  In particular, we choose the operator
$Q$ in this theorem to be $Q=A_2$, where
\begin{align}
  A_2 = \frac{n-1}{n+1}\Pi_{\frac{n}{2}} - \frac{1}{n-1}\Pi_{\frac{n}{2}-1}\ .
\end{align}
As we further explain in  \cref{corSU(2)}, this operator is a special case of a family of operators $\{A_k\}$ defined below in \cref{eq:Ak}, which are uniquely specified (up to normalization) by the property that they are orthogonal to all $(k-1)$-local operators, and have support only on $\{\Pi_j\}$ with $j\ge  \frac{n-k}{2}$. Furthermore, operator $Q=A_2$ also satisfies the other conditions in \cref{thm:tdesign}, namely $\Tr(A_2)=0$, and $q(j=\frac{n}{2})=n-1$ and $q(j=\frac{n}{2}-1)=-1$. Therefore, this theorem implies the upper bound $t_{\max} \leq \frac{1}{2} \|A_2\|_1 -1= n-2$. Combining these lower and upper bounds we conclude that $t_{\max} = n-2$.

Next, we consider the group $\mathcal{V}_{n,2}^{\SU(2)}$, generated by 2-qubit $\SU(2)$-invariant gates, namely 
\begin{align}
  \mathcal{V}_{n,2}^{\SU(2)}\!=\langle \e^{\i \theta R_{ab}}, \e^{\i\theta} \id: 0\le a <b \le n, \theta\in[0,2\pi)\rangle ,
\end{align}
where $R = \frac{1}{2} (X \otimes X + Y \otimes Y + Z \otimes Z)$. 
It is worth noting that this group is equal to the group generated by 3-qubit SU(2)-invariant unitaries, that is  $\mathcal{V}_{n,2}^{\SU(2)}=\mathcal{V}_{n,3}^{\SU(2)}$ \cite{Marvian2024Rotationally} (In other words,  any 3-qubit $\SU(2)$-invariant unitary is realizable with 2-qubit $\SU(2)$-invariant unitaries).

For $H_0 = \id$ and $H_1 = R_{ab}$, projected to sectors with angular momenta $j=\frac{n}{2}, \frac{n}{2}-1$, the matrix $\mathbf{M}^{\Delta}$, 
\begin{equation}
  \mathbf{M}^\Delta = \begin{pmatrix}
    1 & n - 1 \\ 
    \frac{1}{2}n(n-1) & \frac{1}{2}n(n-1)(n-5)
  \end{pmatrix}
\end{equation}
has full rank. Therefore, \cref{prop1} implies that $t_{\max} +1 \geq m(n, \frac{n}{2}-2) = \frac{1}{2} n (n - 3)$.

Similar to the previous case,  the upper bound of $t_{\max}$ can be calculated using $Q=A_4$ defined as
\begin{align}
  A_4 = \frac{(n - 2)(n - 3)}{2(n+1)} \Pi_{0} - \frac{n-3}{n-1} \Pi_{1} + \frac{1}{n-3} \Pi_{2},
\end{align}
which satisfies all conditions in \cref{thm:tdesign}.

Therefore, we have $t_{\max} \leq\frac{1}{2}\|A_4\|-1$. Furthermore, using an argument similar to the one in \cref{jdj}, and the fact that for $n \geq 9$, $\frac{1}{2}\|A_4\| \leq m(n, \frac{n}{2}-3)$, in \cref{app:SU2}, 
we show that this is indeed the optimal operator that satisfies the conditions of the theorem, i.e.,
\begin{align}\label{eq:SU2-tmax-k=2}
  t_{\max} =\frac{1}{2}\|A_4\|-1 = (n - 1)(n - 3)-1\ .
\end{align}

\subsubsection{General \texorpdfstring{$k$}{k}}

Recall that in the case of U(1) symmetry, $\{C_l\}$ operators played a crucial role in understanding the constraints imposed by locality.   
 In the case of $\SU(2)$ symmetry 
Ref. \cite{Marvian2024Rotationally} introduces an orthogonal basis for $\mathrm{span}\{\Pi_j\}$, with similar properties. For even $l=0,\cdots, 2 \lfloor \frac{n}{2} \rfloor$, this basis is defined as\footnote{Equivalently, as discussed in \cite{Marvian2024Rotationally} this basis can be defined as a linear combination of permutation operators (See also \cite{kazi2024universality}).}
\begin{align}\label{eq:Cl-SU2}
  C_l = \frac{1}{(l/2)!} \sum_{i_1 \neq \cdots \neq i_l} R_{i_1i_2} \cdots R_{i_{l-1}i_l}=\sum_{j=j_{\min}}^{n / 2}  c_l(j)\ \Pi_j\ ,
\end{align}
where, as shown in \cite{Marvian2024Rotationally}, 
the eigenvalues $c_l(j)$ are remarkably all integers.  See \cref{eq:Cl-expression} in \cref{app:SU2} for the explicit forms of the eigenvalues $c_l(j)$ for arbitrary $l$.  For example,
\begin{align}
  C_2 &= 2J^2 - \frac{3}{2}n \mathbb{I} = \sum_{j=j_{\min}}^{n / 2} \big[2j(j+1) - \frac{3}{2}n\big]\  \Pi_j\ .
\end{align}
Similar to the  $C_l$ operators defined in \cref{eq:u1cls} for the case of $\U(1)$ symmetry, the  $C_l$ operators in \cref{eq:Cl-SU2} for $\SU(2)$ symmetry are also Hermitian and and  
satisfy the orthogonality relation\footnote{The double factorial $n!!$ is defined as the product of all the positive integers up to $n$ that have the same parity, i.e., $n!! := n(n-2)\cdots$.}  
\bes
\begin{align}\label{eq:Cl-orthogonal-SU2}
    \Tr(C_l C_{l'}) &=\sum_{j=j_{\min}}^{n / 2}  c_l(j) c_{l'}(j) \Tr(\Pi_j)\\ &= \delta_{ll'}\times  (l+1)!!\times  (l-1)!! \times 2^n \binom{n}{l}\ ,
\end{align}
\ees
and therefore, form a complete basis for the space $\mathrm{span}\{\Pi_j\}=\mathrm{span}\{C_l\}$ which has dimension $\lfloor \frac{n}{2} \rfloor+1$  (See \cref{lem:C2msquare} in \cref{app:SU2}  for the proof of this identity). These properties imply that for any arbitrary operator $B$ in this space, we have the completeness relation
\begin{align}\label{eq:Cl-complete-SU2}
    \Tr(B O)=\sum_{l:\mathrm{even}}^n \frac{\Tr(B C_l)\times  \Tr(C_l O)}{\Tr(C_l^2)} ,
\end{align}
where $O$ is an arbitrary operator on $n$ qubits.  Furthermore, the definition of $C_l$ operators in \cref{eq:Cl-SU2}, immediately implies that
for any $k$-local operator $O$,
\begin{align}
    \tr (O C_l) = 0\ \ : l>k\ . 
\end{align}
Note that similar to the U(1) case we can rescale $\Pi_j$ and $C_l$ to $\overline\Pi_j$ and $\overline C_l$ such that $\Tr\overline\Pi_j = \Tr\overline C_l^2 = \Tr\id = 2^n$. 
Then, they satisfy the following reciprocity relation
\begin{subequations}\label{eq:reciprocity-SU2}
  \begin{align}
   \overline C_{l} &= \sum_{j=j_\text{min}}^{n/2} \bar{c}_{l,j} \overline\Pi_j\ , \label{eq:reciprocity-SU2a}\\
   \overline\Pi_j &= \sum_{l: \text{even}} \bar{c}_{l,j} \overline C_{l}\ , \label{eq:reciprocity-SU2b}
  \end{align}
\end{subequations}
where $\bar{c}_{l,j}=\sqrt{\frac{\Tr(\Pi_j)}{\Tr(C_l^2)}}c_l(j)$ defines a real orthogonal $(2 \lfloor \frac{n}{2} \rfloor+1)\times (2 \lfloor \frac{n}{2} \rfloor+1)$  matrix. Then, similar to the U(1) case, we can interpret $\overline\Pi_j$ and $\overline C_l$ as the frequency and spatial domains. In \cref{fig:dual-domains-SU2} we present the components of $\overline\Pi_j$, $\overline C_l$ and other operators defined below relative to these bases. Notice that the blue dots in the bottom plot are symmetric with respect to the mirror reflection about the $l=50$ axis. As we discussed in \cref{app:SU2}, this is a consequence of the relation 
\begin{align}\label{eq:Cbar-symmetry}
  \Tr(\overline{C}_l \Pi_j) = (-1)^{\frac{n}{2}-j} \Tr(\overline{C}_{n-1-l} \Pi_j).
\end{align}
when $n$ is odd.\\

\begin{figure}[htb]
  \centering
  \includegraphics[width=0.93\columnwidth]{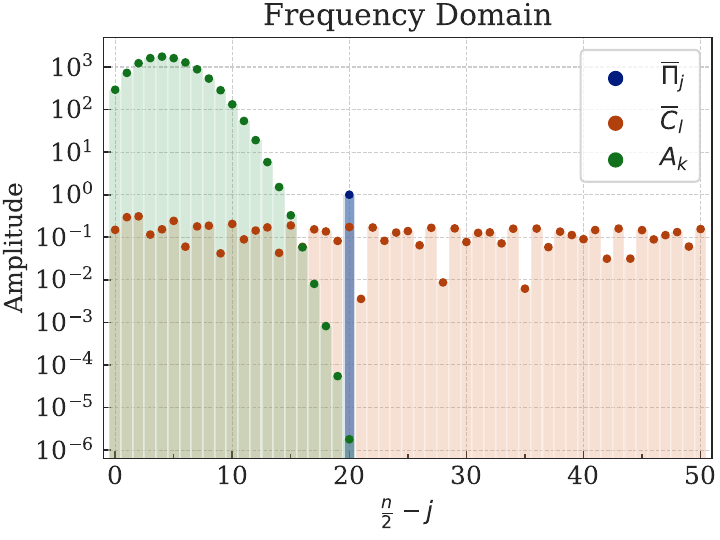}\vspace{5mm}
  \includegraphics[width=0.93\columnwidth]{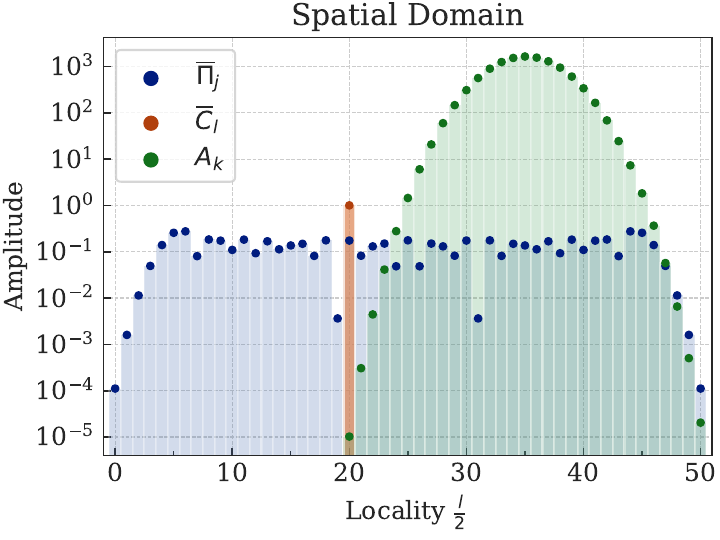}
  \caption{\textbf{Amplitudes of operators $\overline\Pi_j$, $\overline C_l$ and $A_k$ in the frequency and spatial domains for $n-2j=l=k=40$ and $n=101$.} Recall that $\{\overline\Pi_j\}$ and $\{\overline C_l\}$ are both orthonormal bases for the center of the space of $\SU(2)$-invariant Hermitian operators. $\{A_k\}$ is also a (non-orthogonal) basis for this space.  Similar to \cref{fig:dual-domains-U1}, the top and bottom plots represent the absolute values of coefficients of the expansion of these operators in $\{\overline\Pi_j\}$ and $\{\overline C_l\}$ bases, respectively. More precisely,  for any operator $O$, they represent $|o_{w}| = 2^{-n} \abs{\Tr(\overline\Pi_j O)}$ and $|\tilde{o}_{l}| = 2^{-n} \abs{\Tr(\overline C_l O)}$, respectively.  Note that $\overline\Pi_j$ is localized in the frequency domain but is spread out in the spatial domain, whereas $\overline C_l$ is localized in the spatial domain but is spread out in the frequency domain. $A_k$ operator has support only on high angular momenta $j\geq 30$ and high locality $l\geq 40$.}
  \label{fig:dual-domains-SU2}
\end{figure}

\noindent{\textbf{Lower bound on $t_{\max}$}:} We start by establishing a lower bound on $t_{\max}$ by applying \cref{prop1}. In particular, we show that for the group $\mathcal{V}^{\SU(2)}_{n, k}$, the condition $\rk \mathbf{M}^\Delta = \abs{\Delta}$ holds for any subset of irreps $\Delta$ with $\abs{\Delta} = \lfloor\frac{k}{2}\rfloor + 1$.

Similar to the case of $\U(1)$ symmetry, operators     $\{C_{l}: l = 0, 2, \cdots, 2 \floor{\frac{k}{2}}\}$
form a basis of the center of the Lie algebra of $\mathcal{V}_{n,k}^{\SU(2)}$. For any subset of $\floor{\frac{k}{2}} + 1$ angular momenta $\Delta \subseteq \set{j_{\min}, \cdots, \frac{n}{2}}$, consider the matrix
\begin{equation}
    \mathbf{M}_{j,l}^\Delta = \tr \Pi_j C_l = c_l(j) \binom{n}{\frac{n}{2}-j} \frac{2j+1}{n/2+j+1},
\end{equation}
where $l$ is an even integer in the interval $0 \leq l \leq k$, $j \in \Delta$, and $c_l(j)$ is the eigenvalue of $C_l$ on the irrep $j$. Similar to the $\U(1)$ case, one can argue that $\mathbf{M}^\Delta$ has full rank using the fact from Ref. \cite{Marvian2024Rotationally} that
\begin{equation}
\begin{split}
  \operatorname{span}_\real & \set{C_{l} : l = 0, 2, \cdots, 2 \floor{\frac{k}{2}}} \\
  & = \operatorname{span}_\real \set{C_2^m : m = 0, \cdots, \floor{\frac{k}{2}}}.
\end{split}
\end{equation}
Since $\rk(\mathbf{M}^\Delta) = |\Delta|$ for any subset of irreps $\Delta$ with $|\Delta| = \floor{\frac{k}{2}}+1$, we know that the $(\floor{\frac{k}{2}}+2)$th smallest multiplicity determines the lower bound. Recall that 
as $j$ decreases from its maximum value 
$j_{\max}=\frac{n}{2}$ down to $j \approx \sqrt{n\ln n}$,  the multiplicity $m(n,j)$ is monotonically increasing. Therefore, for $k$ satisfying  $$\frac{n}{2}-\Big(\floor{\frac{k}{2}}+1\Big) < \sqrt{n\ln n},$$
applying \cref{prop1}  we obtain the lower bound  
\begin{align}
  t_{\max} +1 \geq m(n, \frac{n}{2}-\floor{\frac{k}{2}}-1) \!=\! \binom{n}{\floor{\frac{k}{2}}+1} \frac{n-2\floor{\frac{k}{2}}-1}{n-\floor{\frac{k}{2}}} .
\end{align}

\noindent{\textbf{Exact value of $t_{\max}$}:} Next, we find the exact value of $t_{\max}$. For this purpose, we introduce a new basis for the space $\mathrm{span}_\mathbb{R}\{\Pi_j\}$,  which is analogous to the basis $F_k$ in \cref{eq:Fk}, defined in the case of U(1) symmetry, and has similar properties. Namely, for even integers $k=0, 2, \cdots, 2\lfloor \frac{n}{2}\rfloor$, we define operators 
\begin{align}\label{eq:Ak}
  A_k 
  &= \sum_{j=j_{\min}}^{n/2} a_k(j) {\Pi_j} \nonumber\\
         &= \sum_{j=j_{\min}}^{n/2} (-1)^{\frac{n}{2}-j} \binom{\frac{n-k}{2}+j} {n-k} m(n,j) \frac{\Pi_j}{\Tr\Pi_j}\ .
\end{align}
A few remarks are in order: First, $a_k(j)$ is the eigenvalue of $A_k$ in the sector with angular momentum $j$. Recall that $\Tr\Pi_j = (2j+1) \times  m(n,j) $, and therefore
\be
a_k(j)\times (2j+1)=(-1)^{\frac{n}{2}-j} \binom{\frac{n-k}{2}+j} {n-k} \ . 
\ee
Since $k$ is even and $\frac{n}{2}+j$ is an integer, $a_k(j)\times (2j+1)$ is always an integer.  This means that if we choose operator $Q=\sum_{\lambda} q(\lambda) m_\lambda \frac{\Pi_\lambda}{\tr \Pi_\lambda}$ in \cref{thm:tdesign} equal to $A_k$,  the corresponding coefficients $q(\lambda)$ are integers, as required by this theorem. Second, notice that when $j < \frac{n-k}{2}$ the coefficient of $\Pi_j$ is zero. In other words, $A_k$ has restricted support in the $\{\Pi_j\}$ basis,
\begin{align}\label{eq:As-support}
  \Tr(A_k \Pi_j) = 0 \quad \text{when } j < \frac{n-k}{2}\ .
\end{align}
Similar to the $\U(1)$ case, the operator $A_k$ also has the remarkable property that it also has restricted support in the $C_l$ basis (See the example in \cref{fig:dual-domains-SU2}). In \cref{app:SU2}, we prove that
\begin{restatable}{lemma}{lemAk}\label{lem:Ak}
On a system with $n$ qubits, operators $A_k: k=0, 2,\cdots, 2 \lfloor \frac{n}{2} \rfloor$ defined in \cref{eq:Ak} have expansion 
\begin{align}\label{Eq:A2sC2m}
  A_k = 2^{k-n} \sum_{l:\mathrm{even}} \frac{1}{(l+1)!!} \binom{l/2}{k/2} C_{l}\ , 
\end{align}
relative to the basis $C_{l}:  l=0, 2,\cdots, 2 \lfloor \frac{n}{2} \rfloor$.  Furthermore,
\begin{align}\label{ee}
  \| A_k \|_1 = 2^k \binom{(n-1)/2}{k/2}\ .
\end{align}
\end{restatable}
Note that similar to \cref{eq:Fk-norm} 
in 
the case of U(1) symmetry, here for even $n$  the upper index of the binomial coefficients  in  \cref{eq:Fk-norm} becomes half-integer. However, the norm $\| A_k \|_1 $ is always an integer.

An immediate corollary of \cref{Eq:A2sC2m} is that
\begin{corollary}\label{corSU(2)}
    The operator $A_k$ is orthogonal to all $(k-1)$-local operators. Furthermore, together with \cref{eq:As-support}, this condition uniquely determines the operator $A_k$, up to a normalization.
\end{corollary}

\begin{proof}
The proof of this corollary is similar to \cref{cor:Fk-unique}. In particular, $A_k$ is orthogonal to $C_{l}$ for $l<k$, i.e.,
\begin{align}
    \Tr(A_k C_{l}) = 0 : \quad l<k\ ,
\end{align}
because $\binom{l/2}{k/2} = 0$ when $0 \leq l<k$. Then, using the orthogonality of the $C_l$ operators and \cref{eq:Cl-complete-SU2}, this equation implies that $\Tr (A_k O)=0$ for all $(k-1)$-local operators $O$. To show that $A_k$ is unique up to normalization, we consider the projector $\Pi^\Delta = \sum_{j \in \Delta} \Pi_j$ to a subset of arbitrary $\frac{k}{2} + 1$ angular momenta $\Delta \subseteq \set{j_{\min}, \cdots, j_{\max}}$. Then, for any such $\Delta$, the operators $C_l \Pi^\Delta : l = 0, 2, \cdots, k-2$ span an $\frac{k}{2}$-dimensional subspace of the $(\frac{k}{2} + 1)$-dimensional space $\operatorname{span}\set{\Pi_w : w \in \Delta}$. It follows that, up to normalization, there is a unique nonzero operator in this span that is orthogonal to $C_l$ for all even $l < k$.  Choosing $\Delta = \set{j : j \geq \frac{n-k}{2}}$ implies $A_k$ is unique.
\end{proof}

Next, we apply these results to find $t_{\max}$ for the uniform distribution over 
the group generated by $k$-qubit SU(2)-invariant unitaries. Combining the above corollary with \cref{eq:Cl-complete-SU2}, we find that operator $A_{2\lfloor\frac{k}{2}\rfloor+2}$ is orthogonal to all $k$-local operators. Furthermore, $Q = A_{2\lfloor\frac{k}{2}\rfloor+2}$ satisfies all the other assumptions of \cref{thm:tdesign},  which implies 
\begin{align}
  t_{\max} + 1 \leq \frac{1}{2} \|A_{2\lfloor\frac{k}{2}\rfloor+2}\|_1 = 2^{2\lfloor{\frac{k}{2}}\rfloor+1} \binom{(n-1)/2}{\lfloor{\frac{k}{2}}\rfloor+1} \ ,
\end{align}
where the equality follows from \cref{lem:Ak}. Since $q(\frac{n-k}{2})= (n-k+1)\times a_k(\frac{n-k}{2}) = \pm 1$, the operator $Q = A_{2\lfloor\frac{k}{2}\rfloor+2}$ is the optimal solution with support restricted to  $j = \frac{n-k}{2} - 1, \cdots, \frac{n}{2}$.

In \cref{app:SU2}, we prove that, for 
\begin{equation}\label{eq:nbound-SU2}
  n \geq
  \begin{cases}
      13, & \text{when }  k=2,3\\
      2^{\lfloor{\frac{k}{2}}\rfloor} (\lfloor{\frac{k}{2}}\rfloor+2)\ , & \text{when } k\ge 4, \\
  \end{cases}
\end{equation}
this bound is tight, i.e.,
\begin{align}\label{eq:tmax-SU2}
  t_{\max} = 2^{2\lfloor{\frac{k}{2}}\rfloor+1} \binom{(n-1)/2}{\lfloor{\frac{k}{2}}\rfloor+1}-1\ .
\end{align}
We demonstrate this by comparing with the next smallest multiplicity $m(n, j)$ for $j= \frac{n}{2}-\lfloor{\frac{k}{2}}\rfloor- 2$.
Note that \cref{eq:SU2-tmax-k=2} is a special case of \cref{eq:tmax-SU2} for $\lfloor{\frac{k}{2}}\rfloor=1$. 

The explicit expressions of $t_{\max}$ for some small $k$ can be found in \cref{tab:tmax-example}. Assuming $k$ is bounded by the upper bound given in \cref{eq:nbound-SU2}, we find that in the  limit $n \rightarrow \infty$, the asymptotic behavior of $t_{\max}$ is given by 
\begin{align}
  t_{\max} \sim \frac{(2n)^{\lfloor{\frac{k}{2}}\rfloor+1}}{(2\lfloor{\frac{k}{2}}\rfloor+2)!}\ .
\end{align}
In other words,  $t_{\max}$ grows as $ n^{\lfloor\frac{k}{2}\rfloor+1}$. This is consistent with a previous result for the specific case of $k=4$, shown in the appendix of Ref. \cite{Li:2023mek}, namely, $t_{\max}$ grows at least as $n^3$. \\

\noindent{\textbf{Matrix $\mathbf{S}$}:} For completeness, in \cref{app:S-SU2}, we determine matrix  $\mathbf{S}$ in \cref{eq:S-def} in the case of $\SU(2)$ symmetry and show that its matrix elements are given by
\begin{align}\label{eq:S-SU2}
    s_{j'}(j) & = \binom{n-k}{\frac{n-k}{2}+j-j'} - \binom{n-k}{\frac{n-k}{2}+j+j'+1}\ ,
\end{align}
where angular momenta $j = j_{\min}, \cdots, \frac{n}{2}$ and $j' = \frac{k}{2} - \floor{\frac{k}{2}}, \cdots, \frac{k}{2}$ .  While it is not immediately obvious from these matrix elements, \cref{lem1} implies that this matrix has rank $\lfloor \frac{k}{2}\rfloor+1$.

\subsection{\texorpdfstring{$\mathbb{Z}_p$}{Zp} symmetry}\label{sec:Zp}

We consider $n$ qubits, with local orthonormal basis states $\ket{0}$ and $\ket{1}$, and representation of $\mathbb{Z}_p$ symmetry given by
\begin{equation}\label{repZp}
  u(a) = \exp \p[\Big]{\frac{\i 2 \pi a}{p} \qout{1}{1}} = \begin{pmatrix}
      1 & \\
      & \omega^a
  \end{pmatrix}\ ,
\end{equation}
where $a = 0, \cdots, p - 1$ is an element of $\mathbb{Z}_p$,  $\omega = \e^{\i 2 \pi / p}$, and the matrix is written in $\{|0\rangle,|1\rangle\}$ basis.  
Then, \cite{marvian2024theoryAbelian} proves the following results:

\begin{proposition}[\cite{marvian2024theoryAbelian}]\label{propzp} Consider a system of $n$ qubits with the on-site representation of  symmetry $\mathbb{Z}_p$ given in \cref{repZp}, and let $\mathcal{V}^{\mathbb{Z}_p}_{n,k}$ be the group generated by  $k$-local $\mathbb{Z}_p$-invariant gates. Then,
\begin{enumerate} 
\item For $k<p$ semi-universality is not achievable with $k$-qubit gates, and the uniform distribution over the group $\mathcal{V}^{\mathbb{Z}_p}_{n,k}$ is not even a 1-design for the uniform distribution over the group $\mathcal{V}^{\mathbb{Z}_p}_{n,n}$. 
\item For $k\ge p$ semi-universality is achieved, i.e., $\mathcal{V}^{\mathbb{Z}_p}_{n,k} \supseteq \mathcal{SV}^{\mathbb{Z}_p}_{n,n}$. 
 \item When $p$ is odd, universality is achieved with $k\ge p$, i.e., $\mathcal{V}^{\mathbb{Z}_p}_{n,k}=\mathcal{V}^{\mathbb{Z}_p}_{n,n}$. 
 \item When $p$ is even, universality is not achieved unless $k=n$,  i.e., $\mathcal{V}^{\mathbb{Z}_p}_{n,k}\neq\mathcal{V}^{\mathbb{Z}_p}_{n,n}$ for $k<n$. 
\end{enumerate}
\end{proposition}

Therefore, in the case of odd $p$, $t_{\max}=\infty$ when $k\geq p$. In the following, we apply  \cref{thm:tdesign} to determine 
$t_{\max}$ for even $p$. First, it can be easily seen that when $p$ is even, operator $Q=Z^{\otimes n}$ satisfies all the assumptions of the \cref{thm:tdesign}: its eigenvalues are integer, it is traceless, and it is orthogonal to all $k$-local operators with $k<n$.\footnote{Note that when $p$ is even, and only in this case, $Z^{\otimes n}$ can be written as a sum of the projectors to irreps of $\mathbb{Z}_p$: $Z^{\otimes n} = P_0-P_1$ where $P_0$ and $P_1$ are projectors to sectors with even and odd Hamming weight, respectively, which themselves can be decomposed to projectors to irreps of $\mathbb{Z}_p$.}
This immediately implies the upper bound 
$t_{\max} + 1 \le \frac{1}{2}\|Z^{\otimes n}\|_1=2^{n-1}$. In the following, we argue that this bound holds as equality and $Q=Z^{\otimes n}$ is indeed the optimal solution that satisfies the above conditions.

%\red{Recall that 
%both $\{S_\nu\}$ and $Q$ are in the real span of $\{\Pi_\lambda\}$, which in this case has dimension $p$. Furthermore, according to \cref{jdsa}, $\Tr(S_\nu Q)=0: \nu\in\irreps_{\mathbb{Z}_p}$.}

Recall that the orthogonality condition $\Tr(H_l Q)=0$ in \cref{thm:tdesign}, which should be satisfied for all $k$-local Hamiltonians $\{H_l\}$, is equivalent to the condition in \cref{jdsa}, namely $\Tr(S_\nu Q)=0: \nu\in\irreps_{\mathbb{Z}_p}$.   Here,   both $\{S_\nu\}$ and $Q$ are in the real span of $\{\Pi_\lambda\}$, which in this case has dimension $p$. As we show below,  for even $p$ and $k$ in the interval $p\le k< n$, 
\begin{align}\label{eq:rnkS}
\dim(\mathrm{span}_\real\{S_\nu: \nu\in\irreps_{\mathbb{Z}_p}(k)\}) \!= \!
 \rk\{\textbf{S}\} \!= \!p-1
\end{align}
where $\textbf{S}$ is the $p\times p$ matrix defined in \cref{eq:S-def}.

It follows that $\Tr(S_\nu Q)=0: \nu\in\irreps_{\mathbb{Z}_p}$ uniquely determines operator $Q$, up to a normalization. Therefore, $Q \propto Z^{\otimes n}$. Finally, because $Q$ must have integer eigenvalues, and eigenvalues of $Z^{\otimes n}$ are $\pm 1$, we know that it is indeed the optimal $Q$, that is
\be
t_{\max} + 1 = \frac{1}{2}\|Z^{\otimes n}\|_1 = 2^{n-1}\ .
\ee
To determine matrix $\mathbf{S}$ for this example, we note that each irreducible representation of $\mathbb{Z}_p$ is one-dimensional, and is of the form $\omega^{a \alpha}$ for some integer $\alpha = 0, \cdots, p - 1$, hence the character of the irrep labeled by $\alpha$ is  $f_\alpha(a) = \omega^{a \alpha}$. Then, applying  \cref{int},
\begin{align}
    s_\alpha(\beta) & = \frac{1}{p} \sum_{a = 0}^{p - 1} (1 + \omega^a)^{n - k} \omega^{a \alpha} \omega^{-a \beta} =  \sum_{l = 0}^{\floor{\frac{n - k}{p}}} \binom{n - k}{c + p l}\ ,
\end{align}
where $0 \leq c < p$ is given by $c \equiv \beta - \alpha \pmod{p}$. 
 In \cref{app:Zp} we explicitly show that for $k$ in the interval,   $p\le k<n$, this matrix has rank $p$, when $p$ is odd and rank $p-1$ when $p$ is even.  For completeness, the multiplicities $m_\beta$ of $\beta = 0, \cdots, p - 1$ are given by
\begin{equation}
  m_\beta = \frac{1}{p} \sum_{a = 0}^{p - 1} (1 + \omega^a)^{n} \omega^{-a \beta} = \sum_{l = 0}^{\floor{\frac{n}{p}}} \binom{n}{\beta + p l}\ .
\end{equation}

 \noindent$\mathbb{Z}_2$ \textbf{symmetry}:  In the example of $p = 2$ {and $2 \leq k < n$}, $\mathbf{S}$ is a $2\times 2$ matrix with matrix elements 
\begin{align}
  s_\alpha(\beta) & = 2^{n - k - 1}\ ,
\end{align}
which are independent of  $\alpha, \beta = 0, 1$.  Therefore, in this case, it is obvious that the rank of matrix $\mathbf{S}$ is 1,  
which is consistent with \cref{eq:rnkS}.
Then, the vector in the kernel of $\mathbf{S}$ is  unique up to a normalization  and given by $\mathbf{q} = (1, - 1)^{\mathrm{T}}$. This indeed corresponds to operator $Q = P_0 - P_1 = Z^{\otimes n}$, which as mentioned above, satisfies the conditions of \cref{thm:tdesign} with $\|Q\|_1 = 2^n$.

\medskip

\noindent$\mathbb{Z}_3$ \textbf{symmetry}: For $p=3$ semi-universality is achieved only with $k$-qubit gates with $k\ge 3$ \cite{marvian2024theoryAbelian}.   
 Suppose for simplicity that $n$ is an odd multiple of $3$ (the case for general $n$ is similar). Then,
\begin{equation}
    \mathbf{S} = \begin{pmatrix}
        N + 1 & N & N \\
        N & N + 1 & N \\
        N & N & N + 1
    \end{pmatrix},
\end{equation}
where $N = \floor{2^{n - 3} / 3}$. It is then easy to verify that $\mathbf{S}$ has full rank, e.g., by $\det \mathbf{S} = 1 + 3 N$.  This explicitly shows that in this case for $k\ge 3$ there are no constraints on the relative phases, as stated in \cref{propzp}.

\subsection{\texorpdfstring{$\SU(d)$}{SU(d)} symmetry}

Next, we consider $n$-qudit systems with the total Hilbert space  $(\mathbb{C}^d)^{\otimes n}$ and with $\SU(d)$ symmetry for arbitrary $d \geq 3$. That is, we consider unitaries $V$ commuting with $u^{\otimes n}$ for arbitrary $u\in\SU(d)$.

Interestingly, as was noticed in  \cite{marvian2022quditcircuit}, in this case, 2-qudit SU($d$)-invariant gates are not semi-universal for $d\ge 3$. However, it was recently shown \cite{hulse2024framework} that amending 2-qudit gates with any 3-qudit gate that is not realizable with 2-qudit gates, makes them semi-universal\footnote{We note that \cite{zheng2023speeding} establishes semi-universality of 4-qudit gates}, i.e.,  
\be
\mathcal{V}_{n,3}^{\SU(d)} \supseteq \mathcal{SV}_{n,n}^{\SU(d)}\ .
\ee

Consider a set of $\SU(d)$-invariant Hamiltonians $\{H_l\}$, and the corresponding matrix $\mathbf{M}$ defined in \cref{matrix}. According to the Schur-Weyl duality, any $\SU(d)$-invariant Hamiltonian can be written as a linear combination of permutation operators $\{\P(\sigma): \sigma \in \S_n\}$ \cite{harrow2005applications, marvian2022quditcircuit}. 
In particular, for Hamiltonian $H_l=\sum_\sigma h_l(\sigma) \P(\sigma)$, the matrix element
\begin{align}
    \mathbf{M}_{l,\lambda} = m_\lambda \frac{\Tr(\Pi_\lambda H_l)}{\Tr\Pi_\lambda}= \sum_{\sigma} h_l(\sigma) \chi_\lambda(\sigma)\ ,
\end{align} 
where
\be\label{eq:chi}
\chi_\lambda(\sigma)=m_\lambda \frac{\Tr(\Pi_\lambda\P(\sigma))}{\Tr\Pi_\lambda}\ , 
\ee
denotes the character of irrep $\lambda$ of the group element $\sigma \in \mathbb{S}_n$. If $H_l$ are $k$-local, then the above summation can be restricted to $k$-local permutations. We define matrix $\bm{\chi}$ such that its matrix elements are given by $\chi_\lambda(\sigma)$, i.e., 
\begin{align}\label{eq:Chi}
    \bm{\chi}_{\lambda, \sigma}=\chi_\lambda(\sigma),
\end{align}
and $\bm{\chi}^\Delta$ denotes the submatrix obtained by restricting $\bm{\chi}$ to a subset of irreps $\lambda \in \Delta$.

Recall that the only relevant property of matrix $\textbf{M}$ that determines $t_{\text{max}}$ is its kernel (null space). If $\set{H_l}$ consists of $k$-local Hamiltonians that generate $\mathcal{V}_{n, k}^{\SU(d)}$, then this kernel is the same as the kernel of the matrix $\bm{\chi}^\text{T}$ for all irreps $\lambda$  of $\SU(d)$ and all $k$-local permutations $\sigma\in \mathbb{S}_n$, which is also equal to the kernel of the matrix $\mathbf{S}$, as defined in \cref{eq:S-def}.

Finally, note that since $\Tr(\Pi_\lambda \P(\sigma))$ is constant over the conjugacy class of $\sigma$, to determine the kernel of matrix $\bm{\chi}^\text{T}$, it suffices to consider only one representative $\sigma$ from each conjugacy class (Recall that a conjugacy class consists of all permutations with the same cycle type, which means they all act on equal number of qudits). Since any  $k$-local permutation $\sigma\in\mathbb{S}_n$ can be interpreted as an element of $\mathbb{S}_k$, we can label such conjugacy classes with conjugacy classes of $\mathbb{S}_k$.

In summary, we conclude that $t_{\text{max}}$ is determined by the kernel of matrix $\bm{\chi}^{\text{T}}$, where $\bm{\chi}$ is the matrix defined by \cref{eq:Chi} for all irreps $\lambda$ of $\SU(d)$, or equivalently, all irreps of $\mathbb{S}_n$ that appear on $n$ qudits, and all conjugacy classes of $\mathbb{S}_k$. As an example, in \cref{tab:character-table} in \cref{app:SUd}, we calculate $\chi_\lambda(\sigma)$ for $\sigma \in \S_4$.  

Similar to the case of $\U(1)$ and $\SU(2)$ symmetries, $t_{\text{max}}$ is determined by the irreps of $\SU(d)$ with the lowest multiplicities. However, while it is relatively straightforward to order the irreps by their multiplicities for $\U(1)$ and $\SU(2)$, this is more complicated for $\SU(d)$.
Recall that, according to the Schur-Weyl duality, the multiplicity of each irrep of $\SU(d)$ is equal to the dimension of the corresponding irrep of $\mathbb{S}_n$. Therefore, to find $t_{\text{max}}$, we need to determine which of the irreps of $\mathbb{S}_n$ that appear in $(\mathbb{C}^d)^{\otimes n}$ have the lowest dimensions. The following facts summarize the relevant results from \cite{rasala1977minimal}.\\

\begin{fact}[\cite{rasala1977minimal}]\label{fact:dims}
For $n\ge 9$ (resp. $n\ge 15$) the first four (resp. seven) lowest dimensions of irreps of $\mathbb{S}_n$ are given by $m_\lambda$ in the first four (resp. seven) rows of \cref{tab:min-dim}, where $\lambda$ is an irrep of $\mathbb{S}_n$ with dimension $m_\lambda$. For each $m_\lambda$ there is exactly one other irrep $\lambda'$ of $\mathbb{S}_n$ that satisfies $m_\lambda = m_{\lambda'}$, where $\lambda'$ corresponds to the transpose of the Young diagram $\lambda$. In addition, when $n\geq 22$, the eighth lowest dimension of irreps of $\mathbb{S}_n$ is given by the last row in \cref{tab:min-dim}.

\begin{table}[htb]
  \centering
  \begin{tabular}{*{3}{c}}
    \toprule
    $\lambda$ & $m_\lambda$ \\
    \midrule
    $[n]$ & $1$ \\
    \midrule
    $[n-1,1]$ & $n-1$ \\
    \midrule
    $[n-2,2]$ & $\frac{1}{2}n(n-3)$ \\
    $[n-2,1,1]$ & $\frac{1}{2}(n-1)(n-2)$ \\
    \midrule
    $[n-3,3]$ & $\frac{1}{6}n(n-1)(n-5)$ \\
    $[n-3,1,1,1]$ & $\frac{1}{6}(n-1)(n-2)(n-3)$ \\
    $[n-3,2,1]$ & $\frac{1}{3}n(n-2)(n-4)$ \\
    \midrule
    $[n-4,4]$ & $\frac{1}{24}n(n-1)(n-2)(n-7)$ \\
    \bottomrule
  \end{tabular}
  \caption{Minimal dimension irreps of $\S_n$. Recall that irreps of $\S_n$ can be labeled by Young diagrams, or, equivalently,  partitions of $n$, $\lambda \vdash n$, which means $\lambda = [\lambda_1, \lambda_2, \cdots, \lambda_r]$ with $\sum_{i=1}^r \lambda_i = n$ and $\lambda_1 \geq \lambda_2 \geq \cdots \geq 0$. The transpose irrep $\lambda'$ of $\mathbb{S}_n$ shows up in $(\complex^d)^{\otimes n}$ if and only if $\lambda_1 \leq d$. This is because the irreps of $\S_n$ in $n$-qudit system $(\mathbb{C}^d)^{\otimes n}$ only contains Young diagrams with at most $d$ rows.}
  \label{tab:min-dim}
\end{table}
\end{fact}

As mentioned in the above fact, there are exactly two irreps of $\mathbb{S}_n$ with the dimensions given in the second column. However, as explained in the caption of \cref{tab:min-dim}, for $n$ qudits with the total Hilbert space $(\mathbb{C}^d)^{\otimes n}$, when $n \geq d + 5$, the irreps $\lambda$ given in the first column of \cref{tab:min-dim}, are the only irreps with the dimension given by the corresponding entry in the second column.

\subsubsection{Example for \texorpdfstring{$k=2$}{k=2}}

  While  $2$-qudit SU($d$)-invariant gates are not 
semi-universal for $d\ge 3$, as a simple example, it is useful to consider the scenario in which the group  $\mathcal{V}_{n,2}^{\SU(d)}$
 generated  by  $2$-qudit gates  is amended with $\mathcal{SV}^{\SU(d)}$, to make the resulting group semi-universal. In this case, to determine $t_{\text{max}}$ we need only to consider matrix $\bm{\chi}$ for the 2-local conjugacy classes of $\mathbb{S}_n$,   namely the identity and transposition conjugacy classes, 
with the representative elements $(1)$ and $(12)$, respectively, where we have used the cyclic notation to denote permutations (As mentioned before, these can be equivalently interpreted as the two conjugacy classes of $\mathbb{S}_2$).

We first use \cref{prop1} to establish a lower bound on $t_{\max}$. Therefore, we consider the matrix $\bm{\chi}_{\lambda,\sigma} : \sigma \in \S_2$ and $\lambda \in \Delta = \{ [n], [n-1,1] \}$, i.e., the trivial representation and the standard representation of $\S_n$,  which are the two lowest dimensional irreps when $n\geq 5$. From the character table of $\S_n$, we have
\begin{align}
    \bm{\chi}^\Delta =
    \begin{pmatrix}
        1 & 1 \\
        n-1 & n-3 \\
    \end{pmatrix}.
\end{align}
This matrix has full rank, and  according to  \cref{fact:dims} when $n \geq 9$, the lowest dimensional irrep other than $[n]$ and $[n-1,1]$ is $[n-2,2]$, whose dimension is $\frac{n(n-3)}{2}$. Therefore, applying \cref{prop1}, we conclude that when $n\geq \max\{9,d+2\}$, $t_{\max} \geq \frac{n(n-3)}{2} - 1$.

Next, we use  \cref{thm:tdesign}  to obtain an upper bound on $t_{\max}$. To achieve this, it suffices to find integers $q(\lambda)$ such that operator  $Q=\sum_{\lambda} q(\lambda) m_\lambda \frac{\Pi_\lambda}{\tr \Pi_\lambda}$  is orthogonal to all 2-local operators. Indeed, we find an operator $Q$ with such properties whose support is restricted to 4 irreps of $\SU(d)$ with the lowest multiplicities. As can be seen from \cref{tab:min-dim}, for $n\geq 9$, these irreps are 
\begin{align}\label{eq:Delta-k=2}
\Delta=\{ [n], [n-1,1], [n-2,2] ,  [n-2,1,1]\} . 
\end{align}
For these 4 irreps the characters of the two representative elements $(1)$ and $(12)$ of $\S_n$ 
 are given by the matrix
\begin{align}\label{eq:M4}
    {\bm{\chi}}^\Delta = 
    \begin{pmatrix}
        1 & 1 \\
        n-1 & n-3 \\
        \frac{n(n-3)}{2} & \frac{(n-3)(n-4)}{2} \\
        \frac{(n-1)(n-2)}{2} & \frac{(n-2)(n-5)}{2}
    \end{pmatrix}\ .
\end{align}
Consider integer vectors $\mathbf{q}\in \mathbb{Z}^4$ in the kernel of  $({\bm{\chi}}^\Delta)^{\text{T}}$, i.e., the subspace orthogonal to the two column vectors of matrix ${\bm{\chi}}^\Delta$ in \cref{eq:M4}. Solving the linear equations, we find $\mathbf{q}$ is in the span of the column vectors of 
\begin{align}\label{eq:q-sol-k=2}
  \begin{pmatrix}
    \binom{n-2}{2} & \binom{n-1}{2} \\
    -\binom{n-3}{1} & -\binom{n-2}{1} \\
    1 & 0 \\
    0 & 1 \\
  \end{pmatrix}.
\end{align}
While  each of these two 
integer vectors $\mathbf{q}$
yield an upper bound on $t_{\max}$, it turns out that the strongest bound is obtained by considering their linear combinations with coefficients $\pm 1$, which gives the vector $(2 - n, 1, 1, -1)^\T$ (Detailed calculation is similar to the one shown in \cref{app:SUd} for the case of $k=4$.)
This corresponds to the operator
\begin{align}\label{eq:Q-k=2}
    Q & = (2 - n) \frac{\Pi_{[n]}}{\tr \Pi_{[n]}} + (n - 1) \frac{\Pi_{[n - 1, 1]}}{\tr \Pi_{[n - 1, 1]}} \\
    & \mathrel{\phantom{=}} {} + \frac{n (n - 3)}{2} \frac{\Pi_{[n - 2, 2]}}{\tr \Pi_{[n - 2, 2]}} - \frac{(n - 1) (n - 2)}{2} \frac{\Pi_{[n - 2, 1, 1]}}{\tr \Pi_{[n - 2, 1, 1]}} \notag
\end{align}
which results in the upper bound $t_{\max} + 1 \leq \frac{1}{2} \|Q\|_1 = \frac{1}{2} (n + 1)(n - 2)$. Note that the support of matrix $Q$ is restricted to irreps $\lambda\in \Delta$ of $\SU(d)$ whose dimension in $(\mathbb{C}^d)^{\otimes n}$, denoted as  $m_\lambda$ scales up to $n^2$, i.e., $m_\lambda = \mathcal{O}(n^2)$. For any other operator $Q'$ satisfying the conditions in  \cref{eq:cond} with support on $\lambda_0\not\in \Delta$, we have $\frac{1}{2}\|Q'\|_1 \ge m_{\lambda_0}$ from \cref{Eq:change}. According to \cref{fact:dims}, when $n\geq \max\{15,d+3\}$, the smallest value of $m_{\lambda_0}$ for $\lambda_0 \not\in \Delta$ is $\frac{1}{6}n(n-1)(n-5)$, which is cubic in $n$. This implies $\frac{1}{2}\|Q'\|_1 > \frac{1}{2} \|Q\|_1 = \frac{1}{2} (n + 1)(n - 2)$, and we conclude that
\begin{equation}
    t_{\max}  = \frac{1}{2} (n + 1)(n - 2)-1\ .
\end{equation}

\subsubsection{Example for \texorpdfstring{$k=3$}{k=3}}

%\noindent $\textbf{k=3:}$ 
Next, we consider $t_{\max}$ for the group $\V_{n,3}^{\SU(d)}$, generated by 
$3$-qudit $\SU(d)$-invariant unitaries. In this case, the group  
$\S_3$ has 3 different conjugacy classes corresponding to the identity operator $(1)$, the transpositions, i.e., 2-cycles, with the representative element $(12)$, and 3-cycles, with representative element $(123)$. 

Then, when restricted to 4 irreps whose multiplicity scales at most quadratically with $n$, namely, $\Delta$ in \cref{eq:Delta-k=2}, we obtain $4\times 3$ matrix 
\begin{align}
    \bm{\chi}^\Delta = 
    \begin{pmatrix}
        1 & 1 &1 \\
        n-1 & n-3 & n-4 \\
        \frac{n(n-3)}{2} & \frac{(n-3)(n-4)}{2} & \frac{(n-3)(n-6)}{2} \\
        \frac{(n-1)(n-2)}{2} & \frac{(n-2)(n-5)}{2} & \frac{(n-4)(n-5)}{2}
    \end{pmatrix}\ .
\end{align}

Before finding the exact value of $t_{\max}$, we note that by applying  \cref{prop1} we can immediately establish a simple lower bound on $t_{\max}$. To see this note that the upper $2\times 3$ submatrix of $\bm{\chi}^\Delta$, obtained by considering irreps $[n]$ and $[n-1,1]$, has rank 2. By \cref{prop1}, this means
$$
t_{\max}\ge \min\{m_\lambda: \lambda\neq [n],[n-1,1]\}-1=\frac{n(n-3)}{2}-1\ .
$$
Here, the second equality holds for $n \geq \max\{9,d+2\}$, and 
follows from \cref{fact:dims}, which implies the minimum multiplicity $m_\lambda$ for  $\lambda\neq [n],[n-1,1]$ is achieved for irrep $\lambda=[n-2,2]$.\footnote{We note that this bound was established in \cite{hulse2024framework}, using an argument similar to \cref{prop1}.} It is also interesting to note that the $3\times 3$ submatrix corresponding to $[n], [n-1,1], [n-2,2]$ also has rank 2,\footnote{This can be understood as a consequence of the fact that the Young diagrams of these three irreps have at most two rows, which are exactly the irreps that appear in qubit systems, $d=2$. From this perspective, the above fact about rank deficiency of the $3\times 3$ sub-matrix is related to the universality of 2-qubit SU(2)-invariant gates on $n=3$ qubits.} which means the above lower bound is the strongest one that can be obtained from \cref{prop1}.

In order to determine the exact value of $t_{\max}$ via \cref{thm:tdesign} we note that the kernel of the matrix $(\bm{\chi}^\Delta)^\T$ is spanned by the first column vector in \cref{eq:q-sol-k=2}, which results in $Q$ operator 
\begin{align}\label{eq:Q3}
    Q &= \frac{1}{2} (n-2)(n-3) \frac{\Pi_{[n]}}{\Tr\Pi_{[n]}} + (3-n)(n-1) \frac{\Pi_{[n-1, 1]}}{\Tr\Pi_{[n-1, 1]}} \nonumber\\
    &\quad + \frac{1}{2} n(n-3) \frac{\Pi_{[n-2, 2]}}{\Tr\Pi_{[n-2, 2]}},
\end{align}
and the upper bound $t_{\max} \leq \frac{1}{2} \|Q\|_1-1 = (n-1)(n-3)-1$. Since the kernel of $(\bm{\chi}^{\Delta})^\T$ is one dimensional, for any other operator $Q'$ that is not proportional to $Q$ and satisfies the conditions in \cref{eq:cond}, $Q'$ must have support outside of $\Delta$. However, notice that the set $\Delta$ in this case is the same as the one in the $k=2$ case in \cref{eq:Delta-k=2}. Therefore, from the discussion below \cref{eq:Q-k=2}, we know that when $n\geq \max\{15,d+3\}$, $\frac{1}{2} \|Q'\|_1 \geq \frac{1}{6}n(n-1)(n-5)$, which implies that $\frac{1}{2}\|Q'\|_1 > \frac{1}{2} \|Q\|_1 = (n-1)(n-3)$, and
\begin{align}
    t_{\max}  =  \frac{1}{2} \|Q\|_1-1 =(n-1)(n-3)-1\ .
\end{align}
Note that the optimal operator $Q$ in \cref{eq:Q3} does not have support on the subspace with irrep $[n-2,1,1]$. This means that the support of $Q$ is restricted to Young diagrams with at most 2 rows, which are exactly the diagrams that appear in the case of qubit systems with $\SU(2)$ symmetry (i.e., $d=2$). This explains why the above value of $t_\text{max}$ is the same as the value of 
$t_\text{max}$ for qubit systems with  $\SU(2)$ symmetry with $k=2,3$ in \cref{eq:tmax-SU2}.

\subsubsection{Example for \texorpdfstring{$k=4$}{k=4}}

For gates acting on $k=4$ qudits, we have to consider a larger $\bm{\chi}_{\lambda,\sigma}$ matrix that contains $\sigma \in \S_4$.
Recall that $\S_4$ has 5 conjugacy classes, with representative elements 
$$(1),~ (12),~ (123),~ (12)(34) \text{ and } (1234) ,$$ 
which includes two additional elements compared with $k=3$, namely $(12)(34)$ and $(1234)$.

In this case, the dimensions of qudits $d=3$ and $d\geq 4$ should be treated separately. This is because for $d=3$, $(1234)$ does not introduce new linearly independent vectors to $\bm{\chi}$, while for $d\geq 4$, it does. That is, restricted only to the irreps $\lambda$ that show up for $\SU(3)$, namely Young diagrams with at most three rows, the vector $\chi_{\lambda, (1234)}$ linearly depends on the other vectors given by conjugacy classes with representatives $(1)$, $(12)$, $(123)$ and $(12)(34)$ respectively, which can be checked explicitly using the character table. Interestingly, we find nonetheless that the results coincide for $d = 3$ and $d \geq 4$.

The calculations are similar to the $k=2,3$ examples. Therefore, we only list the results here, and present the details of the calculation in \cref{app:SUd}.  In particular, applying the lower bounds obtained from \cref{prop1}, we find that in both cases, when $n\geq \max\{15,d+3\}$,
\begin{align}
    t_{\max}  \geq \frac{1}{6}n(n-1)(n-5)-1\ .
\end{align}
Furthermore, we show that when $n\geq 22$, the actual value of $t_{\max}$ is determined by an operator $Q$ with support only on 4 irreps
\begin{align}
    [n]\ ,\  [n-1,1]\ ,\ [n-2,2]\ ,\  [n-3,3]\ .
\end{align}
Note that these irreps correspond to Young diagrams with at most two rows. Therefore, using a similar argument as the $k=3$ case, the exact value of $t_{\max}$ is the same as the $\SU(2)$ case, namely%
\begin{align}
    t_{\max}  =\frac{1}{2}\|Q\|_1 -1= \frac{2}{3}(n-1)(n-3)(n-5)-1\ .
\end{align}

%\red{As we further explain in Appendix \cref{app:SUd}, the optimality of this operator $Q$ can be shown using the fact that for $n\ge 22$,  the 8th irrep with the lowest dimension is $[n-4,4]$, whose dimension scales as $n^4/24$. Therefore, from \cref{Eq:change} we know that including this irrep will result in an operator $Q'$ with larger $l_1$-norm. It is worth noting that the exact values of $t_{\max}$ is the same for both cases and coincide with the result for both cases with $d=3$, and $d> 3$, and is identical with the value of $t_{\max}$ for  $\SU(2)$, because the corresponding $Q$ operator only has support on Young diagrams with at most two rows. }

\subsubsection{ \texorpdfstring{$k=4$}{k=4} without 4-cycles}

Finally, when $d\geq4$, it is interesting to consider a subgroup $\mathcal{T} \subseteq  \mathcal{V}^{\SU(d)}_{n,4}$ generated by all 3-qudit $\SU(d)$-invariant Hamiltonians and a 4-qudit Hamiltonian $\P_{12}\P_{34}$, i.e., 
\be\label{eq:group-T}
\mathcal{T} =\langle  \mathcal{V}^{\SU(d)}_{n,3}, \e^{\i\theta\P_{12}\P_{34}}: \theta \in [0, 2\pi) \rangle\ .
\ee
Equivalently, $\mathcal{T}$ is generated by Hamiltonians 
$$\mathbb{I}\ ,\ \P_{12}\ ,\ \P_{(123)}+ \P_{(132)}\ ,\   \P_{12}\P_{34}\ , $$
and their permuted versions (Note that compared with $\mathcal{V}^{\SU(d)}_{n,4}$, we have excluded Hamiltonians such as $ \P_{(1234)}+\P^\dag_{(1234)}$, which cannot be realized with 3-qudit $\SU(d)$-invariant Hamiltonian, when $d\ge 4$).

Ref. \cite{Li:2023mek} has previously studied the design properties of $\mathcal{T}$ and  established the lower bound  $t_{\max}\ge n(n-3)/2-1$ for $n\geq \max\{9,d+1\}$. In \cref{app:SUd}, we show that when $n\geq \max\{15,d+3\}$, the lower bound of $t_{\max}$ given by \cref{prop1} is the same as the case of $k=3,4$, namely
\begin{align}
    t_{\max}  \geq \frac{1}{6}n(n-1)(n-5)-1\ .
\end{align}
Also in \cref{app:SUd}, we determine the exact value of $t_{\max}$ when $n\geq \max\{22,d+4\}$, namely 
\begin{align}
    t_{\max}=\frac{1}{6}(n-3)(2n^2-3n+4)-1\ .
\end{align}

\cref{tab:SUd} summarizes the results in this section.

\begin{table}[htb]
  \centering
  \begin{tabular}{*{4}{c}}
    \toprule
    $k$ & $t_{\max}+1$ & $n\geq$ \\
    \midrule
    $3$ & $(n-1)(n-3)$ & $\max\{15,d+3\}$ \\
    $4$ & $\frac{2}{3}(n-1)(n-3)(n-5)$ & $\max\{22,d+4\}$  \\
    \midrule
    $0,1$ $(+ \mathcal{SV}^{\SU(d)})$ & $(n-1)$ & $\max\{5,d+1\}$  \\
    $2$ $(+ \mathcal{SV}^{\SU(d)})$ & $\frac{1}{2} (n + 1)(n - 2)$ & $\max\{15,d+3\}$  \\
    $\mathcal{T}$ & $\frac{1}{6}(n-3)(2n^2-3n+4)$ & $\max\{22,d+4\}$  \\
    \bottomrule
  \end{tabular}
  \caption{$t_{\max}$ for $\SU(d)$. In the case of $k=0,1,2$, $\mathcal{SV}^{\SU(d)}$ needs to be amended in order to achieve any higher order $t$-design. $\mathcal{T}$ is defined in \cref{eq:group-T}.}
  \label{tab:SUd}
\end{table}

\section{Proofs of Proposition \ref{prop1} and Theorem \ref{thm:tdesign}}\label{sec:proofs}

In this section, we present the proofs of the results in \cref{sec:t-design}, namely \cref{prop1,thm:tdesign}.

\subsection{Proof of Proposition \ref{prop1}}\label{Sec:proof1}

Applying the decomposition $V=\bigoplus_{\lambda\in\irreps_G} (\mathbb{I}_\lambda\otimes v_\lambda)$ to both sides of \cref{design1}, we find that any distribution $\nu$ is a $t$-design for the Haar measure over $\mathcal{V}^G$, if, and only if 
\begin{align}\label{aa}
  &\forall \lambda_1,\cdots, \lambda_t, \lambda'_1,\cdots, \lambda'_t \in\irreps_G:\nonumber\\ 
  &\quad
    \mathbb{E}_{V\sim \nu}[\bigotimes_{i=1}^t v_{\lambda_i}\otimes v^\ast_{\lambda'_i}] = \mathbb{E}_{V\sim \mu_\mathrm{Haar}}[\bigotimes_{i=1}^t v_{\lambda_i}\otimes v^\ast_{\lambda'_i}] .
\end{align}
Up to the same permutation, 
the two sides can be rewritten as 
\be
B[{\mathbf{n}}, {\mathbf{n}}_\ast;\eta] :=\mathbb{E}_{V\sim \eta}\bigotimes_{\lambda\in \irreps_G} [v_\lambda^{\otimes {\mathbf{n}}(\lambda)}\otimes ({v_\lambda^\ast})^{\otimes {\mathbf{n}}_\ast(\lambda)}]\ , 
\ee
where $\eta=\nu$ and $\eta=\mu_\mathrm{Haar}$, on the left-hand and right-hand sides, respectively, and 
${\mathbf{n}}(\lambda)$ and ${\mathbf{n}}_\ast(\lambda)$ are the number of times $v_\lambda$ and $v^\ast_\lambda$ appear in $\bigotimes_{i=1}^t v_{\lambda_i}\otimes v^\ast_{\lambda'_i}$. We refer to them as partitions of $t$, because they satisfy
\be\label{part}
\sum_{\lambda} {\mathbf{n}}(\lambda)=\sum_{\lambda} {\mathbf{n}_\ast}(\lambda)=t\ , \ \ \ \ \ \ {\mathbf{n}}(\lambda), {\mathbf{n}_\ast}(\lambda)\ge 0 \ .
\ee
Then, \cref{aa} can be rewritten as 
\be\label{aa2}
B[{\mathbf{n}}, {\mathbf{n}}_\ast;\nu]= B[{\mathbf{n}}, {\mathbf{n}}_\ast;\mu_\mathrm{Haar}]\ , 
\ee
where 
${\mathbf{n}}$ and 
${\mathbf{n}}_\ast$ are all possible partitions as defined in \cref{part}. Note that each pair of partitions ${\mathbf{n}}$ and 
${\mathbf{n}}_\ast$ determines a subspace of $\mathcal{H}^{\otimes t}\otimes {\mathcal{H}^\ast}^{\otimes t}$. 

Suppose $\nu$ is the uniform distribution over a compact subgroup of $G$-invariant unitaries $\mathcal{V}^G$, and $W$ is a unitary in this group. Then, the left-invariance of the Haar measure implies that
\be\label{st}
(W^{\otimes t}\otimes {W^\ast}^{\otimes t})\mathbb{E}_{V\sim\nu }[V^{\otimes t}\otimes {V^\ast}^{\otimes t}]= \mathbb{E}_{V\sim\nu }[V^{\otimes t}\otimes {V^\ast}^{\otimes t}]\ .
\ee
Then, using the decomposition $W=\bigoplus_{\lambda\in\irreps_G} (\mathbb{I}_\lambda\otimes w_\lambda)$, we obtain 
\be\label{constr}
\Big[\hspace{-3mm}\bigotimes_{\lambda\in\irreps_G}\big(w_\lambda^{\otimes {\mathbf{n}}(\lambda)}\otimes {{w_\lambda^\ast}}^{\otimes {\mathbf{n}}_\ast(\lambda)}\big)\Big] B[{\mathbf{n}},{\mathbf{n}}_\ast;\eta]=B[{\mathbf{n}},{\mathbf{n}}_\ast;\eta]\ ,
\ee
for all partitions $\mathbf{n}$ and $\mathbf{n}_\ast$. This observation has the following implication.

\begin{lemma}\label{prop0}
  Suppose $\eta$ is the Haar measure over a compact subgroup of $\mathcal{V}^G$ that contains the 1-parameter family of unitaries $\exp(\i s C): s\in \mathbb{R}$, where $C=\sum_{\lambda} c(\lambda) \Pi_\lambda$. Then, for any pair of partitions ${\mathbf{n}}$ and ${\mathbf{n}_\ast}$, at least, one of the following holds: $B({\mathbf{n}},{\mathbf{n}}_\ast;\eta)=0$, or 
  \be
  \sum_{\lambda\in\irreps_G} c(\lambda)\times [{\mathbf{n}}(\lambda)-{\mathbf{n}}_\ast(\lambda)] =0\ .
  \ee
\end{lemma}
\begin{proof}
  Applying \cref{constr} to this 1-parameter family we find
  \be
  \Big[\prod_{\lambda\in\irreps_G} 
  \e^{\i s [{\mathbf{n}}(\lambda)-{\mathbf{n}}_\ast(\lambda)]c(\lambda)} \Big]B[{\mathbf{n}},{\mathbf{n}}_\ast;\eta]=B[{\mathbf{n}},{\mathbf{n}}_\ast;\eta]\ ,
  \ee
  for all $s\in\mathbb{R}$, which implies the lemma. 
\end{proof}
Next, we apply this lemma to the right-hand side of \cref{aa2}, which is by definition equal to the right-hand side of \cref{aa}, up to a permutation. In this case, we can choose $C=\Pi_\mu$ for all $\mu\in \irreps_G$. Note that $\mu_\mathrm{Haar}$ is the uniform distribution over $\mathcal{V}^G$, which in particular contains $\exp(\i \Pi_\mu \theta) : \theta\in[0,2\pi)$. Therefore, the assumption of the lemma is satisfied. 
Then, the lemma implies that $B[{\mathbf{n}},{\mathbf{n}}_\ast;\mu_\mathrm{Haar}]=0$ for all partitions ${\mathbf{n}}$ and ${\mathbf{n}}_\ast$, except those satisfying 
\be\label{const}
\forall \lambda\in\irreps_G:\ \ \ {\mathbf{n}}(\lambda)={\mathbf{n}}_\ast(\lambda)\ .
\ee
Furthermore, in the subspaces where this condition holds, the two sides of \cref{aa2} are automatically equal. 
To see this note that under the transformation $V\rightarrow V \sum_{\mu} \e^{\i\theta_\mu} \Pi_\mu$, which only changes the relative phases between different sectors, 
the operator 
$v_\lambda^{\otimes {\mathbf{n}}(\lambda)}\otimes ({v_\lambda^\ast})^{\otimes {\mathbf{n}}_\ast(\lambda)}$ remains unchanged, 
which implies taking the expected value with respect to the Haar measure over $\mathcal{SV}^G$ and $\mathcal{V}^G$ are the same. Then, 
we can replace $\mu_\mathrm{Haar}$ with 
the Haar measure over $\mathcal{SV}^G$, without changing the expectation value. 
We conclude if $\nu$ is the uniform distribution over 
$\mathcal{W}^G$, and if $\mathcal{W}^G$ is semi-universal, i.e.,
$$\mathcal{SV}^G \subseteq \mathcal{W}^G \subseteq \mathcal{V}^G\ , $$ 
then in the subspaces in which \cref{const} holds, the two sides of \cref{aa} are equal.

Furthermore, if \cref{const} is not satisfied for some $\lambda\in\irreps_G$, then the right-hand side of \cref{aa2} is zero. We conclude that to establish \cref{aa2} for all partitions $\mathbf{n}$ and $\mathbf{n}_\ast$, it suffices to show that the left-hand side of this equation is zero if \cref{const} is not satisfied. We show that this is the case provided that
\be\label{bound}
t< \min\big\{\dim(\mathcal{M}_\lambda): \lambda\in\irreps_G- \Delta \big\}\ .
\ee
Suppose $\mathbf{n}$ and $\mathbf{n}_\ast$ are not equal, i.e., \cref{const} does not hold. Then, exactly one of the following two cases holds: 
\begin{itemize}
\item \textbf{Case (i)}: There exists, at least, one irrep $ \lambda\in \irreps_G-\Delta$ such that ${\mathbf{n}}(\lambda)\neq {\mathbf{n}_\ast}(\lambda)$. 
\item \textbf{Case (ii):} For all irreps $ \lambda\in \irreps_G-\Delta$, ${\mathbf{n}}(\lambda)= {\mathbf{n}_\ast}(\lambda)$, but there exists, at least, one irrep $\lambda \in \Delta$ such that ${\mathbf{n}}(\lambda)\neq {\mathbf{n}_\ast}(\lambda)$.
\end{itemize}
In the following, we consider each case separately and show that in both cases $B[{\mathbf{n}}, {\mathbf{n}}_\ast;\nu]=0$. \\

\noindent \textbf{Case (i)}: In case (i) there exits a $\lambda\in \irreps_G-\Delta$ for which ${\mathbf{n}}(\lambda)\neq {\mathbf{n}_\ast}(\lambda)$. Since $t$ satisfies 
\cref{bound}
then 
\be
0<|{\mathbf{n}}(\lambda)-{\mathbf{n}_\ast}(\lambda)| \le t< \dim(\mathcal{M}_\lambda) \ ,
\ee
where we have used the fact that ${\mathbf{n}}(\lambda)\le t$ and ${\mathbf{n}}(\lambda')\le t$. In particular, this means $${\mathbf{n}}(\lambda)\not\equiv {\mathbf{n}_\ast}(\lambda)\ \ \ \big(\mathrm{mod}~ \dim(\mathcal{M}_\lambda)\big)\ .$$ 
Then, the following lemma together with the fact that $\nu$ is the uniform distribution over $\mathcal{W}^G$ that contains
$\mathcal{SV}^G$, implies $B({\mathbf{n}},{\mathbf{n}}_\ast;\nu)=0$. 

\begin{lemma}\label{prop2}
  Suppose $\eta$ is the Haar measure over a compact subgroup of $\mathcal{V}^G$ that contains $\mathcal{SV}^G$. Then, for any pair of partitions ${\mathbf{n}}$ and ${\mathbf{n}_\ast}$, at least, one of the following holds: $B({\mathbf{n}},{\mathbf{n}}_\ast;\eta)=0$, or 
  \be
  \forall \lambda\in \irreps_G:\ \ {\mathbf{n}}(\lambda) \equiv {\mathbf{n}_\ast}(\lambda)\ \ (\mathrm{mod}~ \dim(\mathcal{M}_\lambda))\ .
  \ee
\end{lemma}
\begin{proof}
To prove this lemma we use the following standard fact, which can be shown, e.g., using the representation theory of $\SU(d)$. 
  
 \begin{fact}
 There exists a non-zero vector that is invariant under $u^{\otimes r}\otimes {u^\ast}^{\otimes s}$ for all $u\in \SU(m)$, if and only if $r\equiv s (\mathrm{mod}~ m)$.
\end{fact}

  To apply this fact, recall that for any set of unitaries $w_\lambda\in \SU(\dim(\mathcal{M}_\lambda))$, there exists $W\in \mathcal{SV}^G$ with decomposition $W=\bigoplus_{\lambda} \mathbb{I}_\lambda\otimes w_\lambda$. Then, using
  \cref{constr} for such $W$, and by applying the above fact, we can prove the lemma. 
\end{proof}

\noindent \textbf{Case (ii):} In this case there exists $\mu\in \Delta$ for which ${\mathbf{n}}(\mu)\neq {\mathbf{n}}_\ast(\mu)$. Recall the assumption that $\dim(\mathrm{span}_\mathbb{R}\{\sum_{\mu\in\Delta} \Tr(H_l\Pi_\mu) |\mu\rangle \}_l)=|\Delta|$, which in words means that under the map $O\rightarrow \sum_{\lambda\in \Delta} \Tr(\Pi_\lambda O) |\lambda\rangle$, the image of operators $\{H_l\}$ span a $|\Delta|$ dimensional space. This, in particular, means that for any $\mu\in \Delta$, there exists a Hermitian operator $A^\mu\in \mathrm{span}_\mathbb{R}\{H_l\}_l$, such that
\begin{align}\label{op}
\Tr(A^\mu \Pi_\lambda) 
\begin{cases}
    \neq 0, & \text{when } \lambda=\mu \\
    =0, & \text{when } \lambda \in \Delta-\{\mu\}.
\end{cases}
\end{align}

Consider the decomposition $A^\mu = A_0^\mu + A_1^\mu$, where
\be\label{proj}
A_1^\mu:=\sum_{\lambda\in\irreps_G} \frac{\Tr(\Pi_\lambda A^\mu)}{\Tr(\Pi_\lambda)} \Pi_\lambda=\sum_{\lambda\in\irreps_G} a_1^\mu(\lambda)\Pi_\lambda ,
\ee
to be projection of $A^\mu$ to the subspace spanned by $\{\Pi_\lambda\}$. Because $\Tr(A_0^\mu\Pi_\lambda)=0$ for all $\lambda\in\irreps_G$, we have $\exp(\i A_0^\mu \theta)\in \mathcal{SV}^G \subseteq \mathcal{W}^G$.
Since $\mathcal{W}^G$ is a group and $[A_1^\mu,A^\mu]=0$, it follows that
\be
\exp(\i A_1^\mu \theta)=\exp(\i A^\mu \theta)\exp(-\i A_0^\mu \theta)\in \mathcal{W}^G\ .
\ee
Furthermore, \cref{op} implies
\be
a_1^\mu(\mu) \neq 0,\ a_1^\mu(\lambda)=0 :~ \forall \lambda\in \Delta-\{\mu\}\ . 
\ee
Finally, we recall that in case (ii) ${\mathbf{n}}(\lambda)= {\mathbf{n}_\ast}(\lambda)$ for $\lambda\in \irreps_G-\Delta$. Therefore,
\begin{align}
  &\sum_{\lambda\in\irreps_G} \hspace{-2mm} a_1^\mu(\lambda)[{\mathbf{n}}(\lambda)-{\mathbf{n}}_\ast(\lambda)]\nonumber= \sum_{\lambda\in\Delta} a_1^\mu(\lambda) [{\mathbf{n}}(\lambda)-{\mathbf{n}}_\ast(\lambda)]\nonumber\\ &\ \ \ \ =a_1^\mu(\mu) [{\mathbf{n}}(\mu)-{\mathbf{n}}_\ast(\mu)] \ .
\end{align}
Since $a_1^\mu(\mu)\neq 0$ and $[{\mathbf{n}}(\mu)-{\mathbf{n}}_\ast(\mu)]\neq 0$, then $\sum_{\lambda\in\irreps_G} a_1^\mu(\lambda) [{\mathbf{n}}(\lambda)-{\mathbf{n}}_\ast(\lambda)]\neq 0$, which by \cref{prop0} implies $B({\mathbf{n}},{\mathbf{n}}_\ast;\nu)=0$.

We conclude that in both cases (i) and (ii), $B[{\mathbf{n}}, {\mathbf{n}}_\ast;\nu]=0$, which implies the left-hand side of \cref{aa2} is zero when the condition $\mathbf{n} = \mathbf{n}_*$ in \cref{const} is not satisfied and therefore, is equal to the right-hand side. As we argued before, they are also equal when this condition is satisfied. We conclude \cref{aa2} holds for all partitions ${\mathbf{n}}$ and $ {\mathbf{n}}_\ast$, and therefore $\nu$ is a $t$-design for $\mu_\mathrm{Haar}$. This completes the proof of \cref{prop1}.

\subsection{Proof of Theorem \ref{thm:tdesign}}
\label{Sec:design} 

First, we establish the upper bound $t_{\max} \leq \frac{1}{2} \|Q\|_1 -1$ for any  operator $Q$, 
or, equivalently, any set of integers $q(\lambda)$ satisfying \cref{eq:cond}. That is, we show that the uniform distribution over $\mathcal{W}^G$ is not $t$-design for the uniform distribution over $\mathcal{V}^G$ with $t \geq \frac{1}{2} \|Q\|_1$.  (We note that this has already been established by our argument in \cref{Sec:comm}, but for completeness, we provide a slightly different version of the proof, consistent with the rest of the notation and arguments in this section). 

Recall that
\begin{align}
  Q &= \sum_{\lambda \in \irreps_G} q(\lambda) m_\lambda \frac{\Pi_\lambda}{\tr \Pi_\lambda}\ , \\\tag{re \ref{QQ}}
  &= \sum_{\lambda \in \irreps_G} \big(\mathbf{n}(\lambda) - \mathbf{n}_*(\lambda)\big) \frac{\Pi_\lambda}{\tr \Pi_\lambda},
\end{align} 
where in the second line we define 
\bes
\begin{align}
  {\mathbf{n}}(\lambda)& := m_\lambda\times \frac{|q(\lambda)|+q(\lambda) }{2} \\ 
  {\mathbf{n}_\ast}(\lambda)& := m_\lambda\times \frac{|q(\lambda)|-q(\lambda) }{2}\ ,
\end{align}
\ees
which means ${\mathbf{n}}(\lambda)=m_\lambda\times q(\lambda)$ if $q(\lambda)>0$, and is zero otherwise. Similarly, ${\mathbf{n}_\ast}(\lambda)=-m_\lambda\times q(\lambda)$ if $q(\lambda)<0$, and is zero otherwise. Clearly, since $q(\lambda)$ are integers, $\mathbf{n}(\lambda)$ and $\mathbf{n}_*(\lambda)$ defined in this way are also integers. Furthermore, since $Q\neq 0$, we know that there exist some $\lambda$, such that $\mathbf{n}(\lambda) \neq \mathbf{n}_*(\lambda)$. This definition implies that 
\be
\sum_{\lambda} {\mathbf{n}}(\lambda)=\sum_{\lambda} {\mathbf{n}_\ast}(\lambda)=\frac{\|Q\|_1}{2} =: t' \ ,
\ee
and therefore, $\mathbf{n}$ and ${\mathbf{n}_\ast}$ are two distinct partitions of $t$.

In the following, we show that with this definition 
\be\label{tr1y}
B[{\mathbf{n}},{\mathbf{n}}_\ast;\nu] \neq 0\ ,
\ee
where $\nu$ is the uniform distribution over $\mathcal{W}^G$. This implies that $\nu$ is not a $t'$-design for the uniform distribution over $\mathcal{V}^G$, because $B[{\mathbf{n}},{\mathbf{n}}_\ast;\mu_\mathrm{Haar}]=0 $, for any distinct pair of partitions $\mathbf{n}$ and ${\mathbf{n}}_\ast$. Therefore, $\nu$ is not a $t$-design for any $t \geq \frac{1}{2} \|Q\|_1$.

To see \cref{tr1y} recall that
\be
B[{\mathbf{n}}, {\mathbf{n}}_\ast;\nu] :=\mathbb{E}_{V\sim \nu}\bigotimes_{\lambda\in \irreps_G} [v_\lambda^{\otimes {\mathbf{n}}(\lambda)}\otimes ({v_\lambda^\ast})^{\otimes {\mathbf{n}}_\ast(\lambda)}]\ , 
\ee
where $\nu$ is the Haar measure over $\mathcal{W}^G$. Next, recall that by assumption $\mathcal{W}^G$ is compact, and generated by the Hamiltonians $\set{H_r}$, which means the exponential map from its Lie algebra (denoted by $\mathfrak{w}=\mathfrak{alg}\{\i H_r\}$) to $\mathcal{W}^G$ is surjective. That is, for any unitary $W\in \mathcal{W}^G$, there exists an operator $\i F\in \mathfrak{w}$ such that $W=e^{\i F}$. Consider decomposition $F=F_0+F_1$, where 
\be
F_1=\sum_\lambda \frac{\Tr(\Pi_\lambda F)}{\Tr(\Pi_\lambda)} \Pi_\lambda
\ee
is the projection of $F$ to the subspace spanned by $\{\Pi_\lambda\}$.
This means $F_0=F-F_1$ satisfies $\Tr(F_0\Pi_\lambda)=0$ for all $\lambda\in\irreps_G$. Semi-universality implies that the Lie algebra $\mathfrak{w}=\mathfrak{alg}\{\i H_r\}$ contains $\i F_0$, which in turn implies that it also contains $\i F_1=\i F-\i F_0$. Furthermore, it can be easily seen that $F_1$ should be in the linear span of $H_{r,1}$, i.e.,
$F_1\in \mathrm{span}_\mathbb{R}\{H_{r,1}\}$, where 
\begin{align}\label{eq:Hr1}
    H_{r,1} =\sum_\lambda \frac{\Tr(\Pi_\lambda H_r)}{\Tr(\Pi_\lambda)} \Pi_\lambda
\end{align}
is again the projection of $H_r$ to the subspace spanned by $\{\Pi_\lambda\}$. This is a consequence of the fact that 
all the elements of the Lie algebra $\mathfrak{w}$ are $G$-invariant, which in turn implies that for any $B_1, B_2 \in \mathfrak{w}$, $\Tr(\Pi_\lambda[B_1,B_2])=0$. Therefore, the projection of any element of the Lie algebra $\mathfrak{w}$ to $\mathrm{span}\{\Pi_\lambda\}$ should be a linear combination of the projection of the generators of the Lie algebra (See \cite{Marvian2022Restrict} for further details. See also \cite{zimboras2015symmetry}). 

This implies that 
\be
W=W_1W_0=W_0 W_1\ ,
\ee
where $W_0=\e^{\i F_0}\in\mathcal{SV}^G$ and 
\be
W_1=\prod_r \exp(\i t_r H_{r,1})\ ,
\ee
for certain $t_r\in \mathbb{R}$. This means that the group generated by $\exp(\i t H_{r,1}): t\in \mathbb{R}, r = 1, \dots, N$ is contained in the center of $\mathcal{W}^G$; in fact, it is the identity component of the center. It follows from the assumption that $\mathcal{W}^G$ is compact that its center, and hence the identity component of its center, are also compact.\footnote{Suppose $\mathcal{W}$ is any topological group with the property that the subset consisting of only the identity, $\set{\ident}$, is closed (this is always true if $\mathcal{W}$ is a Lie group). Then its center is closed in $\mathcal{W}$. In fact, the centralizer subgroup of any element $V \in \mathcal{W}$ is closed: the map $W \mapsto W V W^{-1} V^{-1}$ is a continuous map from $\mathcal{W}$ to itself, and the inverse image of $\set{\ident}$ is precisely the centralizer of $V$, which is therefore closed. The center is the intersection of all centralizers, hence is also closed. See also \cite{husain1966introduction}.}
We conclude that
\be\label{eq:B0B1}
B[{\mathbf{n}}, {\mathbf{n}}_\ast;\nu]=B[{\mathbf{n}}, {\mathbf{n}}_\ast;\nu_0] B[{\mathbf{n}}, {\mathbf{n}}_\ast;\nu_1]\ , 
\ee
where $\nu_0$ and $\nu_1$ are, respectively, the uniform distributions over $\mathcal{SV}^G$ and over the group generated by $\exp(\i t H_{r,1}):t\in \mathbb{R}$.

Next, we argue that for the above choice of partitions $\mathbf{n}$ and $\mathbf{n}_\ast$, $B[{\mathbf{n}}, {\mathbf{n}}_\ast;\nu_1]$ is the identity operator and $B[{\mathbf{n}}, {\mathbf{n}}_\ast;\nu_0]$ is non-zero, which implies that $B[{\mathbf{n}}, {\mathbf{n}}_\ast;\nu]$ is non-zero. 

To show that $B[{\mathbf{n}}, {\mathbf{n}}_\ast;\nu_1]$ is the identity operator, note that for any $V=\exp(\i t H_{r,1})$, we have $v_\lambda=\exp(\i t \frac{\Tr(H_r\Pi_\lambda)}{\Tr(\Pi_\lambda)})\mathbb{I}_{\mathcal{M}_\lambda}$, which means
\begin{align}
  &\bigotimes_{\lambda\in \irreps_G} [v_\lambda^{\otimes {\mathbf{n}}(\lambda)}\otimes ({v_\lambda^\ast})^{\otimes {\mathbf{n}}_\ast(\lambda)}] \noindent\\ \nonumber
  &=\exp\Big( \i t \sum_\lambda [\mathbf{n}(\lambda)-\mathbf{n}_\ast(\lambda)]\frac{\Tr(H_r\Pi_\lambda)}{\Tr(\Pi_\lambda)} \Big) \mathbb{I}\\ \nonumber
  &=\exp(\i t \Tr({H}_r Q)])\mathbb{I}\ ,
\end{align}
which by \cref{cond2} is the identity operator. It follows that for any unitary $V=\prod_r \exp(\i t_r H_{r,1})$, the corresponding 
$\bigotimes_{\lambda\in \irreps_G} [v_\lambda^{\otimes {\mathbf{n}}(\lambda)}\otimes ({v_\lambda^\ast})^{\otimes {\mathbf{n}}_\ast(\lambda)}] $ is also the identity operator. Since $\nu_1$ is the uniform distribution over the group of unitaries with such decomposition, it follows that 
$B[{\mathbf{n}}, {\mathbf{n}}_\ast;\nu_1]$ is the identity operator.

Next, we focus on $B[{\mathbf{n}}, {\mathbf{n}}_\ast;\nu_0]$ and show that it is non-zero. First note that
\be
B[{\mathbf{n}}, {\mathbf{n}}_\ast;\nu_0]=\bigotimes_{\lambda\in \irreps_G} \mathbb{E}_{v_\lambda\sim \mu_\lambda} [v_\lambda^{\otimes {\mathbf{n}}(\lambda)}\otimes ({v_\lambda^\ast})^{\otimes {\mathbf{n}}_\ast(\lambda)}]\ , 
\ee
where $\mu_\lambda$ is the uniform distribution over $\SU(m_\lambda)$. Then, recall that for the above definition of partitions $\mathbf{n}$ and $\mathbf{n}_\ast$, we have the property that for all $\lambda\in\irreps_G$, $ [{\mathbf{n}}(\lambda)- {\mathbf{n}}_\ast(\lambda)]/m_\lambda=q(\lambda)$ is an integer. This implies that there exists a non-zero vector that is invariant under $u^{\otimes {\mathbf{n}}(\lambda)}\otimes {u^\ast}^{\otimes {\mathbf{n}}_\ast(\lambda)}$, which means $B[{\mathbf{n}}, {\mathbf{n}}_\ast;\eta_0]$ is non-zero.

We conclude that for the above choice of partitions ${\mathbf{n}}$ and ${\mathbf{n}}_\ast$, 
\begin{align}
  B[{\mathbf{n}}, {\mathbf{n}}_\ast;\nu]&=B[{\mathbf{n}}, {\mathbf{n}}_\ast;\nu_0] B[{\mathbf{n}}, {\mathbf{n}}_\ast;\nu_1] \nonumber\\ &=B[{\mathbf{n}}, {\mathbf{n}}_\ast;\nu_0]\neq 0 \ . 
\end{align}
On the other hand, because ${\mathbf{n}}$ and ${\mathbf{n}}_\ast$ are distinct partitions, from \cref{const} we know that for the uniform distribution over the group of $G$-invariant unitaries $\mathcal{V}^G$, i.e., $B[{\mathbf{n}}, {\mathbf{n}}_\ast;\mu_\mathrm{Haar}] = 0$. This means that the uniform distribution over $\mathcal{W}^G$ is not a $t'$-design, for the uniform distribution over $\mathcal{V}^G$. Therefore, $\nu$ is not a $t$-design for any $t \geq \frac{1}{2} \|Q\|_1$, which implies that $t_{\max} \leq \frac{1}{2} \|Q\|_1-1$.

Next, we show that if the uniform distribution over $\W^G$ is not a $t$-design for the uniform distribution over $\V^G$, then there exists a set of integers $q(\lambda)$, or equivalently, an operator $Q$ satisfying \cref{eq:cond} with $\frac{1}{2}\|Q\|_1\le t$. This, together with the above upper bound implies that
\begin{align}\tag{re \ref{eq:tmax}}
    t_{\max} = \frac{1}{2} \min_Q \|Q\|_1-1\ .
\end{align}

Suppose the uniform distribution over $\mathcal{W}^G$, denoted as $\nu$, is not an exact $t$ design for the uniform distribution over $\mathcal{V}^G$. Then, there exists distinct partitions $\mathbf{n}$ and $\mathbf{n}_\ast$ of $t$ such that $B({\mathbf{n}},{\mathbf{n}}_\ast;\eta)\neq 0$. Then, \cref{prop0,prop2} imply that these partitions satisfy both of the following properties: 
\begin{enumerate}[(i)]
    \item For all $\lambda\in \irreps_G$, 
\be\label{Eq4}
{\mathbf{n}}(\lambda) \equiv {\mathbf{n}_\ast}(\lambda)\ \ \ \ (\mathrm{mod}~ m_\lambda)\ ,
\ee
    \item For all $\i C=\sum_\lambda \i c_\lambda \Pi_\lambda$ in the Lie algebra $\mathfrak{w}:=\mathfrak{alg}\{\i H_r\}$ generated by $\{\i H_r\}$,
\be\label{Eq5}
\sum_{\lambda\in \irreps_G} c_\lambda [\mathbf{n}(\lambda)-\mathbf{n}_\ast(\lambda)]=0\ .
\ee
\end{enumerate}

We first show that for any $\mathbf{n}(\lambda)$ and $\mathbf{n}_*(\lambda)$ that satisfies the above two properties, there is a set of integers $q(\lambda)$, or equivalently, an operator $Q$ that satisfies \cref{eq:cond}, namely
\begin{align}
    q(\lambda) = \frac{\mathbf{n}(\lambda)-\mathbf{n}_\ast(\lambda)}{m_\lambda}\ ,
\end{align}
which are clearly integers from \cref{Eq4}, and
\begin{align}
    Q = \sum_{\lambda \in \irreps_G} \big( \mathbf{n}(\lambda)-\mathbf{n}_*(\lambda) \big) \frac{\Pi_\lambda}{\tr \Pi_\lambda}\ .
\end{align}
Furthermore, we have
\bes\label{cond4}
\begin{align}
  \|Q\|_1=\sum_\lambda |q(\lambda)| m_\lambda&=\sum_\lambda |\mathbf{n}(\lambda)-\mathbf{n}_\ast(\lambda)|\\ &\le \sum_\lambda \mathbf{n}(\lambda)+\mathbf{n}_\ast(\lambda)=2t\ ,
\end{align}
\ees
which is non-zero, and in the second line we have used the triangle inequality. Moreover,
\be\label{cond5}
\Tr Q=\sum_{\lambda} q(\lambda)\times m_\lambda=\sum_{\lambda} \mathbf{n}(\lambda)-\mathbf{n}_\ast(\lambda)=0\ .
\ee
Therefore $Q$ satisfies both \cref{cond1,cond0}.

Finally, \cref{Eq5} can be rewritten as
\be\label{gh}
\Tr(C Q)= \sum_{\lambda\in \irreps_G} c_\lambda [\mathbf{n}(\lambda)-\mathbf{n}_\ast(\lambda)]=0\ .
\ee
which holds for any $C=\sum_\lambda c_\lambda \Pi_\lambda$, provided that $\i C$ is 
in the Lie algebra $\mathfrak{w}$ generated by $\{\i H_r\}$. Recall that the semi-universality implies that the Lie algebra $\mathfrak{w}$ contains any $G$-invariant anti-Herimitan operator $A$ satisfying $\Tr(A\Pi_\lambda)=0$ for all $\lambda\in\irreps_G$. This, in turn, implies that for any Hamiltonian $H_r$, $\mathfrak{w}$ contains operators $\i H_{r,1}$ defined in \cref{eq:Hr1}. Therefore,  by choosing $C=H_{r,1}$, \cref{gh} implies that
\begin{align}
    \Tr(Q H_r) = \Tr(Q H_{r,1}) = 0.
\end{align}

We conclude that if the uniform distribution over $\mathcal{W}^G$ is not a $t$ design for the uniform distribution over $\mathcal{V}^G$, then there exists integers $q(\lambda)$ and the corresponding operator $Q$ satisfying \cref{eq:cond}, such that $t \geq \frac{1}{2}\|Q\|_1$, which implies
\begin{align}\tag{re \ref{eq:tmax}}
    t_{\max} = \frac{1}{2} \min_Q \|Q\|_1-1.
\end{align}

A direct corollary of the above result is that if there does not exist a solution $Q$ with integer $q(\lambda)$ that satisfies \cref{eq:cond}, then $t_{\max} = \infty$, i.e., \cref{aa} holds for all integer $t$. This is because otherwise, there exists a finite $t$ such that the uniform distribution over $\mathcal{W}^G$ is not a $t$ design for the uniform distribution over $\mathcal{V}^G$. From the above argument, we know that there must exist a solution $Q$ that satisfies \cref{eq:cond}, and hence a contradiction to our assumption.

In the following, we will argue that $t_{\max} = \infty$ implies that, up to possible global phases, $\W^G$ is equal to $\V^G$, i.e., for any $V \in \V^G$, there exists $\theta\in[0,2\pi)$, such that $\e^{\i\theta} V\in \mathcal{W}^G$. As we explain \cref{lem:design} this can be understood more generally in terms of compact manifolds with strictly positive regular measures. 

Suppose instead there is an element $U \in \mathcal{V}^G$ such that $\e^{\i \theta} U \not\in \mathcal{W}^G$ for every phase $\e^{\i \theta}$. Then,  since $\mathcal{W}^G$ is compact and hence closed,  the function $V \mapsto \frac{1}{(\tr \ident)^2} \abs{\tr U^{-1} V}^2$ is bounded away from $1$ on $\mathcal{W}^G$, i.e., there exists some $0 < \varepsilon < 1$ such that
\begin{equation}\label{eq:Wbound}
    \frac{1}{(\tr \ident)^2} \abs{\tr U^{-1} W}^2 \leq 1 - \varepsilon\ ,
\end{equation}
for all $W \in \mathcal{W}^G$. Define the function
\begin{equation}
    p_U(V) = \frac{1}{(\tr \ident)^2} \abs{\tr U^{-1} V}^2 + \frac{\varepsilon}{2}.
\end{equation}
Then, there exists  a finite ball of unitaries $V \in \mathcal{V}^G$ around $U$ such that $p_U(V) > 1$, while $p_U(W) < 1$ for all $W \in \mathcal{W}^G$.
Let $\mathrm{vol}>0$ denote the volume of this ball. Then, since $p_U(V) > 1$ on this region,
\begin{equation}
    \int_{\mathcal{V}^G} \diff V \, p_U(V)^t \geq \mathrm{vol}
\end{equation}
for all integer $t \ge 1$. On the other hand, since $p_U(W) < 1$ for $W \in \mathcal{W}^G$, 
\begin{align}
    \lim_{t \to \infty} \int_{\mathcal{W}^G} \diff W \, p_U(W)^{t} & = 0\ .
\end{align}
This, in particular, implies that there is some finite $t$ such that
\begin{equation}
    \int_{\mathcal{V}^G} \diff V \abs{\tr U^{-1} V}^{2 t} \neq \int_{\mathcal{W}^G} \diff W \abs{\tr U^{-1} W}^{2 t},
\end{equation}
and therefore $\mathcal{W}^G$ is not a $t$-design for $\mathcal{V}^G$.

Recall that any compact Lie group has a finite-dimensional faithful representation \cite{sepanski2006compact}. Therefore, the above construction immediately generalizes to arbitrary compact Lie groups.  In the following, we emphasize a matrix realization because $t$-design properties of the group depend on this representation.
\begin{proposition}
    Let $\mathcal{V}$ be a compact matrix Lie group with uniform Haar measure $\mu_{\mathcal{V}}$ and $\mathcal{W} \subsetneq \mathcal{V}$ be a proper closed subgroup, with Haar measure $\mu_{\mathcal{W}}$. If $\U(1) \cdot \mathcal{W} \neq \mathcal{V}$, i.e., there exists an element $U \in \mathcal{V}$ such that $\e^{\i \theta} U \notin \mathcal{W}$ for all phases $\e^{\i \theta}$, then there exists an integer $t$ such that $\mu_{\mathcal{W}}$ is not a $t$-design for $\mu_{\mathcal{\mathcal{V}}}$.
\end{proposition}

The following proposition, proven in \cref{app:design}, shows that this fact can be understood more generally in terms of regular measures over compact manifolds, using Urysohn's lemma \cite{munkres2000topology} and Stone-Weierstrass theorem \cite{rudin1976principles}. 

\begin{restatable}{proposition}{lemDesign}\label{lem:design}
  Let $N$ be a compact manifold, and let $M\subsetneq N$ be a proper closed submanifold. Let $\mu_N$ and $\mu_M$ be normalized strictly-positive regular measures on $N$ and $M$, respectively. Suppose $f: N \hookrightarrow \mathbb{R}^n$ is an embedding of $N$ into $\mathbb{R}^n$. Then there exists a polynomial function $p \in \mathbb{R}[x_1, \cdots, x_n]$ of degree $t \geq 0$, such that
  \begin{align}
    \int_M \d\mu_M\, p \circ f \neq \int_N \d\mu_N\, p \circ f.
  \end{align}
\end{restatable}

\section{Discussion}\label{sec:discussion}

In this paper, we discussed the design properties of random quantum circuits in the presence of symmetries. As highlighted in \cref{prop:con}, the notion of semi-universality plays a central role in this context: if a compact connected subgroup of $G$-invariant unitaries (such as the subgroup generated by $k$-local $G$-invariant unitaries) is not semi-universal, then the Haar distribution over that subgroup cannot be even a 2-design for the Haar distribution over the group of all $G$-invariant unitaries. When semi-universality holds, \cref{thm:tdesign} determines the maximum $t = t_{\max}$ such that the uniform distribution over the subgroup is a $t_{\max}$-design but not $t_{\max}+1$-design for the Haar measure. We also found a useful lower bound on $t_{\max}$ in \cref{prop1}. It is worth noting that our general results in \cref{thm:tdesign} and \cref{prop1} are formulated in terms of general symmetries and their representations, and therefore are applicable beyond the case of on-site representation of symmetries.  
%According to \cref{thm:tdesign}, finding $t_{\max}$ is reduced to finding solutions to the $Q$ operators that satisfy \cref{eq:cond}. Note that applying this theorem requires a choice of basis for the center of $k$-local $G$-invariant Hamiltonians. 

We applied our general results to qubits systems with on-site representations of $G = \U(1), \SU(2)$ and $\mathbb{Z}_p$ symmetries and showed that in the cases of $k$-qubit gates with $\U(1)$ and $\SU(2)$ symmetries, $t_{\text{max}}$ grows as $n^{\lfloor\frac{k+1}{2}\rfloor}$ 
and $n^{\lfloor\frac{k}{2}\rfloor+1}$ respectively. In the case of $\mathbb{Z}_p$ symmetry, we found that when $p$ is even  $t_{\text{max}}=2^{n-1}-1$, whereas it is infinite for odd $p$. We also studied $\SU(d)$ symmetry on qudit systems for gates acting on $k=2,3,4$ qudits.

In the case of $G=\U(1)$ and $\SU(2)$ symmetries, we introduced new bases 
$\{F_k\}$ and $\{A_k\}$ respectively for the center of $G$-invariant Hamiltonians with certain remarkable properties, namely restricted support in both frequency and spatial domains (See Figs. \ref{domains}, \ref{fig:dual-domains-U1} and \ref{fig:dual-domains-SU2}). This makes them useful for applications beyond the questions studied in this paper, e.g., in the context of Hamiltonian learning. In particular, they are useful for detecting symmetry-protected signatures of $k$-body interactions in experiments \cite{zhukas2024observation}.

\vskip 10pt

\noindent\emph{Note Added:} During the preparation of this article, we became aware of independent work by Yosuke Mitsuhashi, Ryotaro Suzuki, Tomohiro Soejima, and Nobuyuki Yoshioka, which studies similar questions and was posted on arXiv concurrently with the present paper \cite{Mitsuhashi:2024lti,Mitsuhashi:2024qlx}.  Our general result in \cref{thm:tdesign}  is equivalent to the general theorem of \cite{Mitsuhashi:2024lti,Mitsuhashi:2024qlx}, and the final results on $t_{\max}$ for $\U(1)$ and $\SU(2)$ examples are identical.  However, as discussed in \cref{Sec:kqudit}, there is freedom in applying this theorem to each example; namely, it requires a choice of basis for the center of the Lie algebra, which is indeed different in the two papers.  Each choice will result in different combinatorial identities. For the choice of bases made in the present paper, i.e., operators $\{C_l\}$  in \cref{eq:u1cls,eq:Cl-SU2} for  $\U(1)$ and $\SU(2)$ respectively, the final results rely on two combinatorial identities, namely \cref{eq:conj-Fk,eq:conj-Ak}, which were proposed as two conjectures Eq. (86) and Eq. (120) in the arXiv version v1 of the present paper. In this current version (v2), we prove stronger results, namely \cref{eq:TrFkCl,eq:TrA2sC2m} that imply \cref{eq:conj-Fk,eq:conj-Ak}, using standard techniques such as generating functions.  

Our approach uncovered new operator bases $\{F_k\}$ in the case of U(1) and $\{A_k\}$ in the case of SU(2), with various nice properties discussed in the introduction, which make them useful for other applications, such as Hamiltonian learning \cite{zhukas2024observation}.  In addition to $\U(1)$ and $\SU(2)$, our paper also finds the exact $t_{\max}$ for $\mathbb{Z}_p$ symmetry with $k$-qubit gates for arbitrary $p$ and $k$, and $\SU(d)$ symmetry for arbitrary $d\geq 3$ and $k\leq 4$. Furthermore, we also present \cref{prop1}, which gives an easy-to-calculate lower bound on $t_{\max}$.\\

\section*{Acknowledgements}
I.M. acknowledges the helpful discussion with Yosuke Mitsuhashi.  This work is supported by a collaboration between the US DOE and other Agencies. This material is based upon work supported by the U.S. Department of Energy, Office of Science, National Quantum Information Science Research Centers, Quantum Systems Accelerator. Additional support is acknowledged from Army Research Office (W911NF-21-1-0005), NSF Phy-2046195, and NSF QLCI grant OMA-2120757. H.L. was supported by the Quantum Science Center (QSC), a National Quantum Information Science Research Center of the U.S. Department of Energy (DOE) and by the U.S. Department of Energy, Office of Science, Office of Nuclear Physics (NP) contract DE-AC52-06NA25396.

\bibliography{refs}

\onecolumngrid

\appendix

\newpage

\section*{Appendix}

\newcommand\appitem[2]{\hyperref[{#1}]{\textbf{\cref*{#1} \nameref*{#1}}} \dotfill \pageref{#1}\\ \begin{minipage}[t]{0.8\textwidth} #2\end{minipage}}

\begin{itemize}
\item \appitem{app:U1}{}
\item \appitem{app:SU2}{}
\item \appitem{app:Zp}{}
\item \appitem{app:SUd}{}
\item \appitem{App:fail}{}
\item \appitem{app:design}{}
\end{itemize}

\newpage

\newpage

\section{\texorpdfstring{$\U(1)$}{U(1)} symmetry}\label{app:U1}

In this section, we study the $\U(1)$ example in more detail, and prove the identities used to determine $t_{\max}$ in this case. In particular, we determine the value of $\tr F_k C_l$ and show that it is zero when $l < k$, and we also determine the norm $\|F_k\|_1$.
The following lemma from the main text summarizes these identities.

\lemFk*
It is worth mentioning the interesting fact that the coefficients of the expansion of $2^n F_k$ in terms of $C_l$ operators are independent of $n$.

The $F_k$ operators are constructed with the property that it has support only on the low dimensional irreps, resulting in an almost symmetric support on $\{\Pi_w\}$ basis with respect to flipping the irreps with Hamming weight $w$ and $n-w$, which is implemented by the $X^{\otimes n}$ operator. When $k$ is odd, the support of $F_k$ has an even number of elements, and a symmetric support is possible. On the other hand, when $k=2s$ is even, the support of $F_k$ has an odd number of elements, and we have to make a choice. For $F_{2s}$, the choice is to fill $s$ with $s<\frac{n}{2}$ first. Another equally good choice is to fill $n-s$ first. This motivates us to define a set of operators $\widetilde{F}_k$ as
\begin{align}
  \widetilde{F}_k = (-1)^k X^{\otimes n} F_k X^{\otimes n} =\sum_{w = 0}^n (-1)^w \binom{n-\lfloor\frac{k}{2}\rfloor -w -1}{n-k} \Pi_{w}.
\end{align}
Using the identities
\begin{subequations}\label{eq:Xn}
\begin{align}
  X^{\otimes n} \Pi_w X^{\otimes n} &= \Pi_{n-w}\\
  X^{\otimes n} C_l X^{\otimes n} &= (-1)^l C_l ,
\end{align}
\end{subequations}
One can easily check that $\widetilde{F}_k$ is related to $F_k$ and $C_l$ as
\begin{align}
  \widetilde{F}_k = (-1)^k X^{\otimes n} F_k X^{\otimes n} =
  \begin{cases}
      2^{k-n} \sum\limits_l (-1)^l \binom{\lfloor l/2\rfloor}{\lfloor k/2\rfloor} C_l & \text{$k:$  even} \\
      2^{k-n} \sum\limits_{l:\mathrm{odd}} \binom{\lfloor l/2\rfloor}{\lfloor k/2\rfloor} C_l & \text{$k:$  odd }.
  \end{cases}
\end{align}
and in particular, when $k$ is odd we have
\begin{align}
  \widetilde{F}_k = F_k .
\end{align}

Note that $\|F_k\|_1 = \|\widetilde{F}_k\|_1$ for all $k$ and $\tr \widetilde{F}_k C_l = 0$ for all $l < k$. Thus, when $k$ is odd, $\widetilde{F}_{k + 1}$ gives an independent operator which acts as $Q$, but leads to the same value of $t_{\max}$ for $\mathcal{V}_{n, k}^{\U(1)}$.

In addition, we also have the relations
\begin{subequations}\label{eq:Zn}
\begin{align}
  Z^{\otimes n} \Pi_w &= (-1)^w \Pi_w, \\
  Z^{\otimes n} C_l &= C_{n-l} \label{eq:Znb}.
\end{align}
\end{subequations}
By sandwiching with $X^{\otimes n}$ and $Z^{\otimes n}$ respectively, and using \cref{eq:Xn,eq:Zn}, we can easily see \cref{eq:mirror-symmetry}:
\begin{align}\tag{re \ref{eq:mirror-symmetry}}
\Tr(\Pi_w C_l)=(-1)^l \Tr(\Pi_{n-w} C_l)=(-1)^w \Tr(\Pi_{w} C_{n-l}) \ .
\end{align}

\begin{comment}
\begin{subequations} 
\begin{align}
X^{\otimes n} A_k X^{\otimes n}&=??\\
\end{align}
\end{subequations}
\end{comment}

\subsection{Some useful identities involving the binomial coefficients}
In this section, we gather some useful facts about the binomial coefficients that will be used repeatedly in our proofs.

Recall that for arbitrary complex number $\alpha$ and non-negative integers $k$, the binomial coefficient $\binom{\alpha}{k}$ is defined by 
\begin{align}\tag{re \ref{eq:binom-def}}
    \sum_{k=0}^\infty \binom{\alpha}{k} x^k = (1+x)^\alpha.
\end{align}
Another interpretation of this equation is that $(1+x)^\alpha$ can also be viewed as the generating function of the binomial coefficients in the lower index. Now we derive the generating function of $\binom{n}{k}$ in the upper index. Let $\alpha=n$ be a non-negative integer in \cref{eq:binom-def}, multiply both sides by $y^n$, sum over $n$ on both sides of \cref{eq:binom-def}, and we have
\begin{align}
  \sum_{n=0}^\infty \sum_{k=0}^\infty \binom{n}{k} x^k y^n = \sum_{n=0}^\infty (1+x)^n y^n = \frac{1}{1-y-xy} = \frac{1}{1-y} \frac{1}{1-x\frac{y}{1-y}} = \sum_{k=0}^\infty \frac{y^k}{(1-y)^{k+1}} x^k.
\end{align}
Comparing the coefficient of $x^k$ on both sides, we have the generating function of $\binom{n}{k}$ in the upper index $n$,
\begin{align}\label{eq:gen-binom}
  \sum_{n=0}^\infty \binom{n}{k} y^n = \frac{y^k}{(1-y)^{k+1}} .
\end{align}

Recall the equivalent definition of the binomial coefficients in \cref{eq:binom}
\begin{align}\tag{re \ref{eq:binom}}
    \binom{\alpha}{k} := \frac{1}{k!}\Big(\frac{\d}{\d x} \Big)^k (1+x)^\alpha \Big|_{x=0} = \frac{(\alpha)_k}{k!}\ .
\end{align}
As we noted in the main text, $\binom{\alpha}{k}$ is zero if and only if $0$ appears in the falling factorial, which is the case when $\alpha$ is a non-negative integer and $k>\alpha$. When $\alpha<0$, $\binom{\alpha}{k}$ is always non-vanishing, and using the definition we can rewrite it in terms of only positive indices as
\begin{align}\label{eq:extended-binomial}
    \binom{\alpha}{k} %= \frac{(\alpha-k+1)^k}{k!} 
    = (-1)^k \binom{-\alpha+k-1}{k} .
\end{align}
Using \cref{eq:binom}, one can also easily verify
\begin{align}\label{eq:binomder}
  \binom{\alpha}{k}x^{\alpha-k} = \frac{1}{k!} \Big(\frac{\d}{\d x}\Big)^k x^\alpha.
\end{align}

Besides, the binomial coefficients also appear in the general Leibniz rule that we will use later,
\begin{align}\label{eq:general-Leibniz}
    \Big(\frac{\d}{\d x}\Big)^n \big(f(x)g(x)\big) = \sum_{k=0}^n \binom{n}{k} \Big(\big(\frac{\d}{\d x}\big)^k f(x) \Big) \Big(\big(\frac{\d}{\d x}\big)^{n-k} g(x) \Big) .
\end{align}

\subsection{\texorpdfstring{$k=1,2$}{k=1,2} examples}

We briefly outline how to verify that $F_2 = (n-1) \Pi_0 - \Pi_1 + \Pi_n$ and $F_3 = (n-1)(\Pi_0 - \Pi_n) - (\Pi_1 - \Pi_{n - 1})$ satisfy the conditions of \cref{thm:tdesign} for $G_{XX + YY, Z}$ and $\mathcal{V}_{n, 2}^{\U(1)}$, respectively. In particular, we show that $F_2$ is orthogonal to all 1-local operators, and $F_3$ is orthogonal to all 2-local operators.

Note that, because the symmetry group $\U(1)$ is Abelian, $\tr \Pi_w = m_w$. Hence, $\tr (F_2 \Pi_0) / \Tr \Pi_0 = n - 1$, $\tr (F_2 \Pi_1) / \tr \Pi_1 = -1$, and $\tr (F_2 \Pi_n) / \tr \Pi_n = 1$ are integer, with each other $\tr F_2 \Pi_w = 0$, and $\tr F_3 = 0$. For the generators of $G_{XX + YY, Z}$, we find
\begin{equation}
  \begin{split}
    \tr F_2 Z_i & = (n - 1) \tr \Pi_0 Z_i - \tr \Pi_1 Z_i + \tr \Pi_n Z_i \\
              & = n - 1 - (n - 2) - 1 = 0,
  \end{split}
\end{equation}
and $\tr F_2 (X_i X_j + Y_i Y_j) = 0$ since the $XY$ interaction is centerless, i.e., {$\tr \Pi_w  (X_i X_j + Y_i Y_j) = \tr C_l  (X_i X_j + Y_i Y_j) = 0$ for all $w,l\in\{0,\cdots, n\}$.}

Now for $\mathcal{V}_{n,2}^{\U(1)}$ we consider $F_3 = (n - 1) (\Pi_0 - \Pi_n) - (\Pi_1 - \Pi_{n - 1})$, which again has integer eigenvalues (and since the symmetry is Abelian, means its corresponding $\mathbf{q}$ vector is integer). Further, it is easy to see that $\tr F_3 = 0$, $\tr F_3 Z_i = 0$, and $\tr F_3 Z_i Z_j = 0$ using
\begin{equation}
  \begin{split}
    \tr Z_i Z_j \Pi_0 & = \tr Z_i Z_j \Pi_n = 1 \\
    \tr Z_i Z_j \Pi_1 & = \tr Z_i Z_j \Pi_{n - 1} = n - 1.
  \end{split}
\end{equation}

\subsection{Properties of operator \texorpdfstring{$C_l$}{Cl}}
In this subsection, we further study $\{C_l\}$ basis and
 its relation with $\{\Pi_w\}$ basis. In particular, we prove \cref{eq:reciprocity-U1} and the following proposition on $\bar{c}$.

Recall that in \cite{Marvian2022Restrict}, the operator $C_l$ is defined and calculated as 
\begin{align}\tag{re \ref{eq:u1cls}}
  C_l &:= \sum_{\mathbf{b} : \ell(\mathbf{b}) = l} \mathbf{Z}^{\mathbf{b}} = \sum_w c_l(w) \Pi_w,
\end{align}
where
\begin{align}\tag{re \ref{Eq:clw}}
  c_l(w) = \sum_{r = 0}^w (-1)^r \binom{n-w}{l-r} \binom{w}{r}.
\end{align}

We first normalize $\{\Pi_w\}$ and $\{C_l\}$ to $\{\overline\Pi_w\}$ and $\{\overline C_l\}$ with $\Tr\overline\Pi_w^2 = \Tr\overline C_l^2 = \Tr\id = 2^n$, 
\begin{subequations}\label{eq:Pi-C-normalized-U1}
  \begin{align}
    \overline\Pi_w &:= 2^{\frac{n}{2}} \big(\Tr\Pi_w\big)^{-\frac{1}{2}} \Pi_w = 2^{\frac{n}{2}} \binom{n}{w}^{-\frac{1}{2}} \Pi_w, \\
    \overline C_l &:= 2^{\frac{n}{2}} \big(\Tr C_l^2\big)^{-\frac{1}{2}} C_l = \binom{n}{l}^{-\frac{1}{2}} C_l.
  \end{align}
\end{subequations}
Then we have the following proposition
\begin{proposition}\label{prop:cbar-U1}
The two orthonormal bases $\{\overline\Pi_w\}$ and $\{\overline C_l\}$ are related as
\begin{subequations}
  \begin{align}
   \overline C_l &= \sum_w \bar{c}_{l,w} \overline\Pi_w, \tag{re \ref{eq:reciprocity-U1a}}\\
   \overline\Pi_w &= \sum_l \bar{c}_{l,w} \overline C_l, \tag{re \ref{eq:reciprocity-U1b}}
  \end{align}
\end{subequations}
where $\bar{c}$ is an $(n+1)\times (n+1)$ matrix  with matrix elements
\begin{align}
  \bar{c}_{l,w} = 
  \sqrt{\frac{\Tr(\Pi_w)}{\Tr(C_l^2)}} c_l(w) = 
  2^{-\frac{n}{2}} \sqrt{\frac{\binom{n}{w}}{\binom{n}{l}}} c_l(w) .
\end{align}
Furthermore, $\bar{c}$ is symmetric and orthogonal, namely,
\begin{align}\tag{re \ref{eq:cbar-symmetric-orthogonal}}
  \bar{c} = \bar{c}^\T = \bar{c}^{-1}.
\end{align}
Similarly, the matrix $\Tr(\Pi_w C_l)$ is also symmetric in $w$ and $l$, i.e., $\Tr(\Pi_w C_l)=\Tr(\Pi_l C_w)$.
\end{proposition}

We first calculate the generating function for $C_l$, i.e., an operator-valued function in $y$, that will be useful on several occasions.
\begin{align}\label{eq:Cl-Fourier}
  \sum_l C_l y^l &= \sum_w \sum_{r = 0}^w (-1)^r y^r \binom{w}{r} \Pi_w \sum_l \binom{n-w}{l-r} y^{l-r} = \sum_w \sum_{r = 0}^w (-1)^r y^r (1+y)^{n-w} \binom{w}{r} \Pi_w \nonumber\\
  &= \sum_w (1-y)^w (1+y)^{n-w} \Pi_w .
\end{align}

The first application of this generating function is to calculate $\Tr(C_l^2)$. The square of \cref{eq:Cl-Fourier} is
\begin{align}
  \Big(\sum_l C_l y^l\Big)^2 &= \sum_w (1-y)^{2w} (1+y)^{2(n-w)} \Pi_w .
\end{align}
Taking trace on the left-hand side, we have
\begin{align}
  \Tr \Big(\sum_l C_l y^l\Big)^2 = \sum_{l,m} \Tr(C_l C_m) y^l y^m = \sum_l \Tr(C_l^2) y^{2l}, 
\end{align}
where in the last step, we used the orthogonality of the $C_l$ operators. Taking the trace on the right-hand side, we have
\begin{align}
  \Tr \sum_w (1-y)^{2w} (1+y)^{2(n-w)} \Pi_w &= \sum_w \binom{n}{w} (1-y)^{2w} (1+y)^{2(n-w)} \nonumber\\
  &= (1+y)^{2n} \sum_w \binom{n}{w} \Big( \frac{1-y}{1+y} \Big)^{2w} = 2^n(1+y^2)^n = \sum_{l} 2^n \binom{n}{l} y^{2l}
\end{align}
Comparing the two sides of the equation, we know
\begin{align}
  \Tr(C_l^2) = 2^n \binom{n}{l}.
\end{align}

Note that this relation was already shown in \cite{Marvian2022Restrict} using the fact that $\Tr(\mathbf{Z}^{\mathbf{b}} \mathbf{Z}^{\mathbf{b}'})=2^n \delta_{\mathbf{b}, \mathbf{b}'}$, and the fact that the summation in \cref{eq:u1cls} contains $\binom{n}{l}$ such terms. %However, as a warm-up for calculating $\Tr(C_l^2)$ for $\SU(2)$ in the next section, we show how $\Tr(C_l^2)$ can also be found using the moment-generating function technique. This also makes it relatively easy to see why \cref{eq:PiwCl} is true.

Now we are ready to prove \cref{prop:cbar-U1}.
\begin{proof}
\cref{eq:reciprocity-U1a} is clear from \cref{eq:u1cls}, and here we prove \cref{eq:reciprocity-U1b}. Since both $\Pi_w$ and $C_l$ form an orthogonal basis of the center of $\U(1)$ invariant Hamiltonians, we can express $\Pi_w$ in the $C_l$ basis, i.e., the inverse of \cref{eq:u1cls}. Multiplying by $\Pi_w$ on both sides of \cref{eq:u1cls} and taking trace, we have
\begin{align}
  \Tr(C_l\Pi_w) = \binom{n}{w} \sum_{r = 0}^w (-1)^r \binom{n-w}{l-r} \binom{w}{r}.
\end{align}
Therefore
\begin{align}
  \Pi_w = \sum_l\frac{\Tr(\Pi_w C_l)}{\Tr(C_l^2)} C_l = 2^{-n} \binom{n}{w} \sum_l \binom{n}{l}^{-1} \sum_{r = 0}^w (-1)^r \binom{n-w}{l-r} \binom{w}{r} C_l .
\end{align}
This implies
\begin{align}
  \overline{\Pi}_w = 2^{-\frac{n}{2}} \sum_l \Big(\binom{n}{w}\Big/\binom{n}{l}\Big)^{\frac{1}{2}} \sum_{r = 0}^w (-1)^r \binom{n-w}{l-r} \binom{w}{r} \overline{C}_l = \sum_l \bar{c}_{l,w} \overline C_l.
\end{align}

From \cref{eq:reciprocity-U1}, it is clear that $\bar{c}$ is orthogonal, i.e., $\bar{c}^\T \bar{c} = \id$. Now we prove that $\bar{c}$ is also symmetric, i.e., $\bar{c}^\T = \bar{c}$. This is equivalent to proving that $\Tr(C_l\Pi_w) = \Tr(C_w\Pi_l)$, namely
\begin{align}
  \binom{n}{w} \sum_{r = 0}^w (-1)^r \binom{n-w}{l-r} \binom{w}{r} = \binom{n}{l} \sum_{r = 0}^l (-1)^r \binom{n-l}{w-r} \binom{l}{r},
\end{align}
Consider the generating function,
\begin{align}
  \sum_{w,l} \Tr(C_l\Pi_w) x^w y^l &= \sum_{w} \binom{n}{w} (1-y)^w (1+y)^{n-w} x^w \nonumber\\
  &= (1+y)^n \sum_{w} \binom{n}{w} \Big(\frac{x(1-y)}{1+y}\Big)^w \nonumber\\
  &= (1+y)^n \Big(1+ \frac{x(1-y)}{1+y}\Big)^n \nonumber\\
  &= (1+x+y-xy)^n 
\end{align}
which is symmetric in $x$ and $y$. This implies $\sum_{w,l} \Tr(C_l\Pi_w) x^w y^l = \sum_{w,l} \Tr(C_l\Pi_w) y^w x^l = \sum_{w,l} \Tr(C_w\Pi_l) x^w y^l$, and thus $\Tr(C_l\Pi_w) = \Tr(C_w\Pi_l)$.

%\begin{align}\label{eq:c-lw}
%  \bar{c}_{l,w} := \sqrt{\frac{\Tr(\Pi_w)}{\Tr(C_l^2)}} c_l(w) =   2^{-\frac{n}{2}} \sqrt{\frac{\binom{n}{w}}{\binom{n}{l}}} c_l(w)
%  = 2^{-\frac{n}{2}} \Big(\binom{n}{w}\Big/\binom{n}{l}\Big)^{\frac{1}{2}} \sum_{r = 0}^w (-1)^r \binom{n-w}{l-r} \binom{w}{r}.
%\end{align}
%\cref{eq:reciprocity-U1} implies that the matrix $[\bar{c}_{l,w}]$ is involutory. Remarkably, this matrix is also symmetric, and thus orthogonal. This can be seen as follows. 
\end{proof}

There is another interesting way to see why $\bar{c}$ is symmetric. We make the substitution $y = \frac{1-z}{1+z}$ in \cref{eq:Cl-Fourier} (note that this defines $z$ so that it is symmetric with $y$ in the sense that they satisfy $(1+y)(1+z)=2$),
\begin{align}
  \sum_l C_l \Big(\frac{1-z}{1+z}\Big)^l &= (1+y)^n \sum_w \Big(\frac{1-y}{1+y}\Big)^w \Pi_w \nonumber\\
  2^{-n} \sum_l C_l (1-z)^l (1+z)^{n-l} &= \sum_w z^w \Pi_w .
\end{align}
Note the similarity between this equation and \cref{eq:Cl-Fourier}. Therefore we know
\begin{align}
  \Pi_w = 2^{-n}\sum_l c_w(l) C_l \quad \implies \quad \overline\Pi_w = \sum_l \bar{c}_{w,l} \overline C_l.
\end{align}
Comparing this equation with \cref{eq:reciprocity-U1b}, we know $\bar{c}_{l,w} = \bar{c}_{w,l}$.

\subsection{Properties  of operator \texorpdfstring{$A_k$}{Ak}}

Recall the operator $A_k$ defined in \cref{Ak},
\begin{equation}\tag{re \ref{Ak}}
    A_k = \sum_{w = 0}^k (-1)^w \binom{n-w}{k-w} \Pi_w.
\end{equation}
In this section, we prove several properties of $A_k$, including,
\begin{align}\label{eq:gen-conj-Ak-U1}
    \Tr(A_k C_l) = \sum_{w = 0}^n (-1)^w \binom{n-w}{k-w} \binom{n}{w} \sum_{r = 0}^w (-1)^r \binom{n-w}{l-r} \binom{w}{r}  = 2^k \binom{n}{l} \binom{l}{k} ,
\end{align}
and 
\begin{equation}\label{eq:self-inv}
    \Pi_w = \sum_{k = 0}^w (-1)^k \binom{n - k}{w - k} A_k.
\end{equation}
The second equation means that the coefficient matrix of the $A_k$ operators in terms of the projectors $\Pi_w$ is self-inverse.
%\red{Take the trace of both sides of A17 with $C_l$. then A16 should imply A17}
Both identities can be shown by considering the generating functions of the $A_k$ operators,
\begin{align}\label{eq:Ak-Fourier}
  \sum_k A_k x^k = \sum_{w = 0}^k (-x)^w \sum_k \binom{n-w}{k-w} x^{k-w} \Pi_w = \sum_{w = 0}^k (-x)^w (1+x)^{n-w} \Pi_w.
\end{align}
Taking trace on the product of the two generating functions $\sum_k A_k x^k$ and $\sum_l C_l y^l$, we have
\begin{align}
  \sum_{k=0}^n \sum_{l=0}^n \Tr(A_k C_l) x^k y^l &= \sum_{w = 0}^n (-x)^w (1+x)^{n-w} (1-y)^w (1+y)^{n-w} \binom{n}{w} \nonumber\\
  &= (1+x)^n (1+y)^n \sum_{w = 0}^n \binom{n}{w} \Big(\frac{-x}{1+x} \frac{1-y}{1+y}\Big)^w  \nonumber\\
  &= (1+x)^n (1+y)^n \Big(1- \frac{x}{1+x} \frac{1-y}{1+y}\Big)^n  \nonumber\\
  &= \big(1+y+2xy\big)^n  \nonumber\\
  &= \sum_{l=0}^n \binom{n}{l} \big(1+2x\big)^l y^l  \nonumber\\
  &= \sum_{l=0}^n \sum_{k=0}^l 2^k \binom{n}{l} \binom{l}{k} x^k y^l,
\end{align}
where we sum over $w$ in the third line, and then perform binomial expansions in the last two lines. Therefore, matching coefficients of the monomials $x^k y^l$, we have
\begin{align}
  \Tr(A_k C_l) = 2^k \binom{n}{l} \binom{l}{k} .
\end{align}

Making the substitution $x = -\frac{z}{1+z}$ in \cref{eq:Ak-Fourier} (similar to the $C_l$ operators, $x$ and $z$ are symmetric in the sense that they satisfy $(1+x)(1+z) = 1$), we have
\begin{align}
  \sum_k A_k (-z)^k (1+z)^n = \sum_{w = 0}^k z^w \Pi_w.
\end{align}
Again, noting the similarity between this equation and \cref{eq:Ak-Fourier}, we conclude
\begin{equation}\tag{re \ref{eq:self-inv}}
    \Pi_w = \sum_{k = 0}^w (-1)^k \binom{n - k}{w - k} A_k.
\end{equation}

Using \cref{eq:Cl-orthogonal-U1}, we know that $A_k$ can be expanded in the $C_l$ basis as
\begin{align}
  A_k = \sum_l \frac{\Tr(A_kC_l)}{\Tr(C_l^2)} C_l = 2^{k-n} \sum_l \binom{l}{k} C_l\ .
\end{align}
Since the eigenvalues of $A_k$ have alternating signs $(-1)^w$ and $C_n = Z^{\otimes n}$, we know that
\begin{align}
    \| A_k \|_1 = \Tr(A_k C_n) = 2^k \binom{n}{k} .
\end{align}
Alternatively, we can calculate directly from the definition,
\begin{equation}
\begin{split}
    \|A_k\|_1 &= \sum_{w = 0}^k \binom{n-w}{k-w} \tr \Pi_w = \sum_{w=0}^k \binom{n-w}{k-w} \binom{n}{w} = \binom{n}{k} \sum_{w = 0}^k \binom{k}{w} = \binom{n}{k} 2^{k}
\end{split}
\end{equation}
where we use the identity
\begin{equation}
    \binom{n-w}{k-w}\binom{n}{w} = \binom{n}{k}\binom{k}{w}.
\end{equation}

\subsection{Properties of operator \texorpdfstring{$F_k$}{Fk}}

Recall the operator $F_k$ defined in \cref{eq:Fk},
\begin{equation}\tag{re \ref{eq:Fk}}
  F_k = \sum_{w = 0}^n (-1)^{w} \binom{n - \floor{\frac{k + 1}{2}} - w}{\floor{\frac{k}{2}} - w} \Pi_w.
\end{equation}
The coefficients are tabulated for small $k$ in \cref{fig:Fk-coef}. Note that the fact that the support of the $F_k$ operators on $\set{\Pi_w}$ is strictly increasing for in $k = 0, \cdots, n$ implies that they are linearly independent.

\begin{table}[ht]
    \begin{tblr}{ccccccccc}
        \diagbox{$F_k$}{$\Pi_w$} & $0$ & $1$ & $2$ & $3$ & $\cdots$ & $n - 2$ & $n - 1$ & $n$ \\ \midrule
        $F_3$ & $n - 2$ & $-1$ & $0$ & $0$ & $\cdots$ & $0$ & $1$ & $2 - n$ \\
        $F_4$ & $\frac{(n - 3)(n - 2)}{2}$ & $3 - n$ & $1$ & $0$ & $\cdots$ & 0 & $-1$ & $n - 3$ \\
        $F_5$ & $\frac{(n - 4)(n - 3)}{2}$ & $4 - n$ & $1$ & $0$ & $\cdots$ & $-1$ & $n - 4$ & $-\frac{(n-4)(n-3)}{2}$ \\
        $F_6$ & $\frac{(n-5)(n-4)(n-3)}{6}$ & $-\frac{(n-5)(n-4)}{2}$ & $n-5$ & $-1$ & $\cdots$ & $1$ & $5 - n$ & $\frac{(n-5)(n-4)}{2}$
    \end{tblr}
    \caption{The coefficients of $F_k$ in the basis of projectors $\Pi_w$ for $k = 3, \cdots, 6$.}\label{fig:Fk-coef}
\end{table}

In the first version of this work on arXiv, we proposed a combinatorial conjecture as Eq. (86) that 
\begin{align}\label{eq:conj-Fk}
  \Tr(F_k C_l) = \sum_{w = 0}^n (-1)^{w} \binom{n - \lfloor\frac{k + 1}{2}\rfloor - w}{\lfloor\frac{k}{2}\rfloor - w} \binom{n}{w} \sum_{r = 0}^w (-1)^r \binom{n-w}{l-r} \binom{w}{r} = 0 \quad \text{when } l<k.
\end{align}
Here, we prove a stronger result for all $l$ and $k$,
\begin{align}\label{eq:TrFkCl}
  \Tr(F_k C_l) &= \sum_{w = 0}^n (-1)^{w} \binom{n - \lfloor\frac{k + 1}{2}\rfloor - w}{n-k} \binom{n}{w} \sum_{r = 0}^w (-1)^r \binom{n-w}{l-r} \binom{w}{r} \nonumber\\ 
  &= \Big(1- \frac{(1-(-1)^k) (1+(-1)^l)}{4} \Big) 2^k \binom{n}{l} \binom{\floor{\frac{l}{2}}}{\floor{\frac{k}{2}}} =
  \begin{cases}
      2^k \binom{n}{l} \binom{\lfloor l/2\rfloor}{k/2} & \text{when $k$ is even} \\
      \frac{1-(-1)^l}{2} 2^k \binom{n}{l} \binom{(l-1)/2}{(k-1)/2} & \text{when $k$ is odd}.
  \end{cases}
\end{align}
This identity implies \cref{eq:conj-Fk}: when $\lfloor\frac{l}{2}\rfloor < \lfloor\frac{k}{2}\rfloor$, $\Tr(F_k C_l)=0$ because $\binom{\floor{\frac{l}{2}}}{\floor{\frac{k}{2}}} = 0$; when $k$ is odd and $l=k-1$, $\binom{\floor{\frac{l}{2}}}{\floor{\frac{k}{2}}}=1$, but still $\Tr(F_k C_l)=0$ due to the factor $\frac{1-(-1)^l}{2} = 0$. In summary, we know $\Tr(F_k C_l)=0$ when $l<k$.

\begin{comment}
\color{red}
Check? For odd k?
\begin{align}
F_k&=\frac{4k(k-2)}{2^n k!!} \frac{(2\lfloor{k/2}\rfloor-2)!}{\lfloor{k/2\rfloor}!(\lfloor{k/2}\rfloor-1)!}\sum_{l:\mathrm{odd}} \frac{(l-1)!!}{(l-k)!!}  C_l \\ &=\frac{4k(k-2)}{2^n k!!} \frac{(2\lfloor{k/2}\rfloor-2)!}{\lfloor{k/2\rfloor}!(\lfloor{k/2}\rfloor-1)!}{2^{(k-1)/2}} \sum_{l:\mathrm{odd}} \frac{(\frac{l-1}{2})!}{(\frac{l-k}{2})!} C_l
\end{align}

$$\sum_{l:\mathrm{odd}}^n c_l(m) \times \frac{(\frac{l-1}{2})!}{(\frac{l-k}{2})!} \times  {{n}\choose{m}}= 0  \mathrm{\ \  :   if }   k<m<n-k   $$

\color{black}
\end{comment}

\begin{proof}
  Using \cref{eq:Cl-Fourier}, we have
  \begin{align}
  \sum_{l=0}^{n} \Tr(F_{k} C_{l}) y^l %&= \sum_{w = 0}^n (-1)^{w} \binom{n - \lfloor\frac{k + 1}{2}\rfloor - w}{n-k} \binom{n}{w} \sum_{r = 0}^w (-1)^r \Big(\sum_{l=0}^n \binom{n-w}{l-r}y^l \Big) \binom{w}{r} \nonumber\\
  %&= \sum_{w = 0}^n (-1)^{w} \binom{n - \lfloor\frac{k + 1}{2}\rfloor - w}{n-k} \binom{n}{w} (1-y)^w (1+y)^{n-w} \nonumber\\
  &= (1+y)^n \sum_{w = 0}^n (-1)^{w} \binom{n - \lfloor\frac{k + 1}{2}\rfloor - w}{n-k} \binom{n}{w} u^{-w} ,
\end{align}
where introduce $u = \frac{1+y}{1-y}$.
Using \cref{eq:binomder}, we have
\begin{align}\label{eq:gen-l}
  \sum_{l=0}^{n} \Tr(F_{k} C_{l}) y^l &= (1+y)^n u^{-\lfloor k/2\rfloor} \sum_{w = 0}^n (-1)^{w} \binom{n}{w} \binom{n - \lfloor\frac{k + 1}{2}\rfloor - w}{n-k} u^{\lfloor k/2\rfloor-w} \nonumber\\
  &= (1+y)^n u^{-\lfloor k/2\rfloor} \sum_{w = 0}^n (-1)^{w} \binom{n}{w} \frac{1}{(n-k)!} \Big(\frac{\d}{\d u}\Big)^{n-k} u^{n - \lfloor\frac{k + 1}{2}\rfloor - w} \nonumber\\
  &= (1+y)^n u^{-\lfloor k/2\rfloor} \frac{1}{(n-k)!} \Big(\frac{\d}{\d u}\Big)^{n-k} u^{- \lfloor\frac{k + 1}{2}\rfloor} (u-1)^n \nonumber\\
  &= (1+y)^n u^{-\lfloor k/2\rfloor} \sum_{i=0}^{n-k}\frac{1}{i!} \Big(\frac{\d}{\d u}\Big)^i u^{- \lfloor\frac{k + 1}{2}\rfloor} \frac{1}{(n-k-i)!} \Big(\frac{\d}{\d u}\Big)^{n-k-i} (u-1)^n \nonumber\\
  &= (1+y)^n u^{-\lfloor k/2\rfloor} \sum_{i=0}^{n-k} \binom{- \lfloor\frac{k + 1}{2}\rfloor}{i} u^{- \lfloor\frac{k + 1}{2}\rfloor-i} \binom{n}{k+i} (u-1)^{k+i} \nonumber\\
  &= (1+y)^n \sum_{i=0}^{n-k} (-1)^i \binom{\lfloor\frac{k + 1}{2}\rfloor+i-1}{i} \binom{n}{k+i} \Big(\frac{2y}{1+y}\Big)^{k+i} ,
\end{align}
where we sum over $w$ in the third line, use the general Leibniz rule \cref{eq:general-Leibniz} in the fourth line, use the relation \cref{eq:binomder} in the fifth line, and in the last line we subsitute back $\frac{1+y}{1-y} = u$ and replace the negative binomial coefficient with its positive equivalent. Let's introduce the variable $z$ for the index $n$ and, noting that the sum over $i$ in \cref{eq:gen-l} can be replaced by an unbounded sum due to the binomial coefficient $\binom{n}{k + i}$, consider the generating function
\begin{align}\label{eq:gen-FkCln}
  \sum_{n = 0}^\infty \sum_{l = 0}^n \Tr(F_{k} C_{l}) y^l z^n &= \sum_{i = 0}^\infty (-1)^i \binom{\lfloor\frac{k + 1}{2}\rfloor+i-1}{i} \Big(\frac{2y}{1+y}\Big)^{k+i} \sum_n \binom{n}{k+i} z^n(1+y)^n \nonumber\\
  &= \sum_i (-1)^i \binom{\lfloor\frac{k + 1}{2}\rfloor+i-1}{\lfloor\frac{k + 1}{2}\rfloor-1} \Big(\frac{2yz}{1-z(1+y)}\Big)^{k+i} \frac{1}{1-z(1+y)} \nonumber\\
  &= (-1)^{1-\lfloor\frac{k + 1}{2}\rfloor}\Big(\frac{2yz}{1-z(1+y)}\Big)^{\lfloor\frac{k}{2}\rfloor+1} \frac{1}{1-z(1+y)} \sum_i \binom{\lfloor\frac{k + 1}{2}\rfloor+i-1}{\lfloor\frac{k + 1}{2}\rfloor-1} \Big(\frac{-2yz}{1-z(1+y)}\Big)^{\lfloor\frac{k + 1}{2}\rfloor+i-1} \nonumber\\
  &= \Big(\frac{2yz}{1-z(1+y)}\Big)^{\lfloor\frac{k}{2}\rfloor+1} \frac{1}{1-z(1-y)} \Big(\frac{2yz}{1-z(1-y)}\Big)^{\lfloor\frac{k + 1}{2}\rfloor-1} \nonumber\\
  &= (2yz)^k \Big(\frac{1}{1-z(1+y)}\Big)^{\lfloor\frac{k}{2}\rfloor+1}  \Big(\frac{1}{1-z(1-y)}\Big)^{\lfloor\frac{k + 1}{2}\rfloor} 
\end{align}
where in the second line we sum over $n$ using the generating function in \cref{eq:gen-binom}, the third line rearranges terms (recall that $k = \floor{\frac{k + 1}{2}} + \floor{\frac{k}{2}}$), and the fourth line sums over $i$ using
\begin{equation}\label{eq:gen-binomshift}
    \sum_{i = 0}^\infty \binom{i + \alpha}{\alpha} x^i = x^{-\alpha} \sum_{i = 0}^\infty \binom{i + \alpha}{\alpha} x^{i+\alpha} = x^{-\alpha} \sum_{j = \alpha}^\infty \binom{j}{\alpha} x^{j} = x^{-\alpha} \sum_{j = 0}^\infty \binom{j}{\alpha} x^{j} = \frac{1}{(1-x)^{\alpha+1}}.
\end{equation}

In the last line of \cref{eq:gen-FkCln}, we may expand the two factors in parenthesis as geometric series. Due to the leading term $(2yz)^k$, every term in the result will be of degree at least $k$ in $y$. Therefore, the coefficient of $y^l$ is $0$ when $l < k$, i.e. $\Tr(F_{k} C_{l})=0$.

When $k=2s+1$ is odd, we have
\begin{align}
    \sum_{s,l,n} \Tr(F_{2s+1} C_{l}) x^s y^l z^n &= \sum_s x^s (2yz)^{2s+1} \Big(\frac{1}{1-z(1+y)}\Big)^{s+1}  \Big(\frac{1}{1-z(1-y)}\Big)^{s+1} \nonumber\\
  &= \sum_s x^s (2yz)^{2s+1} \Big(\frac{1}{(1-z)^2-(yz)^2}\Big)^{s+1}  \nonumber\\
  &= \sum_s x^s (2yz)^{-1} \Big(\frac{4}{(\frac{1-z}{yz})^2-1}\Big)^{s+1}  \nonumber\\
  &= (2yz)^{-1} \frac{4}{(\frac{1-z}{yz})^2-1} \sum_s  \Big(\frac{4x}{(\frac{1-z}{yz})^2-1}\Big)^s  \nonumber\\
  &= \frac{2yz}{(1-z)^2-(4x+1)(yz)^2},
\end{align}
where in the fifth line we sum the geometric series.
Notice that the right-hand side is an odd function in $y$, and thus $\Tr(F_{k} C_{l}) = 0$ when $k$ is odd and $l$ is even. Expanding the denominator as a geometric series, we have
\begin{align}
  \sum_{s, n} \sum_{l = 0}^n \Tr(F_{2s+1} C_{l}) x^s y^l z^n &= \frac{2yz/(1-z)^2}{1-(4x+1)(\frac{yz}{1-z})^2} \nonumber\\
  &= \frac{2}{1-z} \sum_l (4x+1)^l(\frac{yz}{1-z})^{2l+1} \nonumber\\
  &= \frac{2}{1-z} \sum_{s,l} \binom{l}{s} (4x)^s y^{2l+1} (\frac{z}{1-z})^{2l+1} \nonumber\\
  &= \sum_{s,l,n} 2^{2s+1} \binom{n}{2l+1} \binom{l}{s}x^s y^{2l+1} z^n,
\end{align}
where in the third line we use the binomial theorem and the fourth uses \cref{eq:gen-binom} in the variable $z / (1 - z)$. Comparing the two sides, we have
\begin{align}
  \Tr(F_{k} C_{l}) &= \frac{1-(-1)^l}{2} 2^k \binom{n}{l} \binom{(l-1)/2}{(k-1)/2} .
\end{align}
When $k=2s$, we have
\begin{align}
  \sum_{s,n,l} \Tr(F_{2s} C_{l}) x^s y^l z^n &= \sum_s x^s (2yz)^{2s} \Big(\frac{1}{1-z(1+y)}\Big)^{s+1}  \Big(\frac{1}{1-z(1-y)}\Big)^s \nonumber\\
  &= \sum_s x^s (2yz)^{2s} \Big(\frac{1}{(1-z)^2-(yz)^2}\Big)^s \Big(\frac{1}{1-z(1+y)}\Big) \nonumber\\
  &= \sum_s \Big(\frac{4x}{(\frac{1-z}{yz})^2-1}\Big)^s \Big(\frac{1}{1-z(1+y)}\Big) \nonumber\\
  &= \frac{1-z+yz}{(1-z)^2-(4x+1)(yz)^2}
\end{align}
By performing a similar calculation, we have
\begin{align}
  \Tr(F_{k} C_{l}) &= 2^k \binom{n}{l} \binom{\lfloor l/2\rfloor}{k/2}.
\end{align}
\end{proof}
Note that $C_0 = \id$. Therefore we also know $\Tr(F_k) = 0$ for $k \geq 1$.

Using \cref{eq:Cl-orthogonal-U1}, we know that $F_k$ can be expanded in the $C_l$ basis as
\begin{align}
  F_k = \sum_l \frac{\Tr(F_kC_l)}{\Tr(C_l^2)} C_l
  &= 2^{k-n} \sum_l \Big(1- \frac{(1-(-1)^k) (1+(-1)^l)}{4} \Big) \binom{\lfloor\frac{l}{2}\rfloor}{\lfloor\frac{k}{2}\rfloor} C_l \nonumber\\
  &=
  \begin{cases}
      2^{k-n} \sum_l \binom{\lfloor l/2\rfloor}{k/2} C_l & \text{when $k$ is even} \\
      2^{k-n} \sum_l \frac{1-(-1)^l}{2} \binom{(l-1)/2}{(k-1)/2} C_l & \text{when $k$ is odd}.
  \end{cases}
\end{align}

Next, we show that
  \begin{align}\tag{re \ref{eq:Fk-norm}}
    \|F_{k}\|_1 &= \sum_{w = 0}^n \Bigg| \binom{n - \lfloor\frac{k + 1}{2}\rfloor - w}{n-k} \Bigg| \binom{n}{w}
    %\sum_{i=0}^{k+1} \binom{n-\lfloor\frac{k+1+i}{2}\rfloor} {\lfloor\frac{k-i}{2}\rfloor} \binom{n}{\lfloor\frac{i}{2}\rfloor} %= 2^{\lfloor\frac{k}{2}\rfloor} \frac{(n-\frac{1+(-1)^k}{2})!!}{(n-k-1)!!\lfloor\frac{k+1}{2}\rfloor!} = 2^k \binom{(n-\frac{1+(-1)^k}{2})/2}{\lfloor\frac{k+1}{2}\rfloor} \\
    = 
    \begin{cases}
    2^k \binom{n/2}{k/2} &\text{ when $k$ is even} \\
    2^k \binom{(n-1)/2}{(k-1)/2} &\text{ when $k$ is odd} 
    \end{cases}.
  \end{align}

\begin{proof}
First, note that
\begin{align}
  F_{k} & = \sum_{w = 0}^n (-1)^w \binom{n - \lfloor\frac{k + 1}{2}\rfloor - w}{n-k} \Pi_w \nonumber\\
  & = \sum_{w = 0}^{ \lfloor\frac{k}{2}\rfloor} (-1)^w \binom{n - \lfloor\frac{k + 1}{2}\rfloor - w}{n-k} \Pi_w + \sum_{w = n-\lfloor\frac{k + 1}{2}\rfloor+1}^n (-1)^{w+n-k} \binom{-\lfloor\frac{k}{2}\rfloor + w-1}{n-k} \Pi_w \nonumber\\
  & = \sum_{w = 0}^{ \lfloor\frac{k}{2}\rfloor} (-1)^w \binom{n - \lfloor\frac{k + 1}{2}\rfloor - w}{n-k} \Pi_w + \sum_{w = 0}^{\lfloor\frac{k - 1}{2}\rfloor} (-1)^{w-k} \binom{n-w-\lfloor\frac{k}{2}\rfloor-1}{n-k} \Pi_{n-w}.
\end{align}
All binomials in this equation are positive. Therefore when $n-k$ is even, the terms with Hamming weight $w$ have sign of $(-1)^w$. Therefore using \cref{eq:TrFkCl}, we know when $n-k$ is even,
\begin{align}
  \| F_{k} \|_1 & = \Tr (F_k Z^{\otimes n})  = \Tr (F_k C_n) = 
  \begin{cases}
      2^k \binom{n/2}{k/2} & \text{when $k$ is even} \\
      2^k \binom{(n-1)/2}{(k-1)/2} & \text{when $k$ is odd}.
  \end{cases}
\end{align}
Now we would like to extend this equality to $n-k$ odd. First consider $k=2s+1$, and we have shown that for odd $n$, the following is true
\begin{align}
  \|F_{2s}\|_1 = \sum_{w = 0}^s \binom{n - s-1 - w}{s-w} \binom{n}{w} = 2^{2s+1} \binom{(n-1)/2}{s}.
\end{align}
For a fixed $s$, both sides can be viewed as a degree $s$ polynomial in $n$. We know that $s+1$ points uniquely determine a degree $s$ polynomial. Since we have already shown that the two polynomials coincide on all odd $n$, we know these two polynomials must be the same one. Therefore the identity is also true for all even $n$. For $k=2s$, using a similar argument we can argue \cref{eq:Fk-norm} is also true for odd $n$.
\end{proof}

\subsection{\texorpdfstring{$t_{\max}$}{tmax} and lower bound on \texorpdfstring{$n$}{n}}

From the discussion in the main text, we know that the set of all $k$-qubit $\U(1)$ invariant unitaries $\V_{n, k}^{\U(1)}$ is a $t$-design for all $\U(1)$ invariant unitaries $\V_{n, n}^{\U(1)}$ if and only if $t\leq t_{\max}$, where
\begin{align}\tag{re \ref{eq:tmax-U1}}
    t_{\max}+1 = \frac{1}{2}\| F_{k+1} \|_1 = 
    \begin{cases}
      2^k \binom{(n-1)/2}{k/2} & \text{when $k$ is even} \\
      2^k \binom{n/2}{(k+1)/2} & \text{when $k$ is odd} .
  \end{cases}
\end{align}

As mentioned in the main text, the above values are guaranteed to give $t_{\max}$ only when $n$ is above a certain threshold, which is expected to be much larger than $k$. This threshold can be estimated as follows. We first consider the case when $k$ is even. The above solution for $t_{\max}$ only involves $\lambda_i: i = 0, \cdots, k+1$. If we include the next irrep $\lambda_{k+2}$ and solve $q_i' : i = 0, \cdots, k+2$ with $|q_{k+2}'| \geq 1$, then we have the upper bound on $t_{\max}+1$ 
\begin{align}\label{rf}
  \frac{1}{2}\sum_{i=0}^{k+2} |q_i'| m_i \geq m_{k+2} = \binom{n}{\frac{k}{2}+1}.
\end{align}
Recall that for a fixed even $k$, $\|F_{k + 1}\|_1$ scales as $n^{\frac{k}{2}}$, whereas for sufficiently large $n$, the right-hand of \cref{rf} scales as $n^{\frac{k}{2}+1}$. Therefore, for sufficiently large $n$, $ \frac{1}{2} \|F_{k + 1}\|_1$ should be the optimal solution, i.e., 
\begin{align}\label{eq:bound-condition-U1}
  \binom{n}{\frac{k}{2}+1} \ge \frac{1}{2} \|F_{k + 1}\|_1\ .
\end{align}
To determine how large $n$ should be, %recall that $ \frac{(n-m+1)^m}{m!} \leq \binom{n}{m} \leq \frac{n^m}{m!} $ and $ \frac{(n-m+1)^m}{m!} \sim \binom{n}{m} \sim \frac{n^m}{m!} $, i.e., $\lim\limits_{n/m\rightarrow\infty} \binom{n}{m}/(\frac{n^m}{m!})=1$. 
we first give a lower bound on the ratio
\begin{align}
    \frac{\binom{n}{\frac{k}{2}+1}}{\frac{1}{2} \|F_{k + 1}\|} = \frac{\binom{n}{\frac{k}{2}+1}}{2^k \binom{(n-1)/2}{k/2}} = \frac{n}{2^\frac{k}{2} (\frac{k}{2}+1)} \frac{n-1}{n-1} \frac{n-2}{n-3} \cdots \frac{n-\frac{k}{2}}{n-k+1} \geq \frac{n}{2^\frac{k}{2} (\frac{k}{2}+1)}.
\end{align}
Therefore, if 
\begin{align}
    n \geq 2^\frac{k}{2} (\frac{k}{2}+1) = 2^{\frac{k-2}{2}} (k+2),
\end{align}
then the condition \cref{eq:bound-condition-U1} is satisfied, and the upper bound in \cref{eq:tmax-upper-U1} is tight.

Similarly, for odd $k$, the condition corresponding to \cref{eq:bound-condition-U1} is 
\begin{align}\label{eq:nbound-cond-kodd}
  \binom{n}{\frac{k+3}{2}} \ge \frac{1}{2} \|F_{k + 1}\|_1\ .
\end{align}
When $k\geq 5$, a lower bound on the ratio is
\begin{align}
    \frac{\binom{n}{(k+3)/2}}{\frac{1}{2} \|F_{k + 1}\|_1} = \frac{\binom{n}{(k+3)/2}}{2^k \binom{n/2}{(k+1)/2}} &= \frac{n(n-1)}{2^\frac{k-3}{2} (k+3)} \frac{n-2}{n} \frac{n-3}{n-2} \cdots \frac{n-\frac{k+1}{2}}{n-k+3} \frac{1}{n-k+1} \nonumber\\
    & \geq \frac{1}{2^\frac{k-3}{2} (k+3)} \frac{(n-1)(n-3)}{n-k+1} \geq \frac{n}{2^\frac{k-3}{2} (k+3)}.
\end{align}
Therefore, if 
\begin{align}
    n \geq 2^{\frac{k-3}{2}} (k+3) 
\end{align}
then the upper bound in \cref{eq:tmax-upper-U1} is tight.
Together, we can write
\begin{align}\tag{re \ref{eq:nbound-U1}}
  n \geq 2^{\lfloor\frac{k}{2}\rfloor} \lfloor\frac{k+3}{2}\rfloor ,
\end{align}
for any even $k$ and odd $k\geq 5$. The only exceptions to this formula is when $k=1$ and $k=3$, where we have $n \geq 3$ and $n \geq 7$ respectively (rather than $n \geq 2$ and $n \geq 6$ predicted by this formula). In the following, we will show that when $(k,n) = (3,6)$, while \cref{eq:nbound-cond-kodd} is not satisfied, our $t_{\max}$ given in \cref{eq:tmax-U1} is still correct. Therefore, the bound on $n$ given \cref{eq:nbound-U1} is also correct for $k=3$.

When $(k,n) = (3,6)$, we know that the solution to $Q$ in \cref{thm:tdesign} is three dimensional, and is spanned by $F_4$, $F_5$ and $F_6$. Equivalently, we know $\mathbf{q}^\T$ is in the span of the rows of
\begin{align}
    \mathbf{Q} = 
    \begin{pmatrix}
        6 & -3 & 1 & 0 & 0 & -1 & 3 \\
        8 & -3 & 0 & 1 & 0 & -3 & 8 \\
        3 & -1 & 0 & 0 & 1 & -3 & 6 \\
    \end{pmatrix},
\end{align}
or in other words
\begin{align}
    \mathbf{q}^\T = 
    \begin{pmatrix}
        a_1 & a_2 & a_3 \\
    \end{pmatrix}
    \mathbf{Q}\  ,
\end{align}
for arbitrary $a_1, a_2, a_3 \in\mathbb{Z}$.  Then  $t_{\max}+1 \leq \frac{1}{2}|\mathbf{q}|\cdot \mathbf{m}=\frac{1}{2} \sum_{w=0}^6 \Tr(\Pi_w) |q(w)| $. We define the $l_1$-norm of a vector $\mathbf{r}$, as 
\be\label{eq:l1}
\|\mathbf{r}\|_1=\sum_{w} |r(w)|=|\mathbf{r}|\cdot \mathbf{1}\ ,
\ee
where $\mathbf{1}$ is the vector with all 1 components. Then,
 $|\mathbf{q}|\cdot \mathbf{m}$ is equal to  $|D^T \mathbf{q}|\cdot \mathbf{1}=\|D^T \mathbf{q}\|_1$, where $\mathbf{1}$ is the vector with all 1 components, $D = \mathrm{diag}\{\Tr\Pi_0, \Tr\Pi_1, \cdots \Tr\Pi_6\} = 
\mathrm{diag}\{\binom{6}{0}, \binom{6}{1}, \cdots \binom{6}{6}\}$, and
\begin{align}
    (D\mathbf{q})^\T  = 
    \begin{pmatrix}
        a_1 & a_2 & a_3 \\
    \end{pmatrix}
    \begin{pmatrix}
        6 & -18 & 15 & 0 & 0 & -6 & 3 \\
        8 & -18 & 0 & 20 & 0 & -18 & 8 \\
        3 & -6 & 0 & 0 & 15 & -18 & 6 \\
    \end{pmatrix}.
\end{align}
Then, it is not hard to check that $|\mathbf{q}|\cdot \mathbf{m}$ is minimized when $(a_1, a_2, a_3) = (1, 0, 0)$ or $(0, 0, 1)$, corresponding to the operator $F_4$ and $\widetilde{F}_4$ respectively. Therefore we know that the $t_{\max}$ from \cref{eq:tmax-U1} also holds for $(k,n) = (3,6)$.

\newpage
\section{\texorpdfstring{$\SU(2)$}{SU(2)} symmetry}\label{app:SU2}

In this section, we study the $\SU(2)$ example in more detail, and prove the identities used to determine $t_{\max}$ in this case. In particular, we determine the value of $\tr A_{2s} C_{2m}$ and show that it is zero when $m < s$, and we also determine norm $\|A_k\|_1$.
The following lemma summarizes these identities.

\lemAk*

Recall that the isotypic decomposition $(\complex^2)^{\otimes n} = \bigoplus_{j = j_{\min}}^{n / 2} \hilbert[Q]_j \otimes \hilbert[M]_j$ implies that the Hermitian projector $\Pi_j$ to the charge sector with angular momentum $j$ satisfies $\tr \Pi_\lambda = d(j) m(n, j)$ where
\begin{equation}
    \dim \hilbert[Q]_j = d(j) = 2 j + 1,
\end{equation}
and $m(n, j) = \dim \hilbert[M]_j$ is its multiplicity. The multiplicities satisfy a recurrence relation \cite{Marvian2024Rotationally} which can be solved as
\begin{equation}
  m(n, j) = \binom{n}{n / 2 - j} \frac{2 j + 1}{n / 2 + j + 1}.
\end{equation}

\subsection{Properties of operator \texorpdfstring{$C_{l}$}{C2m}}
In this subsection and the following, we will use the index $i$ to label the irreps of $\SU(2)$, which is related to the angular momentum as $i = \frac{n}{2}-j$. $i$ takes value from $0, 1, \cdots, \lfloor \frac{n}{2} \rfloor$, and can also be understood as the number of boxes in the second row of the Young diagram.

Recall that in \cite{Marvian2024Rotationally}, the operator $C_l: l=0,2,\cdots $ as defined in \cref{eq:Cl-SU2} is shown to have the eigen-decomposition 
\begin{align}\label{eq:Cl-expression}
    C_l=C_{2m} = (2m-1)!! \binom{n}{2m} \sum_{i=0}^{\lfloor \frac{n}{2} \rfloor} \sum_{r=0}^m (-4)^r \binom{m}{r} \binom{n-2r}{i-r} \frac{n-2i+1}{n-i-r+1} \frac{\Pi_i}{m_i},
\end{align}
where we write 
\be
 m=\frac{l}{2}\ : l = 0, 2, \cdots, 2\floor{\frac{n}{2}}\ ,
\ee
and
\begin{align}
    m_i := m(n,\frac{n}{2}-i) = \binom{n}{i}\frac{n-2i+1}{n-i+1} = \binom{n}{i} - \binom{n}{i-1}.
\end{align}
The factor $(2m-1)!! \binom{n}{2m}$ in $C_{2m}$ is just an overall factor that makes the eigenvalues of $C_{2m}$ integer which does not make any difference for our purpose. Therefore, in order to simplify our calculation, in the following we will define another operator $\widetilde{C}_{2m}$ without this factor, i.e.,
\begin{align}\label{eq:Ctilde}
    \widetilde{C}_{2m} = \Big((2m-1)!! \binom{n}{2m}\Big)^{-1} C_{2m} =\sum_{i=0}^{\lfloor \frac{n}{2} \rfloor} \sum_{r=0}^m (-4)^r \binom{m}{r} \binom{n-2r}{i-r} \frac{n-2i+1}{n-i-r+1} \frac{\Pi_i}{m_i}.
\end{align}
In the following, we will prove several interesting properties of $\widetilde{C}_{2m}$, including $\widetilde{C}_{2m} = \widetilde{C}_{n-1-2m} \widetilde{C}_{n-1}$ when $n$ is odd, and we also calculate $\Tr(\widetilde{C}_{2m}^2)$.

We first introduce two different ways of rewriting the operator $\widetilde{C}_{2m}$, both of which will be useful in later calculations. The first one is straightforward: using
\begin{align}
    \binom{n-2r}{i-r} \frac{n-2i+1}{n-i-r+1} = \binom{n-2r}{i-r} - \binom{n-2r}{i-r-1},
\end{align}
we can write
\begin{align}
    \widetilde{C}_{2m} = \sum_{i=0}^{\lfloor \frac{n}{2} \rfloor} \sum_{r=0}^m (-4)^r \binom{m}{r} \Big( \binom{n-2r}{i-r} - \binom{n-2r}{i-r-1} \Big) \frac{\Pi_i}{m_i}.
\end{align}
While the label $i$ for the irreps of $\SU(2)$ is only meaningful for $i = 0, 1, \cdots, \lfloor \frac{n}{2} \rfloor$, the coefficients of the right-hand side of
\begin{align}\label{eq:C2mcoeff}
    \Tr(\widetilde{C}_{2m} m_i \frac{\Pi_i}{\Tr\Pi_i}) = \sum_{r=0}^m (-4)^r \binom{m}{r} \Big( \binom{n-2r}{i-r} - \binom{n-2r}{i-r-1} \Big)
\end{align}
are non-vanishing only for $i = 0, 1, \cdots, n+1$ and we can formally extend this label. From this expression, it is clear that these coefficients satisfy the relation
\begin{align}
    \Tr(\widetilde{C}_{2m} m_i \frac{\Pi_i}{\Tr\Pi_i}) = -\Tr(\widetilde{C}_{2m} m_{n+1-i} \frac{\Pi_{n+1-i}}{\Tr\Pi_{n+1-i}}).
\end{align}

Next we will argue that $\widetilde{C}_{2m}$ can also be written as
\begin{align}
    \widetilde{C}_{2m} = \sum_{i=0}^{\lfloor \frac{n}{2} \rfloor} \sum_r (-1)^r \binom{n-2m}{i-r} \binom{2m+1}{r} \frac{\Pi_i}{m_i}.
\end{align}
\begin{proof}
We consider the following generating function for the index $i$, where we are formally extending the index $i$ to $n + 1$ (in particular, the right-hand side is meaningful for all $i$ and encodes the values of \cref{eq:C2mcoeff} as the coefficients of $x^i$ for $i = 0, 1, \cdots, \floor{\frac{n}{2}}$):
\begin{align}
  \sum_{i=0}^{n+1} \Tr(\widetilde{C}_{2m} m_i \frac{\Pi_i}{\Tr\Pi_i}) x^i = \sum_{i=0}^{n+1} \sum_{r=0}^m (-4)^r \binom{m}{r} \Big( \binom{n-2r}{i-r} - \binom{n-2r}{i-r-1} \Big) x^i.
\end{align}
The generating function for the first term is
\begin{align}
  \sum_{i,r} (-4)^r \binom{m}{r} \binom{n-2r}{i-r} x^i
  &= \sum_{i,r} (-4x)^r \binom{m}{r} \binom{n-2r}{i-r} x^{i-r} \nonumber\\
  &= (1+x)^n \sum_r \binom{m}{r} \Big(\frac{-4x}{(1+x)^2}\Big)^r \nonumber\\
  &= (1+x)^n \Big(1-\frac{4x}{(1+x)^2}\Big)^m \nonumber\\
  &= (1+x)^{n-2m} (1-x)^{2m},
\end{align}
where in the second line we sum over $i$ using
\begin{equation}
    \sum_{i=0}^{n+1} \binom{n-2r}{i-r} x^i = x^r \sum_{i = r}^{n + 1} \binom{n-2r}{i-r} x^{i - r} =  (1+x)^{n-r}.
\end{equation}
Similarly, we can calculate
\begin{align}
  \sum_{i,r} (-4)^r \binom{m}{r} \binom{n-2r}{i-r-1} x^i
  &= x\sum_{i,r} (-4x)^r \binom{m}{r} \binom{n-2r}{i-r-1} x^{i-r-1} \nonumber\\
  &= x(1+x)^{n-2m} (1-x)^{2m}.
\end{align}
Putting these two results together, we formally have
\begin{align}\label{eq:TrC2mmixi}
  \sum_i\Tr(\widetilde{C}_{2m} m_i \frac{\Pi_i}{\Tr\Pi_i}) x^i &= (1+x)^{n-2m} (1-x)^{2m+1} = \sum_{i,r} (-1)^r \binom{n-2m}{i-r} \binom{2m+1}{r} x^{i},
\end{align}
from which we can read
\begin{align}
    \widetilde{C}_{2m} = \sum_{i=0}^{\lfloor \frac{n}{2} \rfloor} \Tr(\widetilde{C}_{2m} m_i \frac{\Pi_i}{\Tr\Pi_i}) \frac{\Pi_i}{m_i} = \sum_{i=0}^{\lfloor \frac{n}{2} \rfloor} \sum_r (-1)^r \binom{n-2m}{i-r} \binom{2m+1}{r} \frac{\Pi_i}{m_i}.
\end{align}
\end{proof}

In this form, it is easier to see that $\widetilde{C}_{n-1}$ has a particularly simple form when $n$ is odd:
\begin{align}\label{eq:Cn-1}
    \widetilde{C}_{n-1} = \sum_{i=0}^{\lfloor \frac{n}{2} \rfloor} \sum_r (-1)^r \binom{1}{i-r} \binom{n}{r} \frac{\Pi_i}{m_i} = \sum_{i=0}^{\lfloor \frac{n}{2} \rfloor} (-1)^i \Big(\binom{n}{i} - \binom{n}{i-1}\Big) \frac{\Pi_i}{m_i} = \sum_{i=0}^{\lfloor \frac{n}{2} \rfloor} (-1)^i \Pi_i.
\end{align}
This operator is similar to the $C_n = Z^{\otimes n}$ operator in the $\U(1)$ case, in the sense that their eigenvalues have alternating signs $(-1)^i$ and $(-1)^w$, respectively. Furthermore, similar to \cref{eq:Znb}, when $n$ is odd we have the following relation
\begin{align}
  \widetilde{C}_{2m}\widetilde{C}_{n-1} &= \sum_{i=0}^{\lfloor \frac{n}{2} \rfloor} \sum_r (-1)^{i-r} \binom{n-2m}{i-r} \binom{2m+1}{r} \frac{\Pi_i}{m_i} \nonumber\\
  &= \sum_{i=0}^{\lfloor \frac{n}{2} \rfloor} \sum_r (-1)^{r} \binom{n-2m}{r} \binom{2m+1}{i-r} \frac{\Pi_i}{m_i} = \widetilde{C}_{n-1-2m}.
\end{align}
This implies that 
\begin{align}\label{eq:Ctilde-symmetry}
  \Tr(\widetilde{C}_{2m} \Pi_i) = (-1)^i \Tr(\widetilde{C}_{2m} \widetilde{C}_{n-1} \Pi_i)\widetilde{C}_{n-1-2m} = (-1)^i \Tr(\widetilde{C}_{n-1-2m} \Pi_i).
\end{align}

Besides, \cref{eq:Cn-1} also implies that 
\begin{align}
    \| C_{n-1} \|_1 = n!! \| \widetilde{C}_{n-1} \|_1 = n!! 2^n \sim \sqrt{2n} \big( \frac{n}{e} \big)^{\frac{n}{2}} 2^n\ ,
\end{align}

When $n$ is even, we can write
\begin{align}
    \frac{1}{n+1}\widetilde{C}_n &= \frac{1}{n+1}\sum_{i=0}^{\lfloor \frac{n}{2} \rfloor} \sum_r (-1)^r \binom{0}{i-r} \binom{n+1}{r} d_i \frac{\Pi_i}{m_i d_i} \nonumber\\
    &= \frac{1}{n+1} \sum_{i=0}^{\lfloor \frac{n}{2} \rfloor} (-1)^i \binom{n+1}{i} (n+1-2i) \frac{\Pi_i}{m_i d_i}  \nonumber\\
    &= \sum_{i=0}^{\lfloor \frac{n}{2} \rfloor} (-1)^i \Big(\binom{n}{i} - \binom{n}{i-1}\Big) \frac{\Pi_i}{m_i d_i} = \sum_{i=0}^{\lfloor \frac{n}{2} \rfloor} (-1)^i \frac{\Pi_i}{d_i},
\end{align}
where $d_i = d(\frac{n}{2}-i) = n-2i+1$ is dimension of the irreps of $\SU(2)$ with label $i$. This implies that 
\begin{align}
    \| C_n \|_1 = (n-1)!! \binom{n}{n/2} \sim \sqrt{2(n-1)} \big( \frac{n-1}{e} \big)^{\frac{n-1}{2}} \frac{2^n}{\sqrt{n\pi/2}} \sim \frac{2}{\sqrt{\pi}} \big( \frac{n}{e} \big)^{\frac{n-1}{2}} 2^n.
\end{align}

Next, we will prove that
\begin{lemma}\label{lem:C2msquare}
The $\widetilde{C}_{2m}$ operators defined in \cref{eq:Ctilde} satisfy
\begin{align}
    \Tr(\widetilde{C}_{2m}^2) = (2m+1) 2^n \binom{n}{2m}^{-1},
\end{align}
or in terms of the $C_{2m}$ operators,
\begin{align}
    \Tr(C^2_{2m}) = (2m+1)!! (2m-1)!! 2^n \binom{n}{2m}.
\end{align}
\end{lemma}
\begin{proof}
Similar to the $\U(1)$ case, we would first calculate the generating function $\sum_m \widetilde{C}_{2m} y^{2m}$. However, it turns out that in this case, it is easier to deal with the generating function $\sum_m \binom{n}{2m} \widetilde{C}_{2m} y^{2m}$.

Recall that from \cref{eq:TrC2mmixi} we know
\begin{align}
  \sum_i\Tr(\widetilde{C}_{2m} m_i \frac{\Pi_i}{\Tr\Pi_i}) x^i &= (1+x)^{n-2m} (1-x)^{2m+1} .
\end{align}
Therefore
\begin{align}
  \sum_{i,m} \binom{n}{2m} \Tr(\widetilde{C}_{2m} m_i \frac{\Pi_i}{\Tr\Pi_i}) x^i y^{2m} &= (1-x)(1+x)^n \sum_m \binom{n}{2m} \Big(\frac{y(1-x)}{1+x} \Big)^{2m} \nonumber\\
  &= \frac{1}{2} (1-x)(1+x)^n \Big(\Big(1+\frac{y(1-x)}{1+x} \Big)^n + \Big(1-\frac{y(1-x)}{1+x} \Big)^n \Big) \nonumber\\
  &= \frac{1}{2} (1-x) \big((1+x+y(1-x))^n + (1+x-y(1-x))^n \big) \nonumber\\
  &= \frac{1}{2} (1-x) \big((1+y+x(1-y))^n + (1-y+x(1+y))^n \big) \nonumber\\
  &= \frac{1}{2} (1-x) \sum_i \binom{n}{i} \Big( \Big(\frac{1-y}{1+y}\Big)^i (1+y)^n + \Big(\frac{1+y}{1-y}\Big)^i (1-y)^n \Big) x^i \nonumber\\
  &= \frac{1}{2} (1-x) \sum_i \binom{n}{i} \big( f_i(y) + f_i(-y) \big) x^i 
\end{align}
where in the second line we have used
\begin{equation}
    \sum_m \binom{n}{2m} a^{2m} = \frac{1}{2} \big( (1+a)^n + (1-a)^n \big),
\end{equation}
and in the last line we introduced $f_i(y) := (1-y)^i (1+y)^{n-i}$. Comparing the coefficients of $x^i$, we have
\begin{align}
  \sum_m \binom{n}{2m} \Tr(\widetilde{C}_{2m} m_i \frac{\Pi_i}{\Tr\Pi_i}) y^{2m} &= \binom{n}{i} \frac{1}{2}\big( f_i(y) + f_i(-y) \big)  - \binom{n}{i-1} \frac{1}{2}\big( f_{i-1}(y) + f_{i-1}(-y) \big) .
\end{align}
Dividing $m_i = \binom{n}{i} - \binom{n}{i-1}$ on both sides, we have
\begin{align}
  \sum_m \binom{n}{2m} \Tr(\widetilde{C}_{2m} \frac{\Pi_i}{\Tr\Pi_i}) y^{2m} &= \frac{n-i+1}{n-2i+1} \frac{1}{2}\big( f_i(y) + f_i(-y) \big) - \frac{i}{n-2i+1} \frac{1}{2}\big( f_{i-1}(y) + f_{i-1}(-y) \big).
\end{align}
Therefore we have the generating function 
\begin{align}
    \sum_m \binom{n}{2m} \widetilde{C}_{2m} y^{2m} = \sum_{i=0}^{\lfloor\frac{n}{2}\rfloor} \Big( \frac{n-i+1}{n-2i+1} \frac{1}{2}\big( f_i(y) + f_i(-y) \big) - \frac{i}{n-2i+1} \frac{1}{2}\big( f_{i-1}(y) + f_{i-1}(-y) \big) \Big) \Pi_i .
\end{align}
Taking trace of the square of both sides, we have
\begin{align}
    \Tr \Big(\sum_m \binom{n}{2m} \widetilde{C}_{2m} y^{2m}\Big)^2 = \sum_{i=0}^{\lfloor\frac{n}{2}\rfloor} \Big( \frac{n-i+1}{n-2i+1} \frac{1}{2}\big( f_i(y) + f_i(-y) \big) - \frac{i}{n-2i+1} \frac{1}{2}\big( f_{i-1}(y) + f_{i-1}(-y) \big) \Big)^2 m_i d_i.
\end{align}
Notice that the summand is invariant under the change of variable $i \mapsto n+1-i$. Therefore we can write
\begin{align}
  &\quad \sum_{i=0}^{\lfloor\frac{n}{2}\rfloor} \Big( \frac{n-i+1}{n-2i+1} \frac{1}{2}\big( f_i(y) + f_i(-y) \big) - \frac{i}{n-2i+1} \frac{1}{2}\big( f_{i-1}(y) + f_{i-1}(-y) \big) \Big)^2 m_i d_i \nonumber\\
  &= \frac{1}{8} \sum_{i=0}^n \binom{n}{i} \frac{1}{n-i+1} \Big((n-i+1) \big( f_i(y) + f_i(-y) \big) - i \big( f_{i-1}(y) + f_{i-1}(-y) \big)\Big)^2 \nonumber\\
  &= \frac{1}{8} \sum_i \binom{n}{i} \Big( (n-i+1) \big( f_i(y) + f_i(-y) \big)^2 - 2i \big( f_i(y) + f_i(-y) \big) \big( f_{i-1}(y) + f_{i-1}(-y) \big) \Big) + \binom{n}{i-1} i \big( f_{i-1}(y) + f_{i-1}(-y) \big)^2 
\end{align}
The first term is
\begin{align}
  &\quad \frac{1}{8} \sum_i \binom{n}{i} \Big( (n-i+1) \big( f_i(y) + f_i(-y) \big)^2 \nonumber\\
  &= \frac{1}{8} \sum_i \binom{n}{i} (n-i+1) \big( (1-y)^{2i} (1+y)^{2(n-i)} + 2(1-y^2)^n + (1+y)^{2i} (1-y)^{2(n-i)} \big) \nonumber\\
  &= \frac{1}{8} (n+1) \sum_i \binom{n}{i} \big( (1-y)^{2i} (1+y)^{2(n-i)} + 2(1-y^2)^n + (1+y)^{2i} (1-y)^{2(n-i)} \big) \nonumber\\
  &\quad - \frac{1}{8} n \sum_i \binom{n-1}{i-1} \big( (1-y)^{2i} (1+y)^{2(n-i)} + 2(1-y^2)^n + (1+y)^{2i} (1-y)^{2(n-i)} \big) \nonumber\\
  &= \frac{1}{8} (n+1) \big( 2^{n+1}(1+y^2)^n + 2^{n+1}(1-y^2)^n \big) - n \big( 2^n(1+y^2)^n + 2^n(1-y^2)^n \big) \nonumber\\
  &= (n+2)2^{n-3} \big( (1+y^2)^n + (1-y^2)^n \big)
\end{align}
The second (crossing) term is
\begin{align}
  &\quad -\frac{1}{4} \sum_i \binom{n}{i} i \big( f_i(y) + f_i(-y) \big) \big( f_{i-1}(y) + f_{i-1}(-y) \big) \nonumber\\
  &= -\frac{1}{4} n \sum_i \binom{n-1}{i-1} \big( (1-y)^{2i-1} (1+y)^{2n-2i+1} + 2(1-y^2)^{n-1}(1+y^2) + (1+y)^{2i-1} (1-y)^{2n-2i+1} \big) \nonumber\\
  &= -\frac{1}{4} n \big( 2^n(1-y^2) (1+y^2)^{n-1} + 2^n(1+y^2)(1-y^2)^{n-1} \big) \nonumber\\
  &= -2^{n-2}n \big( (1-y^2) (1+y^2)^{n-1} + (1+y^2)(1-y^2)^{n-1} \big) 
\end{align}
Note that each summand in the third term is related to the first term with $i \mapsto n+i-1$, we know the third term is equal to the first one. Therefore together we have
\begin{align}
  \Tr \Big(\sum_m \binom{n}{2m} \widetilde{C}_{2m} y^{2m}\Big)^2 &= (n+2) 2^{n-2} \big( (1+y^2)^n + (1-y^2)^n \big) - n 2^{n-2} \big( (1-y^2) (1+y^2)^{n-1} + (1+y^2)(1-y^2)^{n-1} \big) \nonumber\\
  &= 2^{n-2} \sum_m \binom{n}{2m} 2(n+2) y^{4m} - \sum_m n \Big( \binom{n-1}{m} (y^{2m}-y^{2m+2}) + \binom{n-1}{m} (-1)^m (y^{2m}+y^{2m+2}) \Big) \nonumber\\
  &= 2^{n-2} \sum_m \binom{n}{2m} 2(n+2) y^{4m} - 2n \Big( \binom{n-1}{2m} - \binom{n-1}{2m-1}\Big) y^{4m} \nonumber\\
  &= 2^{n-1} \sum_m \binom{n}{2m} (n+2) y^{4m} - \binom{n}{2m}(n-4m) y^{4m} \nonumber\\
  &= 2^n \sum_m \binom{n}{2m} (2m+1) y^{4m}.
\end{align}
Using orthogonality of the $\widetilde{C}_{2m}$ operators, we also know
\begin{align}
  \Tr \Big(\sum_m \binom{n}{2m} \widetilde{C}_{2m} y^{2m}\Big)^2 &= \sum_{m,m'} \binom{n}{2m} \binom{n}{2m'} \Tr \widetilde{C}_{2m} \widetilde{C}_{2m'} y^{2(m+m')} = \sum_{m} \binom{n}{2m}^2 \Tr \widetilde{C}_{2m}^2 y^{4m}.
\end{align}
Comparing the two sides, we have
\begin{align}\label{eq:TrC2m2}
  \Tr \widetilde{C}_{2m}^2 = (2m+1) 2^n \binom{n}{2m}^{-1} .
\end{align}
\end{proof}

Finally, using the completeness relation of the $\widetilde{C}_{2m}$ operators, we have
\begin{align}
  \Pi_i = \sum_m \frac{\Tr(\Pi_i \widetilde{C}_{2m})}{\Tr(\widetilde{C}_{2m}^2)} \widetilde{C}_{2m} &= 2^{-n} \sum_m \sum_{r=0}^m (-4)^r \binom{m}{r} \Big( \binom{n-2r}{i-r} - \binom{n-2r}{i-r-1} \Big) \frac{n-2i+1}{2m+1} \binom{n}{2m} \widetilde{C}_{2m} \nonumber\\
  %&= 2^{-n} d_i \sum_m \frac{1}{(2m+1)!!} \left[\sum_{r=0}^m (-4)^r \binom{m}{r} \Big( \binom{n-2r}{i-r} - \binom{n-2r}{i-r-1} \Big)\right] C_{2m},\\
  %&= 2^{-n} \Tr\Pi_i \sum_m s(n,m,i) C_{2m},
  &= 2^{-n} \sum_m \frac{d_i}{(2m+1)!!} \sum_{r=0}^m (-4)^r \binom{m}{r} \Big( \binom{n-2r}{i-r} - \binom{n-2r}{i-r-1} \Big) C_{2m},
\end{align}
and for comparison, we also present the expansion of $C_{2m}$ in the $\Pi_i$ basis
\begin{align}
    C_{2m} %& = (2m-1)!! \binom{n}{2m} \sum_{i=0}^{\lfloor \frac{n}{2} \rfloor} \left[\sum_{r=0}^m (-4)^r \binom{m}{r} \Big( \binom{n-2r}{i-r} - \binom{n-2r}{i-r-1} \Big) \right] \frac{\Pi_i}{m_i} \\
    %& = \binom{n}{2m} (2m+1)!! (2m-1)!! \sum_{i=0}^{\lfloor \frac{n}{2} \rfloor} s(n,m,i) \Pi_i
    & = \binom{n}{2m} \sum_{i=0}^{\lfloor \frac{n}{2} \rfloor} \frac{(2m-1)!!}{m_i}\sum_{r=0}^m (-4)^r \binom{m}{r} \Big( \binom{n-2r}{i-r} - \binom{n-2r}{i-r-1} \Big) \Pi_i\ .
\end{align}
%where 
%\begin{align}
%  s(n,m,i)=\frac{1}{(2m+1)!! m_i}\sum_{r=0}^m (-4)^r \binom{m}{r} \Big( \binom{n-2r}{i-r} - \binom{n-2r}{i-r-1} \Big)\ ,
%\end{align}
Therefore, if we define 
\begin{subequations}\label{eq:Pi-C-normalized-SU2}
  \begin{align}
    \overline\Pi_i &= \sqrt{\frac{2^n}{\Tr\Pi_i}} \Pi_i ,\\
    \overline C_{2m} &= \sqrt{\frac{2^n}{\Tr C_{2m}^2}} C_{2m} = \frac{\sqrt{2m+1}}{(2m+1)!!} \binom{n}{2m}^{-\frac{1}{2}} C_{2m} = \frac{1}{\sqrt{2m+1}} \binom{n}{2m}^{\frac{1}{2}} \widetilde{C}_{2m},
  \end{align}
\end{subequations}
with the normalization $\Tr\overline\Pi_i^2 = \Tr\overline C_{2m}^2 = \Tr\id = 2^n$, then they satisfy
\begin{subequations}
  \begin{align}
    \overline C_{2m} &= \sum_i \bar{c}_{m,i} \overline\Pi_i , \tag{re \ref{eq:reciprocity-SU2a}}\nonumber\\
    \overline\Pi_i &= \sum_m \bar{c}_{m,i} \overline C_{2m} , \tag{re \ref{eq:reciprocity-SU2b}}
  \end{align}
\end{subequations}
where
\begin{align}\label{eq:c-mi}
  \bar{c}_{m,i} &:= \sqrt{\frac{\Tr(\Pi_i)}{\Tr(C_{2m}^2)}} c_{2m}(n/2-i). %\nonumber\\
  %&= \sqrt{\frac{\binom{n}{2m} d_i}{2^n (2m+1) m_i}} \sum_{r=0}^m (-4)^r \binom{m}{r} \Big( \binom{n-2r}{i-r} - \binom{n-2r}{i-r-1} \Big) \nonumber\\
  %&= 2^{-\frac{n}{2}} \sqrt{\frac{n-i+1}{2m+1}} \Big(\binom{n}{2m}\Big/\binom{n}{i}\Big)^{\frac{1}{2}} \sum_{r=0}^m (-4)^r \binom{m}{r} \Big( \binom{n-2r}{i-r} - \binom{n-2r}{i-r-1} \Big) \nonumber\\
  %&= 2^{-\frac{n}{2}} \sqrt{\frac{n-i+1}{2m+1}} \Big(\binom{n}{2m}\Big/\binom{n}{i}\Big)^{\frac{1}{2}} \sum_{r=0}^i (-1)^r \binom{n-2m}{i-r} \binom{2m+1}{r}.
\end{align}
%and we recall that $d_i = n-2i+1$ is dimension of the irrep of $\SU(2)$ with label $i$.

One can verify that $\frac{1}{\sqrt{2m+1}} \binom{n}{2m}^{\frac{1}{2}}$ is invariant under the transformation $2m \mapsto n-1-2m$, therefore $\overline{C}_{2m}$ satisfies the same symmetry as $\widetilde{C}_{2m}$ in \cref{eq:Ctilde-symmetry}, namely,
\begin{align}\tag{re \ref{eq:Cbar-symmetry}}
  \Tr(\overline{C}_{2m} \Pi_i) = (-1)^i \Tr(\overline{C}_{n-1-2m} \Pi_i).
\end{align}

\subsection{Properties of operator \texorpdfstring{$A_k$}{Ak}}
For simplicity, in this subsection and the following, we introduce the variable
\begin{align}
    s = \frac{k}{2}: k = 0, 2, \cdots, 2\lfloor\frac{n}{2} \rfloor.
\end{align}
Recall that $A_{2s} = A_k$ is defined as
\begin{align}\tag{re \ref{eq:Ak}}
  A_k=A_{2s} %&= \sum_{j=j_{\min}}^{n/2} q_j^{(s)} m_j \frac{\Pi_j}{\Tr\Pi_j}\ , \nonumber\\
         &= \sum_{j=j_{\min}}^{n/2} (-1)^{\frac{n}{2}-j} \binom{\frac{n}{2}- s+j} {n-2s} m(n,j) \frac{\Pi_j}{\Tr\Pi_j}.
\end{align}
In the first version of this work on arXiv, we proposed a combinatorial conjecture as Eq. (120) that 
\begin{align}\label{eq:conj-Ak}
  \Tr(A_{2s} C_{2m}) = \sum_{j=j_{\min}}^{\frac{n}{2}} (-1)^{\frac{n}{2}-j} \binom{\frac{n}{2}-s+j}{-\frac{n}{2}+s+j} \sum_{r=0}^{\frac{n}{2}-j} (-4)^r \binom{m}{r} \binom{n-2r}{\frac{n}{2}-j-r} \frac{2j+1}{n/2+j-r+1} = 0 , \quad \text{when } m<s.
\end{align}
Here, we prove a stronger result
\begin{align}
  \Tr(A_{2s} C_{2m}) = 4^s (2m-1)!! \binom{n}{2m} \binom{m}{s}
\end{align}
for all $s$ and $m$, which implies \cref{eq:conj-Ak}.

In terms of label $i=\frac{n}{2}-j$, the operator $A_{2s}$ defined in \cref{eq:Ak} is
\begin{align}
  A_{2s} &= \sum_{i=0}^{\lfloor n/2\rfloor} (-1)^i \binom{n - s -i}{s-i} m_i \frac{\Pi_i}{\Tr\Pi_i}.
\end{align}

In the following, we will prove
\begin{align}\label{eq:TrA2sC2m}
  \Tr(A_{2s} \widetilde{C}_{2m}) = \sum_{i=0}^{\lfloor n/2\rfloor} (-1)^i \binom{n - s -i}{s-i} \sum_{r=0}^i (-4)^r \binom{m}{r} \Big(\binom{n-2r}{i-r} - \binom{n-2r}{i-r-1}\Big) = 4^s \binom{m}{s} ,
\end{align}
This identity immediately implies \cref{eq:conj-Ak}: when $m < s$, $\Tr(A_{2s} C_{2m})=0$ because $\binom{m}{s} = 0$.

\begin{proof}
First, we show that $\Tr(A_{2s} \widetilde{C}_{2m})$ can be rewritten compactly as
\begin{align}
  \Tr(A_{2s} \widetilde{C}_{2m}) = \sum_{i=0}^{n} (-1)^i \binom{n - s -i}{s-i} \sum_{r=0}^m (-4)^r \binom{m}{r} \binom{n-2r}{i-r} .
\end{align}
This is because when $i \geq n-s$, using \cref{eq:extended-binomial} we can write 
\begin{align}
  \sum_{i=n-s}^{n} (-1)^i \binom{n - s -i}{s-i} \binom{n-2r}{i-r} &= \sum_{i=n-s}^{n} (-1)^{i+n} \binom{-s+i-1}{n-2s} \binom{n-2r}{n-i-r} \nonumber\\
  &= \sum_{i=1}^{s+1} -(-1)^i \binom{n-s-i}{n-2s} \binom{n-2r}{i-r-1} \nonumber\\
  &= \sum_{i=0}^s -(-1)^i \binom{n-s-i}{s-i} \binom{n-2r}{i-r-1},
\end{align}
where in the second step, we made the substitution $i \mapsto n+1-i$, and for the last step we note that the summand is $0$ when $i = 0$ or $i = s + 1$. Then, using
\begin{align}
  \binom{n - s -i}{n-2s} = \frac{1}{(n-2s)!} \Big(\frac{\d}{\d x}\Big)^{n-2s} x^{n - s -i} \Big|_{x=1},
\end{align}
we have
\begin{align}
  \Tr(A_{2s} \widetilde{C}_{2m}) &= \sum_{i=0}^n (-1)^i \binom{n - s -i}{s-i} \sum_{r=0}^i (-4)^r \binom{m}{r} \binom{n-2r}{i-r} \nonumber\\
  &= \frac{1}{(n-2s)!} \Big(\frac{\d}{\d x}\Big)^{n-2s} x^{n - s} \sum_{r=0}^m (-4)^r \binom{m}{r} \sum_{i=0}^n (-1)^i \binom{n-2r}{i-r} x^{-i} \Big|_{x=1} \nonumber\\
  &= \frac{1}{(n-2s)!} \Big(\frac{\d}{\d x}\Big)^{n-2s} x^{- s} \sum_{r=0}^m (4x)^r \binom{m}{r} (x-1)^{n-2r} \Big|_{x=1} \nonumber\\ 
  &= \frac{1}{(n-2s)!} \Big(\frac{\d}{\d x}\Big)^{n-2s} x^{- s} (x-1)^{n-2m} (x+1)^{2m} \Big|_{x=1}
\end{align}
A necessary condition for this term to be non-vanishing is that $(x-1)^{n-2m}$ is fully killed by the derivative $(\frac{\d}{\d x})^{n-2s}$. Therefore $\Tr(A_{2s} \widetilde{C}_{2m})$ vanishes when $m<s$. When $m\geq s$, we also have to first fully kill $(x-1)^{n-2m}$. Therefore we have
\begin{align}
  \Tr(A_{2s} \widetilde{C}_{2m}) &= \frac{1}{(n-2s)!} \Big(\frac{\d}{\d x}\Big)^{n-2s} x^{- s} (x-1)^{n-2m} (x+1)^{2m} \Big|_{x=1} \nonumber\\
  &= \frac{1}{(2m-2s)!}  \Big(\frac{\d}{\d x}\Big)^{2m-2s} x^{- s} (x+1)^{2m} \Big|_{x=1} \nonumber\\
  &= \sum_{i=0}^{2m-2s} \frac{1}{i!}  \Big(\frac{\d}{\d x}\Big)^{i} x^{- s} \frac{1}{(2m-2s-i)!} \Big(\frac{\d}{\d x}\Big)^{2m-2s-i} (x+1)^{2m} \Big|_{x=0} \nonumber\\
  &= \sum_{i=0}^{2m-2s} \binom{-s}{i} x^{- s-i} \binom{2m}{2s+i} (x+1)^{2s+i} \Big|_{x=1} \nonumber\\
  &= 2^{2s}\sum_{i=0}^{2m-2s} (-2)^i \binom{s+i-1}{i} \binom{2m}{2s+i} .
\end{align}
Now consider the generating function
\begin{align}
  \sum_{m}\Tr(A_{2s} \widetilde{C}_{2m}) y^{2m} &= 2^{2s}\sum_{i=0}^{2m-2s} (-2)^i \binom{s+i-1}{i} \sum_{m}\binom{2m}{2s+i} y^{2m} \nonumber\\
  &= 2^{2s-1}\sum_{i=0}^{2m-2s} (-2)^i \binom{s+i-1}{i} \Big(\frac{y^{2s+i}}{(1-y)^{2s+i+1}} + \frac{(-y)^{2s+i}}{(1+y)^{2s+i+1}}\Big) .
\end{align}
The first term is
\begin{align}
  &\quad 2^{2s-1}\sum_{i=0}^{2m-2s} (-2)^i \binom{s+i-1}{i} \frac{y^{2s+i}}{(1-y)^{2s+i+1}}  \nonumber\\
  &= -(-2)^s\frac{y^{s+1}}{(1-y)^{s+2}} \sum_{i=0}^{2m-2s} \binom{s+i-1}{s-1} \frac{(-2y)^{s+i-1}}{(1-y)^{s+i-1}} \nonumber\\
  &= -(-2)^s\frac{y^{s+1}}{(1-y)^{s+2}} \frac{(-2y)^{s-1}}{(1+y)^{s-1}} \frac{1-y}{1+y} \nonumber\\
  &= 2^{2s-1} \frac{1}{1-y} \frac{y^{2s}}{(1-y^2)^{s}} .
\end{align}
Similarly, the second term is obtained by simply replacing $y$ with $-y$
\begin{align}
  2^{2s-1} \frac{1}{1+y} \frac{y^{2s}}{(1-y^2)^{s}}.
\end{align}
Together we have
\begin{align}
  \sum_{m}\Tr(A_{2s} \widetilde{C}_{2m}) y^{2m} = 2^{2s} \frac{y^{2s}}{(1-y^2)^{s+1}} = \sum_m 2^{2s} \binom{m}{s} y^{2m},
\end{align}
From which we can read 
\begin{align}
  \Tr(A_{2s} \widetilde{C}_{2m}) = 2^{2s} \binom{m}{s}.
\end{align}
\end{proof}

Using \cref{eq:TrC2m2}, we know that $A_{2s}$ can be expanded in the $\widetilde{C}_{2m}$ basis as
\begin{align}
  A_{2s} &= \sum_m \frac{\Tr(A_{2s} \widetilde{C}_{2m})}{\Tr(\widetilde{C}_{2m}^2)} \widetilde{C}_{2m} = \sum_m \frac{2^{2s-n}}{2m+1} \binom{n}{2m} \binom{m}{s} \widetilde{C}_{2m} \nonumber\\ &= \sum_m \frac{2^{2s-n}}{(2m+1)!!} \binom{m}{s} C_{2m} = \sum_m \frac{2^{2s-n}}{\sqrt{2m+1}} \binom{m}{s} \binom{n}{2m}^{\frac{1}{2}} \overline C_{2m}.
\end{align}

From the definition of $A_{2s}$, we see that its eigenvalues have alternating signs $(-1)^i$. On the other hand, when $n$ is odd, we know from \cref{eq:Cn-1} that the eigenvalues of $\widetilde{C}_{2n-1}$ are alternating signs $(-1)^i$. Therefore we know when $n$ is odd,
\begin{align}
  \| A_{2s} \|_1 = \sum_{i=0}^s \binom{n - s -i}{s-i} \Big(\binom{n}{i} - \binom{n}{i-1} \Big) = \Tr(A_{2s} \widetilde{C}_{n-1}) = 2^{2s} \binom{(n-1)/2}{s}.
\end{align}
Again, both sides are polynomials in $n$ of degree $s$. Using a similar argument as the $\U(1)$ case, we know if it is true for all odd $n$, it must also be true for all even $n$.

\subsection{\texorpdfstring{$t_{\max}$}{tmax} and lower bound on \texorpdfstring{$n$}{n}}

From the discussion in the main text, we know that the set of all $k$-qubit $\SU(2)$ invariant unitaries $\V_{n, k}^{\SU(2)}$ is a $t$-design for all $\SU(2)$ invariant unitaries $\V_{n, n}^{\SU(2)}$ if and only if $t\leq t_{\max}$, where
\begin{align}\tag{re \ref{eq:tmax-SU2}}
  t_{\max} + 1 = \frac{1}{2} \|A_{2(\lfloor\frac{k}{2}\rfloor+1)}\|_1 = 2^{2\lfloor\frac{k}{2}\rfloor+1} \binom{(n-1)/2}{\lfloor\frac{k}{2}\rfloor+1},
\end{align}
assuming that $n$ is large enough. Recall that $s = \lfloor\frac{k}{2}\rfloor$.
Similar to the $\U(1)$ section, \cref{eq:tmax-SU2} is a tight bound when it is not greater than $m_{s+2}$, and $m_{s+2} \leq m_{r}$ for any $r>s+2$. The first condition can be written as
\begin{align}\label{eq:bound-condition-SU2}
  \binom{n}{s+1} \frac{n-2s-3}{s+2} \geq 2^{2s+1} \binom{(n-1)/2}{s+1}.
\end{align}
When $s\geq 3$, a lower bound on the ratio is
\begin{align}
  \frac{\binom{n}{s+1} \frac{n-2s-3}{s+2}}{2^{2s+1} \binom{(n-1)/2}{s+1}} &= \frac{(n-2s-3)n(n-1)}{2^s(s+2)} \frac{n-2}{n-1} \frac{n-3}{n-3} \cdots \frac{n-s}{n-2s+3} \frac{1}{n-2s+1} \frac{1}{n-2s-1} \nonumber\\
  &\geq \frac{n}{2^s(s+2)} \frac{(n-2s-3)(n-2)}{(n-2s+1)(n-2s-1)} \sim \frac{n}{2^s(s+2)}.
\end{align}
The factor $\frac{(n-2s-3)(n-2)}{(n-2s+1)(n-2s-1)} \geq 1$ as long as $n > \frac{4s^2-4s-7}{2s-5}$. Assuming this is true, we have the lower bound
\begin{align}
  \frac{\binom{n}{s+1} \frac{n-2s-3}{s+2}}{2^{2s+1} \binom{(n-1)/2}{s+1}} \geq \frac{n}{2^s(s+2)}.
\end{align}
It can be checked that when $s\geq 3$, $2^s(s+2) > \frac{4s^2-4s-7}{2s-5}$. Therefore, for $s\geq 3$, if 
\begin{align}\tag{re \ref{eq:nbound-SU2}}
    n \geq 2^s(s+2),
\end{align}
then the condition \cref{eq:bound-condition-SU2} is satisfied. It is not hard to verify that when $s \geq 3$, $n \geq 2^s(s+2)$ also implies $s+2 < \frac{n}{2} - \sqrt{n\ln n}$, which in turn implies that $m_{s+2} \leq m_{r}$ for any $r>s+2$. Therefore when $s \geq 3$, $n \geq 2^s(s+2)$ implies that $\frac{1}{2} \|A_{2(s+1)}\|_1 \leq m_r$ for any $r\geq s+2$, and the upper bound in \cref{eq:tmax-SU2} is tight.

When $s=1$ and $s=2$, the exact solution to \cref{eq:bound-condition-SU2} is $n \geq 9$ and $n \geq 18$ respectively (rather than $n \geq 6$ and $n \geq 16$ predicted by \cref{eq:nbound-SU2}). From \cref{tab:imax}, we know that $m_{s+2} \leq m_{r}$ for any $r>s+2$ is true, if $n\geq 13$ when $s=1$, or $n\geq 15$ when $s=2$. Therefore, based on our current argument, for $s=1$, $t_{\max}$ in \cref{eq:tmax-SU2} is exact when $n\geq 13$. On the other hand, when $s=2$, similar to the $\U(1)$ case, we can show that while $n=16, 17$ does not satisfies \cref{eq:bound-condition-SU2}, $t_{\max}$ given in \cref{eq:tmax-SU2} is still correct, and therefore the bound on \cref{eq:nbound-SU2} holds also for $s=2$. Let's consider $n=16$ and $n=17$ separately in the following. 

When $n=16$, the irreps whose dimensions are smaller than $\frac{1}{2} \|A_{2(s+1)}\|_1 = 1430$ are $i = 0, 1, 2, 3, 4$, and we only need to consider solutions with support on these irreps. Therefore the solution to $Q$ in \cref{thm:tdesign} is two dimensional, and is spanned by $A_6$ and $A_8$. Equivalently, we know $\mathbf{q}^\T$ is in the span of the rows of
\begin{align}
    \mathbf{Q} = 
    \begin{pmatrix}
        -286 & 66 & -11 & 1 & 0 \\
        -2079 & 429 & -54 & 0 & 1 \\
    \end{pmatrix},
\end{align}
or in other words
\begin{align}
    \mathbf{q}^\T = 
    \begin{pmatrix}
        a_1 & a_2 \\
    \end{pmatrix}
    \mathbf{Q}.
\end{align}
Recall the definition of the $l_1$-norm of a vector in \cref{eq:l1}. Let $D = \mathrm{diag}\{m_0, m_1, \cdots m_4\} = \mathrm{diag}\{1, 15, 104, 440, 1260\}$, then $t_{\max}+1 \leq \frac{1}{2} \|D \mathbf{q}\|_1$, where 
\begin{align}
    (D\mathbf{q})^\T  = 
    \begin{pmatrix}
        a_1 & a_2 \\
    \end{pmatrix}
    \begin{pmatrix}
        -286 & 990 & -1144 & 440 & 0 \\
        -2079 & 6435 & -5616 & 0 & 1260 \\
    \end{pmatrix}.
\end{align}
It is not hard to check that $\|D \mathbf{q}\|_1$ is minimized when $(a_1, a_2) = (1, 0)$, corresponding to the operator $A_6$. Therefore we know that the $t_{\max}$ from \cref{eq:tmax-SU2} also holds for $(s,n) = (2,16)$.

For $n=17$, the irreps whose dimensions are smaller than $\frac{1}{2} \|A_{2(s+1)}\|_1 = 1792$ are also $i = 0, 1, 2, 3, 4$, and we only need to consider solutions with support on these irreps. Therefore the calculation is similar, where we have
\begin{align}
    \mathbf{Q} = 
    \begin{pmatrix}
        -364 & 78 & -12 & 1 & 0 \\
        -2925 & 560 & -65 & 0 & 1 \\
    \end{pmatrix},
\end{align}
$D = \mathrm{diag}\{m_0, m_1, \cdots m_4\} = \mathrm{diag}\{1, 16, 119, 544, 1700\}$, and
\begin{align}
    (D\mathbf{q})^\T  = 
    \begin{pmatrix}
        a_1 & a_2 \\
    \end{pmatrix}
    \begin{pmatrix}
        -364 & 1248 & -1428 & 544 & 0 \\
        -2925 & 8960 & -7735 & 0 & 1700 \\
    \end{pmatrix}.
\end{align}
Again, it is not hard to verify that $\|D \mathbf{q}\|_1$ is minimized when $(a_1, a_2) = (1, 0)$, corresponding to the operator $A_6$. Therefore we know that the $t_{\max}$ from \cref{eq:tmax-SU2} also holds for $(s,n) = (2,17)$.

In conclusion, we know that for $s\geq 2$, if 
\begin{align}\tag{re \ref{eq:nbound-SU2}}
    n \geq 2^s(s+2),
\end{align}
then the upper bound in \cref{eq:tmax-SU2} is tight.

\subsection{Dimensions of multiplicities}\label{app:SU2-dim}

In the calculation in the previous subsection, it is crucial to know the maximal $i$ such that $m_i \leq m_r$ for all $r>i$, denoted as $i_{\max}(n)$,
\begin{align}\label{dfref}
    i_{\max}(n) := \max\{i: m_i \leq m_r, \forall r>i\}.
\end{align}
In this section, we give an approximation on $i_{\max}(n)$, or equivalently, $j_0 := \frac{n}{2} - i_{\max}(n)$.

It is easy to verify that $m(n, \frac{n}{2}) = 1$. Further, for most values of $j$, $m(n, j) < m(n, j - 1)$. Using
\begin{equation}
  m(n, j - 1) = \binom{n}{\frac{n}{2} - j + 1} \frac{2 j - 1}{\frac{n}{2} + j} = \binom{n}{\frac{n}{2} - j} \frac{2 j - 1}{\frac{n}{2} - j + 1},
\end{equation}
we find
\begin{equation}
  m(n, j - 1) - m(n, j) = -4 \binom{n}{n / 2 - j} \frac{4 j^2 - (n + 2)}{4 j^2 - (n + 2)^2}.
\end{equation}
It is easy to see that the denominator of the right-hand side is always negative, since the maximum value is $4 j_{\max}^2 = n^2$. Hence $m(n, j - 1) > m(n, j)$ if and only if $4 j^2 > n + 2$.

In the following, we give an estimation of $j_0= \frac{n}{2} - i_{\max}(n)$. Since the peak of $m(n, j)$ for fixed $n$ is around $j \sim \frac{\sqrt{n}}{2}$, we know that $j_0$ scales at least as $n^{1/2}$, i.e., $j_0 = \Omega(n^{1/2})$. We now assume $j_0 = o(n^{2/3})$, perform the estimation, and check the consistency. This assumption allows us to use the approximation of the binomial coefficients as a Gaussian function \cite{spencerasymptopia},
\begin{align}
  \binom{n}{n/2-j} \sim \binom{n}{n/2} \e^{-\frac{2j^2}{n}}. 
\end{align}
Then we have
\begin{align}
  \frac{m(n, j)}{m(n, j_{\min})} \sim \e^{-\frac{2(j^2-j_{\min}^2)}{n}} \frac{2j + 1}{2j_{\min} + 1} \frac{(n/2 + j_{\min} + 1)}{n/2 + j + 1} < \e^{-\frac{2j^2}{n}} (2j + 1) \quad \text{when } j = o(n^{2/3}).
\end{align}
Using this relation, in the following, we determine a range of $j$ for which the ratio $\frac{m(n, j)}{m(n, j_{\min})}$ becomes less than 1. Defining  $\alpha=j/\sqrt{n\ln n}$, we find that the right-hand side of the above equation can be rewritten as 
\begin{align}
  \e^{-\frac{2j^2}{n}} (2j + 1) = \frac{1}{n^{2\alpha^2}}(2\alpha\sqrt{n\ln n} + 1). %\leq 1 \quad \text{for } n = 1, 2, \cdots.
\end{align}

Therefore, we see that as long as $j \geq \alpha\sqrt{n\ln n}$ for some $\alpha > \frac{1}{2}$, we have $\frac{m(n, j)}{m(n, j_{\min})} < 1$ for sufficiently large $n$. This is consistent with our assumption. In practice, we can choose $\alpha = 1$, and it can be checked that we indeed have
\begin{align}
  \frac{m(n, \sqrt{n\ln n})}{m(n, j_{\min})} < 1 \quad \text{for } n = 2, 3, \cdots .
\end{align}
This implies $j_0 \leq \sqrt{n\ln n}$. One can also check that $i_{\max} \geq \lceil\frac{n}{2}- \sqrt{n\ln n} \rceil$ for all $n$. In \cref{tab:imax}, we list the true value of $i_{\max}$ and $\lceil\frac{n}{2}- \sqrt{n\ln n} \rceil$ for $n\leq 20$.

\begin{table}[htb]
  \centering
  \begin{tabular}{*{19}{c}}
    \toprule
    $n$ & 3 & 4 & 5 & 6 & 7 & 8 & 9 & 10 & 11 & 12 & 13 & 14 & 15 & 16 & 17 & 18 & 19 & 20 \\
    \midrule
    $i_{\max}$ & 0 & 0 & 1 & 1 & 2 & 1 & 2 & 2 & 3 & 2 & 4 & 3 & 4 & 4 & 5 & 4 & 6 & 5 \\
    $\lceil\frac{n}{2} - \sqrt{n\ln n} \rceil$ & 0 & 0 & 0 & 0 & 0 & 0 & 1 & 1 & 1 & 1 & 1 & 1 & 2 & 2 & 2 & 2 & 3 & 3 \\
    \bottomrule
  \end{tabular}
  \caption{$i_{\max}$ for each $n =3, 4, \cdots, 20$. $i_{\max}$ is the maximal $i$ such that $m_i \leq m_r$ for all $r>i$. In particular, we see that $i_{\max}\geq 3$ when $n\geq 13$, and $i_{\max}\geq 4$ when $n\geq 15$, which is used in the previous section. For comparison, we also list the value of $\lceil\frac{n}{2} - \sqrt{n\ln n} \rceil$.}
  \label{tab:imax}
\end{table}

\subsection{Calculation of the \texorpdfstring{$\mathbf{S}$}{S} matrix}\label{app:S-SU2}

In this subsection, we calculate the matrix element of $\mathbf{S}$, namely $s_{j'}(j)$. By definition in \cref{eq:S-def}, we have
\begin{align}
  s_{j'}(j) &= \int_{\SU(2)} \d g \Tr(u(g))^{n-k} f_{j'}(g) f_j^*(g).
\end{align}
Since the integrand is a class function, we can use the Weyl integral formula to simplify it to an integration over the maximal torus parametrized by $\theta \in [0, 2\pi)$,
\begin{align}
  s_{j'}(j) &= \frac{1}{\pi} \int_0^{2\pi} \d \theta \sin^2\theta \Tr(u(\theta))^{n-k} f_{j'}(\theta) f_j^*(\theta) .
\end{align}
The character of the maximal torus in the representation with angular momentum $j$ is
\begin{align}
  f_j(\theta) = \sum_{m=-j}^j \e^{\i 2m\theta} = \e^{-\i 2j\theta} \sum_{m=0}^{2j} \e^{\i 2m\theta} = \e^{-\i 2j\theta} \frac{1-\e^{\i 2(2j+1)\theta} }{1-\e^{\i 2\theta}} = \frac{\e^{-\i (2j+1)\theta}-\e^{\i (2j+1)\theta} }{\e^{-\i \theta} -\e^{\i \theta}} = \frac{\sin((2j+1)\theta)}{\sin\theta}.
\end{align}
Therefore we have
\begin{align}
  s_{j'}(j) &= \frac{1}{\pi} \int_0^{2\pi} \d \theta (2\cos\theta)^{n-k} \sin((2j'+1)\theta) \sin((2j+1)\theta) \nonumber\\
  &= -\frac{1}{4\pi} \int_0^{2\pi} \d \theta \sum_{l=0}^{n-k} \binom{n-k}{l} \e^{\i(2l-n+k)\theta} (\e^{\i(2j'+1)\theta} - \e^{-\i(2j'+1)\theta}) (\e^{\i(2j+1)\theta} - \e^{-\i(2j+1)\theta}) \nonumber\\
  &= -\frac{1}{2} \sum_{l=0}^{n-k} \binom{n-k}{l} (\delta_{n-k-2l, 2(j'+j)+2} - \delta_{n-k-2l, 2(j'-j)} - \delta_{n-k-2l, 2(j-j')} + \delta_{n-k-2l, -2(j'+j)-2}) \nonumber\\
  &= -\frac{1}{2} \Bigg(\binom{n-k}{\frac{n-k}{2}-j'-j-1} - \binom{n-k}{\frac{n-k}{2}-j'+j} -\binom{n-k}{\frac{n-k}{2}+j'-j} + \binom{n-k}{\frac{n-k}{2}+j'+j+1} \Bigg) \nonumber\\
  &= \binom{n-k}{\frac{n-k}{2}+j'-j} - \binom{n-k}{\frac{n-k}{2}+j'+j+1}.
\end{align}

As a sanity check, here we also show how $s_{j'}(j)$ can be calculated using \cref{eq:s-multiplicity}. Recall that from the rule of adding angular momenta, $m_{j' j}^l = 1$ when $|j-j'| \leq l \leq j+j'$ and zero otherwise. Therefore we have
\begin{align}
    s_{j'}(j) &= \sum_l m(n-k, l) \times  m_{j' j}^l = \sum_{l=|j-j'|}^{j+j'} m(n-k,l) \nonumber\\
    &= \sum_{l=|j-j'|}^{j+j'} \binom{n-k}{\frac{n-k}{2}-l} - \binom{n-k}{\frac{n-k}{2}-l-1} \nonumber\\
    &= \binom{n-k}{\frac{n-k}{2}+j'-j} - \binom{n-k}{\frac{n-k}{2}-j'-j-1},
\end{align}
which is consistent with the above calculation using integration over characters.

Alternatively, in terms of label $i=\frac{n}{2}-j$ and $i' = \frac{k}{2}-j'$, we can write \begin{align}
  s_{\frac{k}{2}-i'}(\frac{n}{2}-i) &= \binom{n-k}{i-i'} - \binom{n-k}{n-i-i'+1}.
\end{align}

\newpage

\section{\texorpdfstring{$\mathbb{Z}_p$}{Zp} symmetry}\label{app:Zp}

In this section, we consider the $\mathbb{Z}_p$ symmetry on $n$ qubits, and in particular prove \cref{eq:rnkS}, i.e. $\rk \mathbf{S} = p - 1$ when $p$ is even and $k \geq p$. Combined with \cref{thm:tdesign} and the fact that there are exacly $p$ inequivalent irreps of $\mathbb{Z}_p$, this implies the remaining claims of \cref{sec:Zp}.

Consider a Hilbert space spanned by $\ket{\alpha} : \alpha = 0, \dots, p - 1$. Define the operator $Z$ on this space by $Z \ket{\alpha} = \omega^\alpha \ket{\alpha}$ where $\omega = \e^{\i 2 \pi / p}$. Then $Z^p = \ident$ is the identity on this space. Let $\ket{\psi} = \sum_\alpha \ket{\alpha}$ be the (unnormalized) uniform superposition. It is easy to see that, for all $a = 1, \dots, p - 1$, $\qavg{Z^a}{\psi} = 0$. Hence, the vectors $\ket{\psi_a} = Z^a \ket{\psi}$ form an orthogonal basis for the space spanned by $\set{\ket{\alpha}: \alpha = 0, \dots, p - 1}$.

Then, when $p \leq k < n$, the $\mathbf{S}$ matrix has components
\begin{equation}
    s_\alpha(\beta) = \sum_{a = 0}^{p - 1} (1 - \omega^a)^{n - k} \qamp{\alpha}{Z^a}{\psi} \qamp{\psi}{Z^{\dagger a}}{\beta} = \sum_{a = 0}^{p - 1} (1 - \omega^a)^{n - k} \qinn{\alpha}{\psi_a} \qinn{\psi_a}{\beta}.
\end{equation}
The rank of $\mathbf{S}$ is clearly the rank of the operator $\sum_{a} (1 - \omega^a)^{n - k} \qproj{\psi_a}$, which is diagonal in the $\set{\ket{\psi_a}}$ basis. Furthermore, its eigenvalues are $(1 - \omega^a)^{n - k} : a = 0, \dots, p - 1$, which is zero if and only if $p$ is even and $a = p / 2$. Thus,
\begin{equation}
    \rk \mathbf{S} = \begin{cases} p - 1 & p \text{ even} \\
        p & p \text{ odd}. \end{cases}
\end{equation}

This, in particular, means that for even $p$,
\begin{equation}
   \dim(\operatorname{span}_\real \set{r^{n - p} f_\alpha : \alpha = 0, \cdots, p - 1}) = p - 1 .
\end{equation}

\newpage

\section{\texorpdfstring{$\SU(d)$}{SU(d)} symmetry}\label{app:SUd}

In this appendix, we determine the value of $t_{\max}$ for the group generated by $k$-local $\SU(d)$-invariant unitaries for $d \geq 3$ and $k = 4$. This is different for $d = 3$ and $d \geq 4$ since in the former case the charge vector of $\P_{(1234)}$ is not linearly independent of the permutations $\ident$, $\P_{12}$, $\P_{(123)}$, and $\P_{(12)(34)}$. (In our case, we only need that the rank of the character table restricted to the lowest dimensional irreps does not change whether or not we include $\P_{(1234)}$ when $d = 3$.) For the calculation in this appendix, the following fact will be crucial,

\begin{fact}
The character table for $\S_n$ on the seven lowest dimensional irreps when $n\geq 15$ is given by \cref{tab:character-table}.
\begin{table*}[htb]
  \centering
  \caption{Character table for $\S_n$: first rows in the Young diagrams are omitted for simplicity}
  \begin{tabular}{*{6}{c}}
    \toprule
    Irreps & (1) & (12) & (123) & (12)(34) & (1234) \\
    \midrule
    $\cdot$ & $1$ & $1$ & $1$ & $1$ & $1$ \\
    \midrule
    $\ydiag{1}$ & $n-1$ & $n-3$ & $n-4$ & $n-5$ & $n-5$ \\
    \midrule
    $\ydiag{2}$ & $\frac{n(n-3)}{2}$ & $\frac{(n-3)(n-4)}{2}$ & $\frac{(n-3)(n-6)}{2}$ & $\frac{n^2-11n+32}{2}$ & $\frac{(n-4)(n-7)}{2}$ \\
    \midrule
    $\ydiag{1,1}$ & $\frac{(n-1)(n-2)}{2}$ & $\frac{(n-2)(n-5)}{2}$ & $\frac{(n-4)(n-5)}{2}$ & $\frac{n^2-11n+26}{2}$ & $\frac{(n-5)(n-6)}{2}$ \\
    \midrule
    $\ydiag{3}$ & $\frac{n(n-1)(n-5)}{6}$ & $\frac{(n-3)(n-4)(n-5)}{6}$ & $\frac{(n-5)(n^2-10n+18)}{6}$ & $\frac{(n-5)(n^2-13n+48)}{6}$ & $\frac{(n-4)(n-5)(n-9)}{6}$ \\
    \midrule
    $\ydiag{1,1,1}$ & $\frac{(n-1)(n-2)(n-3)}{6}$ & $\frac{(n-2)(n-3)(n-7)}{6}$ & $\frac{(n-3)(n^2-12n+38)}{6}$ & $\frac{(n-3)(n-5)(n-10)}{6}$ & $\frac{(n-5)(n-6)(n-7)}{6}$ \\
    \midrule
    $\ydiag{2,1}$ & $\frac{n(n-2)(n-4)}{3}$ & $\frac{(n-2)(n-4)(n-6)}{3}$ & $\frac{(n-4)(n^2-11n+27)}{3}$ & $\frac{(n-4)(n-6)(n-8)}{3}$ & $\frac{(n-4)(n-6)(n-8)}{3}$ \\
    \bottomrule
  \end{tabular}
  \label{tab:character-table}
\end{table*}
    
\end{fact}

In \cref{app:SU3}, we discuss the case $d=3$, where there are only three-row Young diagrams, and $\P_{(12)(34)}$ is included but not $\P_{(1234)}$. Then in \cref{app:P(12)(34)}, we still include $\P_{(12)(34)}$ and exclude $\P_{(1234)}$, but for general $d\geq 4$. Finally, in \cref{app:k=4}, we consider the case $d\geq 4$ for all $4$-local Hamiltonians, including both $\P_{(12)(34)}$ and $\P_{(1234)}$.

\subsection{\texorpdfstring{$d=3, k = 4$}{d=3,k=4}}\label{app:SU3}

The first four rows \cref{tab:character-table} do not contain Young diagrams with more than three rows. By checking the rank of the upper left $4\times 4$ matrix of the character table in \cref{tab:character-table}, we find its rank is 4. Therefore $t_{\max}+1 \geq \frac{1}{6}n(n-1)(n-5)$. In order to find a tight upper bound for $t_{\max}$, we have to include the next two irreps (excluding the one with four rows in \cref{tab:character-table}), the dimensions of which all scale as $n^3$. Then we can solve the $\mathbf{q}^\T$ vector, which is in the span of the rows of
\begin{align}
    \mathbf{Q} = 
    \begin{pmatrix}
        -\binom{n-3}{3} & \binom{n-4}{2} & -\binom{n-5}{1} & 0 & 1 & 0 & 0 \\
        -2\binom{n-2}{3} & 2\binom{n-3}{2} & -\binom{n-4}{1} & -\binom{n-4}{1} & 0 & 0 & 1 \\
    \end{pmatrix},
\end{align}
or in other words
\begin{align}
    \mathbf{q}^\T = 
    \begin{pmatrix}
        a_1 & a_2 \\
    \end{pmatrix}
    \mathbf{Q}
\end{align}
Multiplying the $i$-th column with $m_i$, we obtain a matrix whose all non-zero terms are of degree $3$. When $n$ is sufficiently large, we can keep only the cubic terms in $n$, whose coefficient matrix is given by
\begin{align}
    \frac{1}{6}
    \begin{pmatrix}
        -1 & 3 & -3 & 0 & 1 & 0 & 0 \\
        -2 & 6 & -3 & -3 & 0 & 0 & 2 \\
    \end{pmatrix}.
\end{align}
From the structure of this matrix, we can see that $(a_1, a_2) = (1,0)$ gives a $t_{\max}$ that has the smallest coefficient in the leading term, namely $\frac{2}{3}n^3$, as long as $n$ is not too small. More precisely, we have
\begin{align}
    \mathbf{q} = 
    \begin{pmatrix}
        1 & 0 \\
    \end{pmatrix}
    \mathbf{Q},
\end{align}
and
\begin{align}
    t_{\max} + 1 \leq \frac{2}{3}(n-1)(n-3)(n-5).
\end{align}
Including other irreps leads to quartic scaling in $n$, and when $n \geq 22$, even the smallest dimension of them is greater than $\frac{2}{3}(n-1)(n-3)(n-5)$. Therefore we know that when $n \geq 22$,
\begin{align}
    t_{\max} + 1 = \frac{2}{3}(n-1)(n-3)(n-5).
\end{align}

\subsection{\texorpdfstring{$d\geq 4, k = 4$}{d>=4,k=3} (with \texorpdfstring{$\P_{(12)(34)}$}{P(12)(34)} but no \texorpdfstring{$\P_{(1234)}$}{P(1234)})}\label{app:P(12)(34)}

The rank of the upper left $4\times 4$ matrix of the character table \cref{tab:character-table} is the same as the previous case, i.e., 4. Therefore $t_{\max}+1 \geq \frac{1}{6}n(n-1)(n-5)$. In order to find a tight upper bound for $t_{\max}$, we have to include the next three irreps, the dimensions of which all scale as $n^3$. Then we can solve the $\mathbf{q}^\T$ vector, which is in the span of the rows of
\begin{align}
    \mathbf{Q} = 
    \begin{pmatrix}
        -\binom{n-3}{3} & \binom{n-4}{2} & -\binom{n-5}{1} & 0 & 1 & 0 & 0 \\
        -\binom{n-1}{3} & \binom{n-2}{2} & 0 & -\binom{n-3}{1} & 0 & 1 & 0 \\
        -2\binom{n-2}{3} & 2\binom{n-3}{2} & -\binom{n-4}{1} & -\binom{n-4}{1} & 0 & 0 & 1 \\
    \end{pmatrix},
\end{align}
or in other words
\begin{align}
    \mathbf{q}^\T = 
    \begin{pmatrix}
        a_1 & a_2 & a_3 \\
    \end{pmatrix}
    \mathbf{Q}
\end{align}
Multiplying the $i$-th column with $m_i$, we obtain a matrix whose all non-zero terms are of degree $3$. When $n$ is sufficiently large, we can keep only the cubic terms in $n$, whose coefficient matrix is given by
\begin{align}
    \frac{1}{6}
    \begin{pmatrix}
        -1 & 3 & -3 & 0 & 1 & 0 & 0 \\
        -1 & 3 & 0 & -3 & 0 & 1 & 0 \\
        -2 & 6 & -3 & -3 & 0 & 0 & 2 \\
    \end{pmatrix}.
\end{align}
From the structure of this matrix, we can see that $(a_1, a_2, a_3) = (1,1,-1)$ gives a $t_{\max}$ that has the smallest coefficient in the leading term, namely $\frac{1}{3}n^3$, as long as $n$ is not too small. With this choice, we have
\begin{align}
    \mathbf{q} = 
    \begin{pmatrix}
        n-3 & 1 & 1 & -1 & 1 & 1 & -1 \\
    \end{pmatrix}^\T,
\end{align}
and
\begin{align}
    t_{\max} + 1 \leq \frac{1}{6}(n-3)(2n^2-3n+4).
\end{align}
Including other irreps leads to quartic scaling in $n$, and when $n \geq 22$, even the smallest dimension of them is greater than $\frac{1}{6}(n-3)(2n^2-3n+4)$. Therefore we know that when $n \geq 22$,
\begin{align}
    t_{\max} + 1 = \frac{1}{6}(n-3)(2n^2-3n+4).
\end{align}

\subsection{\texorpdfstring{$d\geq 4, k=4$}{d>=4,k=4} (with both \texorpdfstring{$\P_{(12)(34)}$}{P(12)(34)} and \texorpdfstring{$\P_{(1234)}$}{P(1234)})}\label{app:k=4}
When $d\geq 4$, we also have to include $\P_{(1234)}$, and now we have to check the rank of the upper left $5\times 5$ matrix of the character table \cref{tab:character-table}. However, since all irreps appearing here have at most $3$ rows, we know $\P_{(1234)}$ is not linearly independent of other permutations. Therefore without any calculation, we know that the lower bound of $t_{\max}$ is the same as the previous cases, i.e., $t_{\max} \geq \frac{1}{6}n(n-1)(n-5)-1$.

In order to find a tight upper bound for $t_{\max}$, we have to include the next two irreps whose dimensions all scale as $n^3$. Since $\P_{(1234)}$ will only be linearly independent when we include $[n-3,1,1,1]$, the only Young diagram that has $4$ or more rows in the previous case, we know the two solutions that do not have support on this Young diagram are still solutions in the present case. Therefore, interestingly, we find this case is the same as the case discussed in \cref{app:SU3}, i.e., when $n\geq 22$, 
\begin{align}
    t_{\max} + 1 = \frac{2}{3}(n-1)(n-3)(n-5).
\end{align}

\clearpage

\section{Failure of semi-universality (Proof of Proposition \ref{prop:con})}\label{App:fail}

In this section, we use a result of \cite{hulse2024framework} to prove \cref{prop:con} (as mentioned in the main text, this result also follows from the result of \cite{robert2015squares}). 

For any $G$-invariant unitary $V\in\mathcal{V}^G$ consider its decomposition as $V=\bigoplus_{\lambda\in\irreps_G} \mathbb{I}_{\mathcal{Q}_\lambda}\otimes v_\lambda$. For any irrep $\lambda\in\irreps_G$, let $\pi_{\lambda}(V)=v_{\lambda}$,
which defines a homomorphism from $\mathcal{V}^G$ to $\SU(\mathcal{M}_{\lambda})$.

\begin{lemma}[\cite{hulse2024framework}]\label{lem:formal}
  Suppose $\mathcal{W}^G\subseteq \mathcal{V}^G$ is not semi-universal, i.e., it does not contain $\mathcal{SV}^G=[\mathcal{V}^G,\mathcal{V}^G]$. Then, at least, one of the following holds:
  \begin{enumerate}
  \item There exists an irrep $\lambda\in\irreps_G$ such that $\pi_{\lambda}(\mathcal{W}^G) $ does not contain $\SU(\mathcal{M}_{\lambda})$.
  \item There exist two distinct irreps $\lambda_1,\lambda_2\in \irreps_G$ such that 
    $\dim(\mathcal{M}_{\lambda_1})=\dim(\mathcal{M}_{\lambda_2})$, and there exists an isometry $S: \mathcal{M}_{\lambda_1}\rightarrow \mathcal{M}_{\lambda_2}$ such that one of the following holds 
    \begin{itemize}
    \item For all $V\in \mathcal{V}^G: \pi_{\lambda_2}(V)=\e^{\i\theta} S\pi_{\lambda_1}(V)S^\dag$, or 
    \item 
      For all $V\in \mathcal{V}^G: \pi_{\lambda_2}(V)=\e^{\i\theta} S\pi_{\lambda_1}(V)^\ast S^\dag$,
    \end{itemize}
    where $\e^{\i\theta}$ is an unspecified phase that depends on $V$.
  \end{enumerate}
\end{lemma}

Case 1 corresponds to restrictions of type $\mathbf{II}$ and $\mathbf{III}$. In this case, it turns out that if $\mathcal{W}^G$ is a connected compact group then 
\be
\mathrm{Comm}\{\pi_{\lambda_1}(V)\otimes \pi_{\lambda_1}(V): V\in \mathcal{W}^G\}\neq \mathrm{Comm}\{\pi_{\lambda_1}(V)\otimes \pi_{\lambda_1}(V): V\in \mathcal{V}^G\}=\mathrm{Comm}\{U\otimes U: U\in \SU(\mathcal{M}_\lambda)\} \ .
\ee
In particular, the right-hand side is a 2-dimensional space, whereas the left-hand side has a larger dimension \cite{dynkin1957maximal, robert2015squares, zimboras2015symmetry}. This immediately implies that the uniform distribution over $\mathcal{W}^G$ is not a 2-design for the uniform distribution over $\mathcal{V}^G$.

Now suppose case 1 does not hold, and therefore by this lemma case 2 holds. This corresponds to type $\mathbf{IV}$ constraints. In this case
\begin{align}\label{hf}
  \mathrm{Comm}\{\pi_{\lambda_1}(V)\otimes \pi_{\lambda_2}(V): V\in \mathcal{W}^G\}&\neq \mathrm{Comm}\{\pi_{\lambda_1}(V)\otimes \pi_{\lambda_2}(V): V\in \mathcal{V}^G\} \nonumber\\ 
  &=\mathrm{Comm}\{U_1\otimes U_2: U_1\in \SU(\mathcal{M}_{\lambda_1}), U_2\in \SU(\mathcal{M}_{\lambda_2})\} \nonumber\\ 
  &=\{c \mathbb{I}_{\mathcal{M}_{\lambda_1}}\otimes \mathbb{I}_{\mathcal{M}_{\lambda_2}}: c\in\mathbb{C}\} \ .
\end{align}
In particular, while the dimension of the right-hand side of \cref{hf} is one, the dimension of the left-hand side is two, which follows from the facts that
\be
\dim(\mathrm{Comm}\{U\otimes U: U\in \SU(d)\})=\dim(\mathrm{Comm}\{U\otimes U^\ast: U\in \SU(d)\})=2\ .
\ee

\newpage

\section{Comparing distributions over a compact manifold and its closed submanifold}\label{app:design}

In this section, we prove a general result that the uniform distribution over a closed submanifold cannot be an $t$-design of uniform distribution of the compact ambient manifold that contains it for arbitrarily large $t$.

\lemDesign*

\begin{proof}
  Since $N$ is a compact manifold and $M \subsetneq N$ is a proper closed submanifold, the complement $N \setminus M$ is a nonempty open subset of $N$. Since the normalized measure $\mu_N$ is strictly positive, which means every open set has a positive measure on $N$, we have $\mu_N(M) = c$ and $\mu_N(N \setminus M) = 1-c$ with $c<1$. Furthermore, the measure $\mu_N$ being (inner) regular means that
  \begin{align}\label{eq:inner-regular}
    \mu_N(N \setminus M) = \sup\{\mu_N(F): F\subseteq N \setminus M \text{ and } F \text{ is compact and measurable}\}.
  \end{align}
  Therefore, there exists a nonempty measurable closed subset $N' \subset N \setminus M$, such that $\mu_N(N') = c' > 0$, where $0<c'\leq 1-c$. (Otherwise, if $c'=0$ for any such $N'$, then $\mu_N(N \setminus M) = 0$ from \cref{eq:inner-regular}.)

  Because $M$ and $N'$ are disjoint closed subsets of $N$, Urysohn's Lemma \cite{munkres2000topology} guarantees the existence of a continuous function $g: N \rightarrow [0,1]$ such that
  \begin{align}
    g|_M \equiv 0 \quad \text{and} \quad g|_{N'} \equiv 1.
  \end{align}

  Consider the embedding $f: N \hookrightarrow \mathbb{R}^n$. Let $\mathcal{P}$ denote the set of polynomial functions on $\mathbb{R}^n$, and define the algebra $A$ of functions on $N$ by
  \begin{align}
    A = \{ p \circ f \mid p \in \mathcal{P} \}.
  \end{align}
  Since $f$ is continuous, $A$ is a subalgebra of $C(N,\mathbb{R})$, the space of continuous real-valued functions on $N$. Moreover, $A$ satisfies the following two properties,
  \begin{enumerate}
    \item \textbf{Separates points:} For any two distinct points $x, y \in N$, since $f$ is injective, $f(x) \neq f(y)$. The polynomials on $\mathbb{R}^n$ separate points in $\mathbb{R}^n$, so there exists $p \in \mathcal{P}$ such that $p(f(x)) \neq p(f(y))$. Thus, $A$ separates points on $N$.

    \item \textbf{Contains constants:} The constant functions are included in $A$ (e.g., $p(x) = c$ for some $c \in \mathbb{R}$).
  \end{enumerate}

  By the Stone-Weierstrass theorem \cite{rudin1976principles}, $A$ is dense in $C(N,\mathbb{R})$ with respect to the supremum norm. Therefore, for any $\varepsilon > 0$, there exists a polynomial $p \in \mathcal{P}$ such that
  \begin{align}
    \| g - p \circ f \|_\infty = \sup_{x \in N} | g(x) - p(f(x)) | < \varepsilon.
  \end{align}
  
  Since $g|_M \equiv 0$, it follows that for all $x \in M$,
  \begin{align}
    | p(f(x)) | = | p(f(x)) - g(x) | < \varepsilon.
  \end{align}
  Thus, using the fact that $\mu_M$ is normalized, we have
  \begin{align}
    \Big| \int_M \d\mu_M\, p \circ f \Big| \leq \int_M \d\mu_M\, | p \circ f | < \varepsilon.
  \end{align}

  On the other hand, using $g(x) - p(f(x)) \leq |g(x) - p(f(x))| < \varepsilon$, we have
  \begin{align}
    \int_N \d\mu_N\, p \circ f \geq \int_N \d\mu_N\, g - \int_N \d\mu_N\, |g(x) - p(f(x))| > \int_{N'} \d\mu_N\, g - \int_N \d\mu_N\, \varepsilon = c' - \varepsilon,
  \end{align}
  where in the last step we use $g|_{N'} \equiv 1$.
  
  Now, compute the difference between the integrals:
  \begin{align}
    \int_N \d\mu_N\, p \circ f  - \int_M \d\mu_M\, p \circ f > c' - 2\varepsilon.
  \end{align}
  Choose $\varepsilon > 0$ small enough so that $\varepsilon < \frac{c'}{2}$. Then
  \begin{align}
    \int_N \d\mu_N\, p \circ f \neq \int_M \d\mu_M\, p \circ f .
  \end{align}
  This completes the proof.
\end{proof}

\end{document}